\documentclass[a4paper,10pt,reqno]{amsart}
          \usepackage{amssymb}
	  \usepackage{amsmath}
          \usepackage{amsfonts}
          \usepackage[english]{babel}
          \usepackage[utf8]{inputenc}

\usepackage[margin=2.5cm]{geometry}

 \usepackage[unicode,colorlinks,plainpages=false,hyperindex=true,bookmarksnumbered=true,bookmarksopen=false,pdfpagelabels]{hyperref}
 \hypersetup{urlcolor=cyan,linkcolor=blue,citecolor=red,colorlinks=true} 
 \usepackage{tikz} 
 \usetikzlibrary{arrows,shapes}  

\makeatletter
\renewcommand*{\p@section}{\,}
\renewcommand*{\p@subsection}{\S\,}
\renewcommand*{\p@subsubsection}{\S\,}
\makeatother

\newtheorem{thm}{Theorem}[section]
\newtheorem{cor}[thm]{Corollary}
\newtheorem{lem}[thm]{Lemma}
\newtheorem{prop}[thm]{Proposition}
\newtheorem{exmp}[thm]{Example}
\newtheorem{rem}[thm]{Remark}
\newtheorem{defn}[thm]{Definition}

\numberwithin{equation}{section}

\newcommand{\CC}{\ensuremath{\mathbb{C}}}
\newcommand{\N}{\ensuremath{\mathbb{N}}}

\newcommand{\Z}{\ensuremath{\mathbb{Z}}}
\newcommand{\ad}{\operatorname{ad}}
 
\newcommand{\tr}{\operatorname{tr}}
\newcommand{\diag}{\operatorname{diag}}

\newcommand{\Mat}{\operatorname{Mat}}
\newcommand{\Id}{\operatorname{Id}}
\newcommand{\Gl}{\operatorname{GL}}

\newcommand{\gl}{\ensuremath{\mathfrak{gl}}}
\newcommand{\g}{\ensuremath{\mathfrak{g}}}
\newcommand{\OO}{\ensuremath{\mathcal{O}}}
\newcommand{\MM}{\ensuremath{\mathcal{M}}}
\newcommand{\NN}{\ensuremath{\mathcal{N}}}
\newcommand{\VV}{\ensuremath{\mathcal{V}}}

\newcommand{\Hom}{\operatorname{Hom}}
\newcommand{\End}{\operatorname{End}}
\newcommand{\cyc}{{\operatorname{cyc}}}
 \newcommand{\rL}{\ensuremath{\mathtt{L}}}
\newcommand{\HH}{\ensuremath{\mathcal{H}}}
\newcommand{\II}{\ensuremath{\mathcal{I}}}
 \newcommand{\HHl}{\ensuremath{\mathcal{H}_{\rL}}}
 \newcommand{\IIl}{\ensuremath{\mathcal{I}_{\rL}}}
\newcommand{\Ord}{{\operatorname{Ord}}} 
\newcommand{\res}{{\operatorname{res}}} 
\newcommand{\LL}{\ensuremath{\mathcal{L}}}
 \newcommand{\LLl}{\ensuremath{\mathcal{L}_{\rL}}}
\newcommand{\As}{\ensuremath{\mathbf{A}}}
\newcommand{\Cs}{\ensuremath{\mathbf{C}}}
\newcommand{\Xs}{\ensuremath{\mathbf{X}}}
\newcommand{\nfat}{\ensuremath{\mathbf{n}}}
\newcommand{\qfat}{\ensuremath{\mathbf{q}}}
\newcommand{\tfat}{\ensuremath{\mathbf{t}}}
\newcommand{\dfat}{\ensuremath{\mathbf{d}}}
\newcommand{\kfat}{\ensuremath{\mathbf{k}}}

\newcommand{\Cnm}{\ensuremath{\mathcal{C}_{n,m}}}
\newcommand{\Cnqm}{\ensuremath{\mathcal{C}_{n,m,\qfat,\dfat}}}
\newcommand{\h}{\ensuremath{\mathfrak{h}}}
\newcommand{\hreg}{\ensuremath{\mathfrak{h}_{\operatorname{reg}}}}
\newcommand{\aaa}{\ensuremath{\mathbf{a}}}
\newcommand{\ccc}{\ensuremath{\mathbf{c}}}
\newcommand{\gfat}{\ensuremath{\mathbf{g}}}
\newcommand{\ffat}{\ensuremath{\mathbf{f}}}
\newcommand{\loc}{{\operatorname{loc}}} 
\newcommand{\hloc}{\ensuremath{\mathfrak{h}_{\loc}}}
\newcommand\dgal[1]{  \left\{\!\!\left\{#1\right\}\!\!\right\} }
\newcommand\br[1]{\{ #1 \}} 
\newcommand\brloc[1]{\{ #1 \}_{\loc}} 

\begin{document}

\title[Integrable systems on MQV]{Integrable systems on multiplicative\\quiver varieties from cyclic quivers}

\author{Maxime Fairon}
 \address[Maxime Fairon]{Institut de Mathématiques de Bourgogne, UMR~5584, CNRS \& Université Bourgogne Europe,
21000~Dijon, France}
 \email{maxime.fairon{\color{white}.\!\!}@{\color{white}.\!\!}u{\color{white}.\!\!}-{\color{white}.\!\!}bourgogne.fr}

  \begin{abstract}
We consider a class of complex manifolds constructed as multiplicative quiver varieties associated with a cyclic quiver extended by an arbitrary number of arrows starting at a new vertex. 
Such varieties admit a Poisson structure, which is obtained by quasi-Hamiltonian reduction.   
We construct several families of Poisson subalgebras inside the coordinate ring of these spaces, which we use to obtain degenerately integrable systems. 
We also extend the Poisson centre of these algebras to maximal abelian Poisson algebras, hence defining Liouville integrable systems. 
By considering a suitable set of local coordinates on the multiplicative quiver varieties, we can derive the local Poisson structure explicitly. This allows us to interpret the integrable systems that we have constructed as new generalisations of the spin Ruijsenaars-Schneider system with several types of spin variables. 
  \end{abstract}

\maketitle

 \setcounter{tocdepth}{2}

\tableofcontents


\section{Introduction}

While integrable systems are extremely rare among finite-dimensional dynamical systems, there has been a plethora of examples that have appeared over the past 50 years and which are named\footnote{The list of names that we present is biased towards many-body systems, and is certainly not exhaustive.} after Toda \cite{To}; Calogero \cite{Cal}, Moser \cite{Mo} and Sutherland \cite{Su}; or Ruijsenaars, Schneider \cite{RS} and van Diejen \cite{vD}.
These related integrable systems, which we will only consider in the classical case, can all be seen as toy models for many-body systems. Interestingly, they continue to attract a lot of attention from the points of view of mathematics and theoretical physics. 
This is especially true when such classical systems are endowed with additional \emph{spin} degrees of freedom. Indeed, these types of classical models have been widely studied over the past few years in order to understand their geometric construction in the real \cite{FFM,Fe1,Fe5,KLOZ,Re3} or complex \cite{AR,AO,CF2,F1,FG,Re3} setting, their bi-Hamiltonian structure \cite{Fe2,Fe3,Fe4}, their spectral description \cite{GG,P}, or their relation to spinning tops \cite{GSZ,SZ,TZ,Z} and integrable hierarchies \cite{PZ1,PZ2}.
Our main aim consists in introducing new families of integrable systems that can also be interpreted as many-body systems with spin variables. To work towards this goal, our approach uses the representation theory of quivers to build their phase spaces as multiplicative quiver varieties \cite{CBS}. 

The motivation for the present investigation stems from the work of Wilson \cite{W}, who noticed that the completed phase space for the complex (i.e. holomorphic) rational Calogero-Moser (CM) system is a quiver variety \cite{Na}, attached to the quiver made of one loop and one extra arrow.
More recently, it was observed that adding several extra arrows to the loop \cite{BP,CS,Ta,W2} yields the phase space of the rational \emph{spin} CM system introduced by Gibbons and Hermsen \cite{GH}.
Furthermore, this strategy has been adapted to unveil new rational CM systems from cyclic quivers \cite{CS,FG,Si}, see Figure \ref{Fig:Tab1}. 
(For other examples of classical real/complex integrable systems on quiver varieties, the reader can consult \cite{F3,FiR,GR,Ne,RaS,We}.)
Building on these works, it seems natural to find analogous constructions for multiplicative quiver varieties attached to the same quivers, as these varieties can be endowed with a Poisson structure \cite{VdB1,VdB2,Y} obtained by quasi-Hamiltonian reduction \cite{AKSM,AMM}. 
Motivated by the article of Fock and Rosly \cite{FockRosly}, works of Chalykh and the author have already appeared in that direction \cite{CF1,CF2,F1} (see Figure \ref{Fig:Tab1} for the corresponding quivers) but the study of multiplicative quiver varieties associated with a cyclic quiver extended by an arbitrary number of arrows starting at a new vertex was not yet available. This manuscript fills this gap, and brings to light new families of many-body systems where the particles are endowed with different types of degrees of freedom colloquially called \emph{spins}.

Before delving into an overview of the present paper, let us mention two important earlier results that were obtained using the present approach regarding complex integrable systems. 
First, it was shown in \cite{CF1} how to obtain three  integrable systems on the multiplicative quiver varieties attached to a cyclic quiver extended by exactly one arrow. One of these systems was identified as the celebrated trigonometric Ruijsenaars-Schneider (RS) system \cite{RS} in local coordinates, and the other two were described as ``cyclic'' extensions of it. It was explained in \cite{BEF,CF1} how they could be seen as classical versions of quantum integrable systems obtained from cyclotomic Double Affine Hecke Algebras \cite{BEF}, quantised Coulomb branches \cite{KN},  or twisted Macdonald-Ruijsenaars systems \cite{CE}. 
Second, in the case of an extension by $d_0\geq2$ arrows to exactly one vertex of the cyclic quiver, it was observed in \cite{CF2,F1} how to obtain the spin version of the trigonometric RS system introduced by Krichever and Zabrodin \cite{KrZ}. In particular, this settled a conjecture of Arutyunov and Frolov \cite{AF} regarding the local Poisson structure underlying that system. Since we work with more general extensions of a cyclic quiver in this paper, we can recover these results by working on suitable closed subvarieties. An explicit description of this connection to the spin RS system will be given in \ref{sss:SpinRS}.

\begin{figure}
 \centering 
   \begin{tikzpicture}[scale=1.2]
\draw[thick] (-1.5,7) -- (10.2,7);
\draw[thick] (1.5,8) -- (1.5,-1.2);
\draw[thick] (-1.5,8) -- (1.5,7);
\node (n1) at (-0.7,7.3) {Quivers};
\node (n2) at (0.7,7.8) {Integrable};
\node (n3) at (0.7,7.55) {Systems};
\node (IS-1) at (3.5,7.7) {rational CM systems};
\node[font=\small] (IS-2) at (3.5,7.4) {(on quiver varieties)};
\draw[thick] (5.7,8) -- (5.7,-1.2);
\node (IS-3) at (8,7.7) {trigonometric RS systems};
\node[font=\small] (IS-4) at (8,7.4) {(on multiplicative quiver varieties)};
\node[circle,fill=black,inner sep=1pt] (sInf) at (-1,6.5) {};
 \node[circle,fill=black,inner sep=1pt] (sZ) at (0.5,6.5) {};
 \draw [->,thick,>=stealth] (sInf) -- (sZ); 
\draw[->,thick,>=stealth] (sZ) to[out=70,in=290,looseness=30] (sZ);
\node[label=right:{Wilson \cite{W}}] (s-QV) at (1.6,6.5) {};
\node[label=right:{Fock-Rosly \cite{FockRosly}}] (s-MQV-FR) at (5.8,6.65) {};
\node[label=right:{Chalykh-Fairon \cite{CF1}}] (s-MQV) at (5.8,6.25) {};
\draw[dashed] (-1.5,5.8) -- (10.2,5.8);
\node[circle,fill=black,inner sep=1pt] (tInf) at (-1,4.8) {};
 \node[circle,fill=black,inner sep=1pt] (tZ) at (0.5,4.8) {};
 \draw [->,thick,>=stealth] (tInf)-- node[above]{$d\geq 2$}   (tZ); 
\draw[->,thick,>=stealth] (tZ) to[out=70,in=290,looseness=30] (tZ);
\node[label=right:{Bielawski-Pidstrygach$^\dagger$ \cite{BP}}] (t1-QV) at (1.6,5.4) {};
\node[label=right:{Chalykh-Silantyev \cite{CS}}] (t2-QV) at (1.6,5.0) {};
\node[label=right:{Tacchella \cite{Ta}}] (t3-QV) at (1.6,4.6) {};
\node[label=right:{Wilson \cite{W2}}] (t4-QV) at (1.6,4.2) {};
\node[label=right:{Chalykh-Fairon \cite{CF2}}] (t-MQV) at (5.8,4.8) {};
\draw[dashed] (-1.5,3.85) -- (10.2,3.85);
  \node (dotU) at (0,3.45) {$\ldots$};
  \node[circle,fill=black,inner sep=2pt] (midCycInf) at (0,2.5) {};
 \node[circle,fill=black,inner sep=1pt] (midhL) at (-0.5,3.37) {};
  \node[circle,fill=black,inner sep=1pt] (midhR) at (0.5,3.37) {};
   \node[circle,fill=black,inner sep=1pt] (midmL) at (-1,2.5) {};
  \node[circle,fill=black,inner sep=1pt] (midmR) at (1,2.5) {};
     \node[circle,fill=black,inner sep=1pt] (midbL) at (-0.5,1.63) {};
  \node[circle,fill=black,inner sep=1pt] (midbR) at (0.5,1.63) {};
 \draw [->,thick,>=stealth] (midCycInf) --  (midmR); 
 \path[->,thick,>=stealth]  (midhR) edge [bend left=15] (midmR);
  \path[->,thick,>=stealth]  (midmR) edge [bend left=15] (midbR);
  \path[->,thick,>=stealth]  (midbR) edge [bend left=15] (midbL);
 \path[->,thick,>=stealth]  (midbL) edge [bend left=15] (midmL);
   \path[->,thick,>=stealth]  (midmL) edge [bend left=15] (midhL);
\node[label=right:{Chalykh-Silantyev \cite{CS}}] (mid1-QV) at (1.6,2.9) {};
\node[label=right:{Fairon-G\"orbe \cite{FG}}] (mid2-QV) at (1.6,2.5) {};
\node[label=right:{Silantyev \cite{Si}}] (mid3-QV) at (1.6,2.1) {};
\node[label=right:{Chalykh-Fairon \cite{CF1}}] (mid-MQV) at (5.8,2.5) {};
\draw[dashed] (-1.5,1.3) -- (10.2,1.3);
  \node (dotB) at (0,0.95) {$\ldots$};
  \node[circle,fill=black,inner sep=2pt] (bCycInf) at (0,0) {};
 \node[circle,fill=black,inner sep=1pt] (bhL) at (-0.5,0.87) {};
  \node[circle,fill=black,inner sep=1pt] (bhR) at (0.5,0.87) {};
   \node[circle,fill=black,inner sep=1pt] (bmL) at (-1,0) {};
  \node[circle,fill=black,inner sep=1pt] (bmR) at (1,0) {};
     \node[circle,fill=black,inner sep=1pt] (bbL) at (-0.5,-0.87) {};
  \node[circle,fill=black,inner sep=1pt] (bbR) at (0.5,-0.87) {};
 \path[->,thick,>=stealth]  (bhR) edge [bend left=15] (bmR);
  \path[->,thick,>=stealth]  (bmR) edge [bend left=15] (bbR);
  \path[->,thick,>=stealth]  (bbR) edge [bend left=15] (bbL);
 \path[->,thick,>=stealth]  (bbL) edge [bend left=15] (bmL);
   \path[->,thick,>=stealth]  (bmL) edge [bend left=15] (bhL); \node[font=\scriptsize] (dot0) at (0,0.55) {$\ldots$};
    \draw [->,thick,>=stealth] (bCycInf) --   (bmR);   \node[font=\scriptsize] (arr0) at (0.5,0.1) {${}^{d_0}$};
 \draw [->,thick,>=stealth] (bCycInf) -- (bhR); \node[font=\scriptsize] (arrminus1) at (0.5,0.4) {${}^{d_{m\!-\!1}}$};
 \draw [->,thick,>=stealth] (bCycInf) --  (bbR); \node[font=\scriptsize] (arr1) at (0.35,-0.4) {${}^{d_1}$};
 \draw [->,thick,>=stealth] (bCycInf) --  (bbL); \node[font=\scriptsize] (arr2) at (-0.35,-0.4) {${}^{d_2}$};
 \draw [->,thick,>=stealth] (bCycInf) --  (bmL);\node[font=\scriptsize] (arr3) at (-0.5,0.1) {${}^{d_3}$};
 \draw [->,thick,>=stealth] (bCycInf) --  (bhL);  \node[font=\scriptsize] (arr4) at (-0.41,0.43) {${}^{d_4}$};
 \node[label=right:{Chalykh-Silantyev$^{\ddagger}$ \cite{CS}}] (b1-QV) at (1.6,0.4) {};
\node[label=right:{Fairon-G\"orbe \cite{FG}}] (b2-QV) at (1.6,0) {};
\node[label=right:{Silantyev$^{\ddagger}$ \cite{Si}}] (b3-QV) at (1.6,-0.4) {};
\node[label=right:{Fairon$^\intercal$ \cite{F1}}] (b-MQV1) at (5.8,0.2) {};
\node[label=right:{This work (general case)}] (b-MQV2) at (5.8,-0.2) {};
  \end{tikzpicture}
 \caption{Summary of integrable systems of CM/RS type obtained on (multiplicative) quiver varieties.  
An integer $d$ or $d_s$ indicates the number of copies of an arrow to consider. 
(Additional comments: ${}^\dagger$only for $d=2$; ${}^\ddagger$only for $d_0\geq 2$ and $d_s=0$ when $s\neq 0$, or for $d_0=d_1=\ldots=d_{m-1}$; ${}^\intercal$only for $d_0\geq 2$ and $d_s=0$ when $s\neq 0$.)} 
\label{Fig:Tab1}
\end{figure}

\medskip

To state our results, let us sketch the construction of the varieties. We fix an integer $m\geq 2$. The cyclic quiver on $m$ vertices has vertex set $\Z_m=\Z/m\Z$ and arrows $x_s:s\to s+1$ for each $s\in \Z_m$. The extension of the quiver is encoded into a vector $\dfat\in \N^{\Z_m}$ with $|\dfat|=\sum_s d_s\geq 1$. Without loss of generality, we assume that $d_0\geq 1$. (If $d_s=0$ for all $s\neq 0$, we are in the case studied in \cite{F1}, or in \cite{CF1} if furthermore $d_0=1$.) 
We let $Q_\dfat$ be the quiver obtained from the cyclic quiver by adding one vertex, denoted $\infty$, and arrows $v_{s,\alpha}:\infty \to s$ with $1\leq \alpha\leq d_s$ for each $s\in \Z_m$, see the fourth quiver depicted in Figure \ref{Fig:Tab1}. There is no arrow $\infty \to s$ when $d_s=0$. 
Next, we fix $n\geq 1$ and consider a moduli space of representations of the double of $Q_\dfat$, which we parametrise by matrices 
\begin{equation*} 
 \begin{aligned}
X_s \in \Mat(n\times n,\CC),\,\,\, 
Y_s \in \Mat(n\times n,\CC), \quad
V_{s,\alpha}\in \Mat(1\times n,\CC), \,\,\,
W_{s,\alpha}\in \Mat(n\times 1,\CC),
 \end{aligned} 
\end{equation*}
where $s\in \Z_m$ and $\alpha\in \{1,\ldots,d_s\}$ in a couple of indices $(s,\alpha)$. 
Under some invertibility conditions on these matrices, we get an important space that we denote $\MM_{Q_\dfat,\nfat}^\bullet$,  see \ref{ss:CycMQV} in full generalities. 
This space admits a quasi-Poisson bracket, i.e. a bivector field whose associated bracket has some failure to satisfy the Jacobi identity \cite{AKSM}. This failure is governed by a natural action of $\CC^\times \times \Gl(n)^m$ on $\MM_{Q_\dfat,\nfat}^\bullet$, which has the further property of being endowed with a corresponding Lie group valued moment map. We can in that way perform a quasi-Poisson version of Hamiltonian reduction \cite{AKSM} and end up with a multiplicative quiver variety, that we denote $\Cnqm$ (it depends on a parameter $\qfat\in\CC^\times \times (\CC^\times)^{\Z_m}$). 

\medskip 

Our first main aim will be to study several Poisson algebras inside the coordinate ring of the space $\MM_{Q_\dfat,\nfat}^\bullet$ before reduction, see Section \ref{S:Subalg}. Note that the ring of functions on $\MM_{Q_\dfat,\nfat}^\bullet$ has a quasi-Poisson bracket only, but upon restriction to invariant functions it defines a Poisson bracket, which we use to form these Poisson algebras of invariant functions. 
We will prove that some of the corresponding Hamiltonian vector fields can be explicitly integrated, and we will write down the precise expression of their flows. The interest of this process is that these Poisson algebras descend to  $\Cnqm$ since they are made of invariant functions. Thus, we can obtain reduced flows on $\Cnqm$ by projecting those on $\MM_{Q_\dfat,\nfat}^\bullet$. 
We will also see that these algebras define holomorphic (degenerately or Liouville) integrable systems, see Section \ref{S:Int}. 
At this point we should emphasise that, to simplify the exposition, we are assuming the dimension vector $\nfat$ to be $(1,n\delta)$ \eqref{Eq:dim-ndelta}, because our proofs of integrability on $\Cnqm$ rely on that choice. Nevertheless, several results of Section \ref{S:Subalg} can be stated on $\MM_{Q_\dfat,\nfat}^\bullet$ without this assumption on $\nfat$, cf. Remark~\ref{Rem:Dim}.

Our second main aim will be to link these integrable systems to the previous results on RS systems recalled earlier. To do so, it is necessary to obtain local coordinates and analyse the local Poisson structure on   $\Cnqm$. 
This will be done after some tedious derivations throughout Section \ref{S:Loc}. The upshot is that we will be able to interpret the variety  $\Cnqm$ as a phase space supporting a many-body system parametrising $n$ particles endowed with several types of spin variables.

Let us point out the major difference between the methods used in this work and the ones considered in \cite{CF1,CF2,F1}. In the latter works, the noncommutative Poisson geometry of Van den Bergh \cite{VdB1} was used to perform most computations. This machinery allows to work with a version of the Poisson bracket on a multiplicative quiver variety directly at the level of the corresponding quiver. It is then possible to induce calculations made at the level of the quiver back to the multiplicative quiver variety, in the spirit of the Kontsevich-Rosenberg principle \cite{Ko,KoR}. While this point of view is extremely interesting and allows to make ``index-free'' computations, it is not taken in this work and we stay completely at a geometric level. 
The advantage is that this makes the results presented in this paper more accessible, since any researcher in mathematical physics knows how to handle Poisson brackets.  The price to pay is that computations become more cumbersome.  

\medskip 

Finally, we have to emphasise that the present work opens the door to several research directions. 
On the one hand, a major achievement of this paper is the construction of a new class of classical integrable systems which generalise the trigonometric Ruijsenaars-Schneider system with spin variables. It is therefore interesting to understand which features of the celebrated Ruijsenaars-Schneider systems can be defined uniformly for this whole new class of systems (e.g. bi-Hamiltonian structure, pole dynamics of solutions to an integrable hierarchy, etc.). We are also led to ask about the quantisation of these systems, due to the prominent role still played by Macdonald-Ruijsenaars operators in mathematical physics nowadays.
On the other hand, we recalled that several known integrable systems on (multiplicative) quiver varieties are connected with various types of structures, such as Double Affine Hecke Algebras, hyperpolygon spaces, or Coulomb branches in supersymmetric gauge theory, to name a few. A particularly challenging task would consist in unveiling similar connections for the integrable systems considered in the present work. 
We plan to investigate some of these problems in the future. 

\medskip

{\bf Layout.} Preliminary materials are reviewed in Section \ref{S:Cons}, such as quasi-Poisson geometry and the construction of multiplicative quiver varieties. In particular, given an extended cyclic quiver $Q_\dfat$, we define the quasi-Hamiltonian variety  $\MM_{Q_\dfat,\nfat}^\bullet$ and the corresponding multiplicative quiver varieties $\Cnqm$ that play a central role in the rest of this paper. 
In Section \ref{S:Subalg}, we study Poisson algebras of invariant functions on $\MM_{Q_\dfat,\nfat}^\bullet$. 
We begin with two large algebras, denoted $\IIl$ \eqref{Eq:Palg-IL} and $\II_+$ \eqref{Eq:Palg-I+}, which have a small Poisson centre. We integrate explicitly the Hamiltonian flows associated with particular functions in their Poisson centre as part of Propositions \ref{Pr:floYcy}--\ref{Pr:floTcy}. 
Next, we move on to the definition of \emph{abelian} Poisson algebras of invariant functions  on $\MM_{Q_\dfat,\nfat}^\bullet$. 
These algebras, denoted $\LLl$ \eqref{Eq:Palg-LL} and $\LL_+$ \eqref{Eq:Palg-L+}, are generated by elements which depend on the restriction of the moment map to closed subvarieties of $\MM_{Q_\dfat,\nfat}^\bullet$, see \ref{sss:Embed}. We give the explicit form of the Hamiltonian flows associated with the generators of these algebras in Propositions \ref{Pr:floYcy2}--\ref{Pr:floTcy2}. 
The remainder of the paper is concerned with the multiplicative quiver variety $\Cnqm$, obtained from $\MM_{Q_\dfat,\nfat}^\bullet$  by quasi-Hamiltonian reduction. In Section \ref{S:Loc}, local coordinates are constructed, and the local Poisson structure is computed, see Corollary \ref{cor:CyPoi}. Some specific functions on $\Cnqm$ are written locally in \ref{ss:Loc-express}, and the relation to previous works is described. 
Finally, we prove as part of Section \ref{S:Int} that the Poisson algebras constructed on $\MM_{Q_\dfat,\nfat}^\bullet$ descend to integrable systems on $\Cnqm$. More precisely, the algebras $\IIl$ and $\II_+$ yield degenerately integrable systems, see Theorems \ref{Thm:DIS-Z} and \ref{Thm:DIS-T}, while the algebras $\LLl$ and $\LL_+$ induce Liouville integrable systems, see Theorems \ref{Thm:LiouIS-L} and \ref{Thm:LiouIS-T}.  
Note that we regard $\Cnqm$ as a complex manifold throughout Section \ref{S:Int}.

\medskip

{\bf Notation.} In the text, we let $\N=\{0,1,2,\ldots\}$, $\Z=\{\ldots,-1,0,1,\ldots\}$, while $\CC$ is the complex field. We write $\N^\times=\N\setminus \{0\}$, $\Z^\times=\Z\setminus \{0\}$ and $\CC^\times=\CC\setminus \{0\}$.  
For any $n_0,n_1\in \N^\times$, $\Mat(n_0\times n_1,\CC)$ denotes the vector space of $n_0\times n_1$ complex matrices with zero element denoted $0_{n_0 \times n_1}$. 
The general linear group and its Lie algebra over $\CC$ of dimension $n\in \N^\times$ are denoted $\Gl(n)$ and $\gl(n)$ respectively; we identify them with subalgebras of $\Mat(n\times n,\CC)$ using their faithful representations in $\CC^n$. 
By convention,  $\Gl(0)=\gl(0)=\{0\}$ and $\Mat(n_0\times n_1,\CC)=\{0\}$ whenever $n_0=0$ or $n_1=0$. 
Given a finite set $I$ and $n_s\in \N$ for each $s\in I$, we set $\Gl(\nfat):=\prod_{s\in I}\Gl(n_s)$ and $\gl(\nfat):=\prod_{s\in I}\gl(n_s)$ for $\nfat:=(n_s)_{s\in I}$. 
The identity operator on a vector space $V$ is denoted $\Id_V$ and we let $\Id_n:=\Id_{\CC^n}$ for any $n\in \N^\times$;  we also regard $\Id_n$ as the identity element of $\Gl(n)$. 
We use Kronecker delta $\delta_{ij}$ or $\delta_{(i,j)}$ for any $i,j\in \Z$ which is $+1$ if $i=j$ and zero otherwise. 
More generally, given a proposition $P$, $\delta_P$ is $+1$ if $P$ is true and is zero if $P$ is false. For example, we write $\delta_{(i\neq j)}=1-\delta_{ij}=\delta_{(i>j)}+\delta_{(i<j)}$ for $i,j\in \Z$.  

\medskip

{\bf Acknowledgements.} I am very grateful to Oleg Chalykh for introducing me to this fascinating subject and for the numerous discussions we had over the years. 
It is also my pleasure to thank Alexey Bolsinov, Oleg Chalykh, L\'{a}szl\'{o} Feh\'{e}r and Derek Harland for insightful discussions and comments, 
as well as suggestions to improve the text received from referees. 
This work is partly based on the PhD thesis of the author \cite{Fth} funded by a University of Leeds 110 Anniversary Research Scholarship. 
The writing up of this paper was supported by a Rankin-Sneddon Research Fellowship of the University of Glasgow.

\section{Construction of the varieties}  \label{S:Cons} 

The Poisson varieties (in fact complex Poisson manifolds) that we consider in the other sections of this paper are examples of multiplicative quiver varieties associated with the cyclic quiver. These spaces are obtained by a process called quasi-Hamiltonian reduction. The reduction is performed in the context of quasi-Poisson geometry, where the master phase space is endowed with a \emph{quasi-Poisson} bracket, an analogue of a Poisson bracket which has some failure to satisfy the Jacobi identity. We review these notions in the remainder of the section.


\subsection{Affine quasi-Poisson geometry} \label{ss:qP}

We follow the algebraic treatment \cite{VdB1} of quasi-Poisson geometry which was introduced by Alekseev, Kosmann-Schwarzbach and Meinrenken \cite{AKSM} (see also \cite{B1,LB,MT14}). 
Consider a complex reductive algebraic group $G$ with Lie algebra $\g$, and denote by $\ad$ the adjoint action of $G$ on $\g$. Assume that $\g$ is endowed with a non-degenerate symmetric bilinear form $(-,-)$ which is $\ad$-invariant. Given an arbitrary basis $(f_a)_{a\in A}$ of $\g$, we form its dual basis $(f^a)_{a\in A}$ with respect to $(-,-)$. We can define the Cartan trivector $\phi \in\wedge^3\g$ by 
\begin{equation}
 \phi =\frac{1}{12} \sum_{a,b,c\in A} C_{abc} \, f^a\wedge f^b \wedge f^c\,, \quad \text{where }C_{abc}=(f_a,[f_b,f_c])\,,
\end{equation}
and we get that it is $\ad$-invariant because the bilinear form is. Let us now fix an affine complex variety $M$ with structure sheaf $\OO_M$, which is endowed with a regular action of the algebraic group $G$. This induces an infinitesimal action of $\g$, which associates with any $\xi\in \g$ the vector field $\xi_M$ defined over any open $U\subset M$ by 
\begin{equation} \label{EqinfVectM}
 \xi_M(F)(m)=\left.\frac{d}{dt}\right|_{t=0} F(\exp(-t\xi)\cdot m)\,, \quad \text{ for all }F\in \OO_M(U),\,\,m\in M\,. 
\end{equation}
This operation can be naturally extended for any $k\geq 1$  in such a way that it returns a $k$-vector field $\psi_M$ to $\psi\in\wedge^k\g$.
\begin{defn}
 A $G$-invariant antisymmetric $\CC$-bilinear map $\br{-,-}:\OO_M\times \OO_M\to \OO_M$ on $M$ which is  a biderivation (i.e. a derivation in each argument) is called a \emph{quasi-Poisson bracket} if for any open $U\subset M$ and  $F_1,F_2,F_3\in \OO_M(U)$, 
 \begin{equation} \label{Eq:JacPhi}
\br{F_1,\br{F_2,F_3}} +\br{F_2,\br{F_3,F_1}} + \br{F_3,\br{F_1,F_2}} = \frac12 \phi_M(F_1,F_2,F_3)\,.  
 \end{equation}
\end{defn}
Using the group structure on $G$, we can introduce for any $\xi \in \g$ the left- and right-invariant vector fields $\xi^L$ and $\xi^R$ on $G$. They are defined for any regular function $F$ over an open $V\subset G$ by  
\begin{equation}
  \begin{aligned} \label{EqinfLR}
    \xi^L(F)(z)=\left.\frac{d}{dt}\right|_{t=0} \, F\left(z \cdot \exp(t \xi) \right)\,, \quad
\xi^R(F)(z)=\left.\frac{d}{dt}\right|_{t=0} \, F\left( \exp(t \xi) \cdot z \right)\,, 
\qquad \,\,z\in V\,.
  \end{aligned}
\end{equation}

\begin{defn}
Let $\Phi:M \to G$ be a morphism of varieties intertwining the $G$-action on $M$ with the conjugation action on $G$. 
We say that $\Phi$ is a \emph{multiplicative moment map} (or simply a \emph{moment map}) if, for any regular function $F$ on $G$, we have the following equality of vector fields on $M$:
\begin{equation} \label{momapScheme}
  \br{F \circ \Phi,-}=\frac12 \sum_{a\in A}\,\Phi^\ast\left((f_a^L+f_a^R)(F) \right) \,\,(f^a)_M\,.
\end{equation}
We call the triple $(M,\br{-,-},\Phi)$ a \emph{quasi-Hamiltonian variety}.
\end{defn}

\begin{exmp}[\cite{VdB1}] \label{Exmp:qP1} Fix  $\nfat=(n_0,n_1)\in \N^\times\times \N^\times$ and form the smooth affine variety  
\begin{align}
 \MM_{\nfat}=\left\{(X,Y)\mid X\in \Mat(n_0\times n_1,\CC),\,Y\in \Mat(n_1\times n_0,\CC)\right\}\,.
\end{align}
The group $\Gl(\nfat):=\Gl(n_0)\times \Gl(n_1)$ acts on  $\MM_\nfat$ through 
\begin{equation*}
 (g_0,g_1)\cdot (X,Y)=(g_0Xg_1^{-1},g_1Yg_0^{-1})\,.
\end{equation*}
We can get an antisymmetric biderivation on $\MM_{\nfat}$ by extending (thanks to the derivation rule in each argument and antisymmetry) the $\CC$-bilinear map 
\begin{equation}
 \br{X_{ij},X_{i'j'}}=0=\br{Y_{kl},Y_{k'l'}},\quad 
 \br{X_{ij},Y_{kl}}=\delta_{kj}\delta_{il}+\frac12 (YX)_{kj} \delta_{il} + \frac12 \delta_{kj} (XY)_{il}\,,
\end{equation}
where $1\leq i,i',l,l'\leq n_0$, $1\leq j,j',k,k'\leq n_1$, while $X_{ij}$ denotes the function that returns the $(i,j)$ entry of $X$ (and similarly for other matrices). Endowing $\gl(\nfat):=\gl(n_0)\times \gl(n_1)$ with the trace pairing, we can check that \eqref{Eq:JacPhi} is satisfied on the generators $(X_{ij},Y_{kl})$ of the algebra of regular functions on $\MM_{\nfat}$, hence $\br{-,-}$ is a quasi-Poisson bracket.     
Furthermore, on the open subvariety  
$\MM_{\nfat}^\bullet \subset \MM_{\nfat}$ defined by the condition $\det(\Id_{n_0}+XY)\neq 0$, $\Phi(X,Y)=(\Id_{n_0}+XY,(\Id_{n_1}+YX)^{-1})\in \Gl(\nfat)$ defines a moment map.

Following the terminology from \cite{B1,Y}, we call the quasi-Hamiltonian variety $\MM_{\nfat}^\bullet$ the \emph{Van den Bergh space}.  When $n_0=n_1$, after restriction to the open subvariety of $\MM_{\nfat}^\bullet$ defined by the condition $\det(X)\neq0$, the moment map $\Phi$ becomes $\Phi(X,Y)=(XZ,X^{-1}Z^{-1})$, where $Z:=Y+X^{-1}$. We get in this way the (quasi-Hamiltonian) double of $\Gl(n_0)$ \cite[Example 5.4]{AKSM}.
\end{exmp}

We now present two results that are central to the construction of multiplicative quiver varieties in the next subsection. 
They can be found in  \cite[\S~5-6]{AKSM}, see also \cite[\S~2]{B1}.

\begin{prop}[Fusion]  \label{Pr:qPfus}
Let $(M,\br{-,-},(\Phi_1,\Phi_2,\Psi))$ be a quasi-Hamiltonian $(G \times G  \times H)$-variety, for complex reductive algebraic groups $G,H$. 
Denote by $(f_a)_{a\in A}$ and $(f^a)_{a\in A}$ two dual bases of the Lie algebra $\g$ of $G$, 
and introduce the antisymmetric biderivation $\br{-,-}_{\operatorname{fus}}$ on $M$ defined by the following bivector field
\begin{equation}
  P_{\operatorname{fus}}=-\frac12 \sum_{a\in A} (f_a,0,0)_M \wedge (0,f^a,0)_M\,. 
\end{equation}
Then the diagonal map $G \to G \times G$ induces a regular $G \times H$ action on $M$ such that for
\begin{equation}
 \br{-,-}^f= \br{-,-}+\br{-,-}_{\operatorname{fus}} \,, \quad  \Phi^f=(\Phi_1 \Phi_2,\Psi)\,,
\end{equation}
the triple $(M,\br{-,-}^f,\Phi^f)$ is a quasi-Hamiltonian $(G \times H)$-variety. 
\end{prop}

We will use fusion in the presence of several quasi-Hamiltonian varieties. For example, if we are given a quasi-Hamiltonian $(G\times H_i)$-variety $M_i$ for $i=1,2$,  we can get that $M_1\times M_2$ is a quasi-Hamiltonian $(G\times H_1\times H_2)$-variety, denoted $M_1 \circledast M_2$. Note that, if we start with the moment map $(\Phi_i,\Psi_i)$ on $M_i$, we can end up with a quasi-Hamiltonian structure such that either $(\Phi_1\Phi_2,\Psi_1,\Psi_2)$ is the moment map, or  $(\Phi_2\Phi_1,\Psi_1,\Psi_2)$ is. These two structures are non-trivially isomorphic.  

\begin{prop}[Reduction]  \label{Pr:qPred}
Let $(M,\br{-,-},\Phi)$ be a quasi-Hamiltonian $G$-variety and fix an element $g \in G$ invariant under conjugation. 
Then the quasi-Poisson bracket satisfies the Jacobi identity when restricted to $G$-invariant regular functions, 
and it gives rise to a well-defined Poisson bracket on the geometric invariant theory (GIT) quotient 
$\Phi^{-1}(g)/\!\!/G$, which is called the \emph{reduced Poisson variety} (if it is not empty). 
\end{prop} 


\begin{exmp}\label{Exmp:qP2}
 Take the quasi-Hamiltonian  $\Gl(\nfat)$-variety $\MM_{\nfat}^\bullet$ from Example \ref{Exmp:qP1} with $\nfat=(n_0,n_1)$. When $n:=n_0=n_1$, we can perform fusion and consider  $\MM_{\nfat}^\bullet$ as a quasi-Hamiltonian $\Gl(n)$-variety for the moment map $\Phi(X,Y)=(\Id_{n}+XY)(\Id_{n}+YX)^{-1}$ and the quasi-Poisson bracket 
 \begin{subequations}
\begin{align}
  \br{X_{ij},X_{kl}}=&\frac12 \big((X^2)_{kj} \delta_{il} - \delta_{kj}(X^2)_{il}\big)  \,, \quad 
  \br{Y_{ij},Y_{kl}}=-\frac12 \big((Y^2)_{il} \delta_{kj} - \delta_{il}(Y^2)_{kj}\big)  \,,\\
 \br{X_{ij},Y_{kl}}=&\delta_{kj}\delta_{il}+\frac12\Big( (YX)_{kj} \delta_{il} +  \delta_{kj} (XY)_{il} + Y_{kj} X_{il} - X_{kj} Y_{il} \Big)\,.
\end{align}
\end{subequations}
Note that $\Phi^{-1}(q\Id_{n})$ is empty if $q\in \CC^\times$ is not an $n$-th root of unity. 
Let us also consider 
\begin{equation*}
\MM_{(n,1)}^\bullet=\{(V,W) \mid  V\in \Mat(1\times n,\CC),\quad W\in \Mat(n\times 1,\CC) \mid 1+VW\neq 0\}\,,
\end{equation*}
with moment map $\MM_{(n,1)}^\bullet\to \Gl(n)\times \CC^\times$, $(V,W)\mapsto((\Id_n+WV)^{-1},1+VW)$. Performing fusion, we obtain that $\MM_{\nfat}^\bullet\times \MM_{(n,1)}^\bullet$ is a quasi-Hamiltonian $(\Gl(n)\times \CC^\times)$-variety with moment map 
\begin{equation*}
 \tilde{\Phi}(X,Y,V,W)=((\Id_{n}+XY)(\Id_{n}+YX)^{-1}(\Id_n+WV)^{-1},1+VW)\,.
\end{equation*}
When $q$ is not an $n$-th root of unity, the reduced Poisson variety $\tilde{\Phi}^{-1}(g)/\!\!/(\Gl(n)\times \CC^\times)$ for $g=(q\Id_n,q^{-n})$ is a smooth complex variety  of dimension $2n$, see e.g. \cite{CF1,Ob}.
\end{exmp}


\subsection{Multiplicative quiver varieties} \label{ss:MQV}

The Van den Bergh space $\MM_{\nfat}^\bullet$ described in Example \ref{Exmp:qP1} is the ``building block'' that is used to construct multiplicative quiver varieties by quasi-Hamiltonian reduction. We now give a formal definition, and we advise the reader to look at Example \ref{Exmp:qP2} or \ref{sss:CycMQV-gen} to have in mind some precise examples of the geometric construction of such varieties. For alternative presentations of the definition, see \cite{CF1,CBS,MGN,TS,Y}. (Note that we do not need the approach of Boalch \cite{Bo15} which generalises multiplicative quiver varieties.)

\subsubsection{Construction using quivers} \label{sss:MQV-quivers}

\begin{figure}
\centering
\begin{tikzpicture}[scale=0.8]
  \node[circle,fill=black,inner sep=2pt]  (vL0) at (-7,0) {};
  \node[circle,fill=black,inner sep=2pt]  (vL1) at (-4,0) {};
    \node[font=\small]  (vL2) at (-7,-0.4) {$0$};
  \node[font=\small]  (vL3) at (-4,-0.4) {$1$};
  \node (titL) at (-7.5,0.5) {i)};
  \path[->,>=latex] (vL0) edge [bend left=10]  node[above,font=\small]{$a$}   (vL1) ;
  \node[circle,fill=black,inner sep=2pt]  (vM0) at (-2,0) {};
  \node[circle,fill=black,inner sep=2pt]  (vM1) at (1,0) {};
    \node[font=\small]  (vM2) at (-2,-0.4) {$0$};
  \node[font=\small]  (vM3) at (1,-0.4) {$1$};
  \node (titM) at (-2.5,0.5) {ii)};
  \path[->,>=latex] (vM0) edge [bend left=10]  node[above,font=\small]{$a$}   (vM1) ;
\path[->,>=latex] (vM1) edge [bend left=10]  node[below,font=\small]{${a^\ast}$}   (vM0) ;
  \node[circle,draw=black,inner sep=2pt,font=\small]  (v0) at (3,0) {$n_0$};
  \node[circle,draw=black,inner sep=2pt,font=\small]  (v1) at (6,0) {$n_1$};
  \node (titleL) at (2,0.5) {iii)};
\path[->,>=latex] (v0) edge [bend left=10]  node[above,font=\small]{$X:=X_a$}   (v1) ;
\path[->,>=latex] (v1) edge [bend left=10]  node[below,font=\small]{$Y:=X_{a^\ast}$}   (v0) ;
 \end{tikzpicture}
 \caption{The Van den Bergh space $\MM_{\nfat}^\bullet$ for $\nfat=(n_0,n_1)$ from Example \ref{Exmp:qP1} depicted as a decorated quiver in iii). It parametrises representations of the quiver in ii), which is the double of the \emph{simplest} quiver with one arrow between two vertices given in i).} 
\label{Fig:Quiv1}
\end{figure}

A quiver (i.e. directed graph) is a quadruple $(I,Q,t,h)$, which consists of a finite set of vertices $I$, a finite set of arrows $Q$ between these vertices, which are oriented thanks to the head and tail maps 
$h:Q \to I$, $t:Q\to I$. 
These maps assign to an arrow $a \in Q$ its tail (starting vertex) $t(a)$ and its head (ending vertex) $h(a)$, respectively.
To shorten the notation, hereafter we simply denote a quiver $(I,Q,t,h)$ as $Q$. 
We also depict an arrow $a\in Q$ as $a:i\longrightarrow j$ or $i\stackrel{a}{\longrightarrow}j$ if $t(a)=i$ and $h(a)=j$ for $i,j\in I$.

The double of $Q$, denoted $\overline{Q}$, 
is the quiver obtained from $Q$ by adding an opposite arrow $h(a)\stackrel{a^\ast}{\longrightarrow} t(a)$ for each $a\in Q$. One then defines the map 
\begin{align*}
 {(-)}^\ast:\overline{Q}\to \overline{Q}\,, \text{ with }
 \left\{ \begin{array}{cc} a \mapsto a^\ast&\text{if }a\in Q\,, \\  a^\ast \mapsto (a^\ast)^\ast:=a&\text{if }a^\ast\in \overline{Q}\setminus Q\,,
\end{array} \right. 
\end{align*}
which is an involution on $\overline{Q}$, and the map 
\begin{align*}
\epsilon:\overline{Q}\to \{-1,1\}\,,  \text{ with }
 \left\{ \begin{array}{cc} \epsilon(a)=+1&\text{if }a\in Q\,, \\  \epsilon(a^\ast)=-1&\text{if }a^\ast\in \overline{Q}\setminus Q\,,
\end{array} \right. 
\end{align*}
which distinguishes the arrows in $\overline{Q}$ belonging to $Q$ from those which have been added as their opposites.  

Next, fix a vector $\nfat=(n_s)\in \N^I$ which assigns a dimension $n_s$ to each vertex $s\in I$. 
We form the following variety consisting of $2|Q|=|\overline{Q}|$ matrices  
\begin{equation} \label{Eq:qHamSpace1}
 \MM_{Q,\nfat}=\left\{(X_a)_{a\in \overline{Q}}  \Big| X_a\in \Mat\big(n_{t(a)}\times n_{h(a)},\CC\big) \right\}\,,
\end{equation}
which can be identified with the affine space of dimension $2\sum_{a\in Q} n_{t(a)}n_{h(a)}$. 
Letting $\VV=\bigoplus_{s\in I} \VV_s$ for $\VV_s=\CC^{n_s}$, the space $\MM_{Q,\nfat}$ parametrises representations of $\overline{Q}$ in $\VV$ since we can write\footnote{We follow the conventions taken by Van den Bergh in \cite{VdB1}. The difference with the convention followed e.g. in \cite{CBS} is that we consider quiver representations as right modules of the path algebra where we write paths from \emph{left to right}.} 
\begin{equation} \label{Eq:qHamSpace1-bis}
 \MM_{Q,\nfat}=\left\{(X_a)_{a\in \overline{Q}}  \Big| X_a\in \Hom(\VV_{h(a)},\VV_{t(a)}) \right\}\,.
\end{equation}
Note that if $n_s=0$ for some $s\in I$, then $\VV_s=\{0\}$ and $X_a$ is the zero map for any $a\in \overline{Q}$ with $t(a)=s$ or $h(a)=s$. 
The variety $\MM_{Q,\nfat}$ contains 
\begin{equation} \label{Eq:qHamSpace2}
 \MM_{Q,\nfat}^\bullet =\left\{(X_a)_{a\in \overline{Q}}\mid \det(\Id_{n_{t(a)}}+X_aX_{a^\ast}) \neq 0 \quad \forall a \in \overline{Q} \right\} \subset  \MM_{Q,\nfat}\,.
\end{equation}
The space $\MM_{Q,\nfat}^\bullet$ is itself a smooth affine variety. 
It is convenient to depict $\MM_{Q,\nfat}^\bullet$ as follows: 
we decorate the quiver $\overline{Q}$ by assigning to each arrow $a\in \overline{Q}$  the matrix-valued function $X_a$, 
and by assigning to each vertex $s\in I$ the dimension $n_s$. 
The decorated quiver associated with the Van den Bergh space $\MM_{\nfat}^\bullet$ from Example \ref{Exmp:qP1} is obtained from the quiver $1\to 2$ with the decoration shown in  Figure \ref{Fig:Quiv1}.

Finally, fix for each $s\in I$ a total order on the set of arrows $\{a\in \overline{Q} \mid t(a)=s\}$. 
We use this order to define for each $s\in I$ a morphism of varieties   
\begin{equation} \label{Eq:qHamMomap}
 \Phi_s:\MM_{Q,\nfat}^\bullet\to \Gl(n_s),\quad 
\Phi_s((X_a)_{a\in \overline{Q}}) = \prod_{\substack{a\in \overline{Q}\, \text{s.t.}\\ t(a)=s}}^{\longrightarrow}  (\Id_{n_{t(a)}}+X_aX_{a^\ast})^{\epsilon(a)}\,.
\end{equation}
(The factors appearing in the product in \eqref{Eq:qHamMomap} are ordered increasingly according to the total order fixed above.)
This gives rise to a morphism of varieties $\Phi=(\Phi_s)_{s\in I}: \MM_{Q,\nfat}^\bullet \to \Gl(\nfat)$, 
where $\Gl(\nfat):=\prod_{s\in I}\Gl(n_s)$.  
Define the action of $\Gl(\nfat)$ on $\MM_{Q,\nfat}^\bullet$ by 
\begin{equation}\label{Eq:qHamAction}
 g \cdot (X_a)_{a\in \overline{Q}} = (g_{t(a)}X_a g_{h(a)}^{-1})_{a\in \overline{Q}}\,, \quad 
 g=(g_s)_{s\in I}\in \Gl(\nfat)\,.
\end{equation}
Note that $\Phi$ intertwines the action \eqref{Eq:qHamAction} on $\MM_{Q,\nfat}^\bullet$ with the conjugation action on $\Gl(\nfat)$.

\begin{defn} \label{Def:MQV}
 Fix $\qfat=(q_s)\in (\CC^\times)^I$, and let $\qfat \cdot \Id:= \prod_{s\in I} q_s \Id_{n_s}\in \Gl(\nfat)$. The GIT quotient  
 \begin{equation} \label{Eq:MQV}
  \NN(Q,<,\nfat,\qfat):= \Phi^{-1}(\qfat \cdot \Id)/\!\!/\Gl(\nfat)
 \end{equation}
is called a \emph{multiplicative quiver variety}.
\end{defn}
The variety $\NN(Q,<,\nfat,\qfat)$ may be reducible or it could be empty.
In \ref{ss:CycMQV}, we will look at examples of smooth (possibly reducible) multiplicative quiver varieties. 
Up to isomorphism, $\NN(Q,<,\nfat,\qfat)$ does neither depend on the direction of the arrows in $Q$, nor on the orders of arrows chosen at each vertex \cite[Theorem 1.4]{CBS}.

\subsubsection{Construction using quasi-Poisson geometry} \label{sss:MQV-constr}

\begin{figure}
\centering
\begin{tikzpicture}[scale=0.8]
  \node[circle,draw=black,inner sep=2pt,font=\small]  (v0) at (-6.2,0) {$n_0$};
  \node[circle,draw=black,inner sep=2pt,font=\small]  (v1) at (-3.2,0) {$n_1$};
  \node (titleL) at (-6.5,1) {a)};
\path[->,>=latex] (v0) edge [bend left=10]  node[above,font=\small]{$X$}   (v1) ;
\path[->,>=latex] (v1) edge [bend left=10]  node[below,font=\small]{$Y$}   (v0) ;
  \node[circle,draw=black,inner sep=2pt,font=\small]  (vm) at (0,0) {$n$};
  \node (vmX) at (0,1) {$X$};
  \node (vmY) at (0,-1) {$Y$};
  \node (titleM) at (-1,1) {b)};
\draw[->,>=latex] (vm) to [out=120,in=60,looseness=15] (vm);
\draw[->,>=latex] (vm) to [out=300,in=240,looseness=15] (vm);
  \node[circle,draw=black,inner sep=2pt,font=\small]  (vR) at (3.5,0) {$n$};
  \node (vrX) at (3.5,1) {$X$};
  \node (vrY) at (3.5,-1) {$Y$};
  \node[circle,draw=black,inner sep=2pt,font=\small]  (vinfty) at (5.8,0) {$1$};
   \node (titleR) at (2.5,1) {c)};
\path[->,>=latex] (vinfty) edge [bend right=10]  node[above,font=\small]{$V$}   (vR);
\path[->,>=latex] (vR) edge [bend right=10]  node[below,font=\small]{$W$}   (vinfty) ;
\draw[->,>=latex] (vR) to [out=120,in=60,looseness=15] (vR);
\draw[->,>=latex] (vR) to [out=300,in=240,looseness=15] (vR);
 \end{tikzpicture}
\caption{a) Decorated quiver associated with the Van den Bergh space $\MM_{\nfat}^\bullet$ from Example \ref{Exmp:qP1} with $\nfat=(n_0,n_1)$. \newline
b) Decorated quiver for the same space seen as a quasi-Hamiltonian $\Gl(n)$-variety after fusion when $n=n_0=n_1$ (with total order $X<Y$ at the vertex), see Example \ref{Exmp:qP2}. \newline
c) Decorated quiver for the space $\MM_{\nfat}^\bullet\times \MM_{(n,1)}^\bullet$ seen as a quasi-Hamiltonian $(\Gl(n)\times \CC^\times)$-variety after fusion (total order $X<Y<W$ at the common vertex), see Example \ref{Exmp:qP2}.} 
\label{Fig:simple}
\end{figure}

The variety $\NN(Q,<,\nfat,\qfat)$ from Definition \ref{Def:MQV} can be equipped with a Poisson bracket using quasi-Poisson reduction \cite{VdB1}.  Indeed, we start by remarking that, as a variety,
$\MM_{Q,\nfat}^\bullet$ \eqref{Eq:qHamSpace2} is the union of $|Q|$ Van den Bergh spaces 
\begin{equation*}
\MM_{(n_{t(a)},n_{h(a)})}^\bullet = \Big\{ (X_a,X_{a^\ast}) \mid  X_a\in \Mat(n_{t(a)}\times n_{h(a)},\CC),\,X_{a^\ast}\in \Mat(n_{h(a)}\times n_{t(a)},\CC) \Big\}\,, \quad a\in Q\,.
\end{equation*}
Each $\MM_{(n_{t(a)},n_{h(a)})}^\bullet$ is a quasi-Hamiltonian $\big(\Gl(n_{t(a)})\times\Gl(n_{h(a)})\big)$-variety for the quasi-Poisson bracket
\begin{equation*}
\begin{aligned}
  &\br{(X_a)_{ij},(X_a)_{kl}}=0\,, \quad \br{(X_{a^\ast})_{ij},(X_{a^\ast})_{kl}}=0\,, \\
 &\br{(X_a)_{ij},(X_{a^\ast})_{kl}}=\delta_{kj}\delta_{il}+\frac12 (X_{a^\ast}X_a)_{kj} \delta_{il} + \frac12 \delta_{kj} (X_a X_{a^\ast})_{il}\,,
\end{aligned}
\end{equation*}
and the moment map $\Psi_a(X_a,X_{a^\ast})=(\Id_{n_{t(a)}}+X_aX_{a^\ast},(\Id_{n_{h(a)}}+X_{a^\ast}X_a)^{-1})$, see Example \ref{Exmp:qP1}. Therefore, $\MM_{Q,\nfat}^\bullet$ can be seen as a quasi-Hamiltonian variety for the action of 
\begin{equation}
 \prod_{a\in Q} \big(\Gl(n_{t(a)})\times\Gl(n_{h(a)})\big) \simeq \prod_{s\in I} \Gl(n_s)^{\kappa(s)}\,,
\end{equation}
where $\kappa(s)$ is the  number of arrows $a\in \overline{Q}$ with $t(a)=s$; equivalently $\kappa(s)$ is the  number of arrows $a\in Q$ with $t(a)=s$ or $h(a)=s$ (counted twice if $t(a)=h(a)=s$). 
Therefore, by repeatedly performing fusion as in Proposition \ref{Pr:qPfus}, of the $\kappa(s)$ copies of the action of $\Gl(n_s)$, for each $s\in I$, we can end up with a structure of quasi-Hamiltonian $\Gl(\nfat)$-variety on $\MM_{Q,\nfat}^\bullet$. Fixing a total order on the elements $a\in \overline{Q}$ with $t(a)=s$ as in \ref{sss:MQV-quivers}, we can uniquely choose the quasi-Hamiltonian structure so that $\Phi=(\Phi_s)_{s\in I}$ is a moment map for $\Phi_s:\MM_{Q,\nfat}^\bullet \to \Gl(n_s)$ defined in \eqref{Eq:qHamMomap}.
This information can be read directly from the decorated quiver and the choice of ordering, see e.g. Figure \ref{Fig:simple}.

In the end, we can perform quasi-Hamiltonian reduction of the space $\MM_{Q,\nfat}^\bullet$ at $\qfat \cdot \Id:= \prod_s q_s \Id_{n_s}$, in the sense of Proposition \ref{Pr:qPred}. We get in this way that the multiplicative quiver variety $\NN(Q,<,\nfat,\qfat)$ \eqref{Eq:MQV} is endowed with a Poisson bracket.

The Poisson bracket on $\NN(Q,<,\nfat,\qfat)$ can be restricted to the smooth locus $\NN^{\textrm{st}}(Q,<,\nfat,\qfat)$,  where it is in fact non-degenerate by the work of Van den Bergh \cite{VdB2}. Yamakawa \cite{Y} observed that, up to symplectomorphism, $\NN^{\textrm{st}}(Q,<,\nfat,\qfat)$ does not depend on the direction of the arrows in $Q$, nor on the orders of arrows chosen at each vertex, see \cite[Proposition 3.3]{Y}. More generally, this symplectomorphism is obtained by restriction of the isomorphism of Crawley-Boevey and Shaw \cite[Theorem 1.4]{CBS}, and the latter can be extended to an isomorphism of Poisson varieties by combining Theorem 4.12 and Propositions 5.6, 5.9 in \cite{F3}.  Geometrically, the independence on the order taken at each vertex is a consequence of the fusion process and follows from \cite[Proposition 5.7]{AKSM}.

\begin{rem}
 By construction, the quasi-Hamiltonian  $\Gl(\nfat)$-variety $\MM_{Q,\nfat}^\bullet$ can be seen as a representation of (a localisation of) the path algebra of the quiver $\overline{Q}$. Note, however, that this holds if we write paths \emph{from left to right}, in agreement with the conventions of Van den Bergh \cite{VdB1,VdB2}. 
\end{rem}
\begin{rem} \label{Rem:qPbracket}
 The quasi-Poisson bracket can be obtained in terms of a bivector (defined using the quiver and the order taken at each vertex) by combining Theorem 6.7.1 and  Proposition 7.13.2 in \cite{VdB1}. It can easily be written in terms of the regular functions returning the different entries $(X_a)_{ij}$ of the matrices $(X_a)_{a\in \overline{Q}}$ parametrising $\MM_{Q,\nfat}^\bullet$ if we use together \cite[Theorem 3.3]{F2} and \cite[Proposition 7.5.1]{VdB1}. 
\end{rem}


\subsection{Varieties associated with the cyclic quiver} \label{ss:CycMQV}

\subsubsection{Notation} \label{sss:CycMQV-not}
Fix an integer $m \geq 2$ and let $\Z_m:=\Z/m \Z$. 
We identify $\Z_m$ as a set with  $\{0,\ldots,m-1\}$ so that an integer $s\in \{0,\ldots,m-1\}$ corresponds to its coset in $\Z_m$. Using the natural total order $0<1<\ldots<m-1$, we define a total order on $\Z_m$ by requiring that the chosen  identification $\Z_m\to \{0,\ldots,m-1\}$ preserves total orders. 
Next, fix $\dfat=(d_0,\ldots,d_{m-1})\in \N^{\Z_m}$ with $d_0\geq1$, so that $|\dfat|:=\sum_{s\in \Z_m} d_s\geq 1$. 

We define the quiver $\overline{Q}_\dfat$ as an extension of the cyclic quiver on $m$ vertices as follows. First, let $Q_\dfat$ be the quiver with vertex set $I=\Z_m \cup \{\infty\}$, and whose edge set $Q_\dfat$ consists, for all $s\in \Z_m$, of $d_{s}+1$ arrows given by $x_s:s \to s+1$ and $v_{s,\alpha}:\infty \to s$ with $\alpha=1,\ldots,d_s$ (there is no arrow $\infty \to s$ if $d_s=0$). The quiver $Q_\dfat$ is the last quiver depicted in Figure \ref{Fig:Tab1}.  Second, $\overline{Q}_\dfat$  is constructed as the double of $Q_\dfat$. 
The quiver $\overline{Q}_\dfat$ has the same vertex set $I$,  and its edge set consists of $2m+2|\dfat|$ arrows given by the $(x_s,v_{s,\alpha})$ described above together with the opposite arrows  $y_s:=x_s^\ast:s+1 \to s$, 
$w_{s,\alpha}:=v_{s,\alpha}^\ast: s \to \infty$ for all $1\leq \alpha \leq d_s$, $s \in \Z_m$.
At each vertex of $\overline{Q}_\dfat$, we introduce the following order of the arrows starting at it: 
 \begin{subequations}
       \begin{align}
 \text{at }0:&\quad&  
 x_0< y_{m-1} < w_{0,1} < \ldots < w_{0,d_0}\,; \label{CyOrd0} \\
\text{at }s\neq 0:&\quad&  
x_s < y_{s-1} < w_{s,1} < \ldots  < w_{s,d_s}\,; \label{CyOrds} \\ 
\text{at }\infty:&\quad&
 v_{0,1}< \ldots < v_{0,d_0} < v_{1,1} <\ldots  < v_{m-1,1} <  \ldots < v_{m-1,d_{m-1}}\,. \label{CyOrdinf}
        \end{align}
\label{Eq:CyOrd}
  \end{subequations}
(We omit the elements $v_{s,\alpha},w_{s,\alpha}$ in the  different orders when $d_s=0$.) 

Finally, consider the dimension vector and parameter
\begin{align}
\nfat=(1,n_0,\ldots,n_{m-1})\in \N\times \N^{\Z_m} \,, \quad \qfat=(q_\infty, q_0,\ldots,q_{m-1})\in \CC^{\times}\times (\CC^\times)^{\Z_m}\,.
\end{align}
We assume that $|\nfat|=\sum_{s\in \Z_m} n_s>0$ and furthermore 
\begin{equation} \label{Eq:Cond-q}
 q_\infty=\prod_{s\in \Z_m} q_s^{-n_s}\,.
\end{equation}

\subsubsection{General construction} \label{sss:CycMQV-gen}

We form the variety $\MM_{Q_\dfat,\nfat}$ associated with $\overline{Q}_\dfat$ as in \ref{sss:MQV-quivers}. 
It is parametrised by $2m+2|\dfat|$ matrices given by 
\begin{equation}  \label{Eq:CycMQV-paraMat}
 \begin{aligned}
 X_s \in \Mat(n_s\times n_{s+1},\CC)\,, \quad 
&Y_s \in \Mat(n_{s+1}\times n_s,\CC)\,, \quad s\in \Z_m\,, \\ 
V_{s,\alpha}\in \Mat(1\times n_s,\CC)\,, \quad
&W_{s,\alpha}\in \Mat(n_s\times 1,\CC)\,, \quad 1\leq \alpha \leq d_s\,,\,\, s\in \Z_m\,,
 \end{aligned} 
\end{equation}
which represent the arrows $x_s,y_s,v_{s,\alpha},w_{s,\alpha}$ respectively. 
The subspace $\MM_{Q_\dfat,\nfat}^\bullet \subset\MM_{Q_\dfat,\nfat}$ is the open subvariety obtained by requiring  the (invertibility)  conditions 
\begin{align*}
&\det(\Id_{n_s}+X_sY_s)\neq 0,\,\,\, \det(\Id_{n_{s+1}}+Y_sX_s)\neq0, \quad s\in \Z_m\,, \\ 
&\det(\Id_{n_s}+W_{s,\alpha}V_{s,\alpha})\neq0, \,\,\, 1+V_{s,\alpha}W_{s,\alpha}\neq 0\,, \quad 
1\leq \alpha \leq d_s\,,\,\, s\in \Z_m\,.
\end{align*}  
Both spaces are endowed with an action of $\Gl(\nfat):=\CC^\times \times \big(\prod_{s\in \Z_m}\Gl(n_s) \big)$ through 
\begin{equation} \label{Eq:gact}
 (\lambda,g)\cdot (X_s,Y_s,W_{s,\alpha},V_{s,\alpha})=(g_s X_s g_{s+1}^{-1},g_{s+1}Y_s g_s^{-1},g_s W_{s,\alpha}\lambda^{-1},\lambda V_{s,\alpha} g_s^{-1})\,, 
\end{equation}
for $\lambda\in \CC^\times$ and $g=(g_s)_{s\in \Z_m}\in \prod_{s\in \Z_m}\Gl(n_s)$. 
The group $\CC^\times$ diagonally embedded into $\Gl(\nfat)$ by $\lambda \mapsto (\lambda,\lambda \Id_{n_s})_{s\in \Z_m}$ acts trivially at each point, 
so we may restrict our attention to the action of $\prod_{s\in \Z_m}\Gl(n_s)\simeq \Gl(\nfat)/\CC^{\times}$.   

Given the total orders \eqref{Eq:CyOrd} at the different vertices, we get that $\MM_{Q_\dfat,\nfat}^\bullet$ is a quasi-Hamiltonian $\Gl(\nfat)$-variety by following \ref{sss:MQV-constr}.
The moment map $\Phi:\MM_{Q_\dfat,\nfat}^\bullet\to \Gl(\nfat)$ given by \eqref{Eq:qHamMomap} can be decomposed  
as $\Phi=(\Phi_\infty,\Phi_s)_{s\in \Z_m}$, where $\Phi_\infty:\MM_{Q_\dfat,\nfat}^\bullet\to \CC^\times$ and $\Phi_s:\MM_{Q_\dfat,\nfat}^\bullet\to \Gl(n_s)$, $s\in \Z_m$, are defined at $p=(X_s,Y_s,W_{s,\alpha},V_{s,\alpha}) \in \MM_{Q_\dfat,\nfat}^\bullet$ by 
\begin{subequations}
 \begin{align}
 \Phi_\infty(p)=&\prod^{\longrightarrow}_{s\in \Z_m}\prod^{\longrightarrow}_{1\leqslant \alpha \leqslant d_s} (1+V_{s,\alpha} W_{s,\alpha})\,, \\
\Phi_s(p)=&  (\Id_{n_s}+X_s Y_s)(\Id_{n_s}+Y_{s-1}X_{s-1})^{-1} 
\prod^{\longrightarrow}_{1\leqslant \alpha \leqslant d_s} (\Id_{n_s}+W_{s,\alpha}V_{s,\alpha})^{-1}\,. 
 \end{align}
 \label{Eq:GeomMomap} 
\end{subequations}
In these expressions, the products are ordered from left to right with increasing indices.
The quasi-Poisson bracket can be obtained explicitly in terms of the regular functions returning the matrix entries 
\begin{equation} \label{Eq:MatEntries}
 (X_s)_{ij},\,\,(Y_s)_{ji},\,\, (V_{s,\alpha})_i,\,\, (W_{s,\alpha})_i,\quad 1\leq i\leq n_s,\,\, 1\leq j \leq n_{s+1},\quad 1\leq \alpha\leq d_s,\,\, s\in \Z_m\,,
\end{equation}
see Remark \ref{Rem:qPbracket}. It can be written for the matrices representing the arrows of the cyclic quiver as  
 \begin{subequations}
 \label{Eq:qPbrack1}
       \begin{align}
\br{(X_r)_{ij},(X_s)_{kl}}\,=\,&\frac12\delta_{(s,r-1)}\,\,  (X_{r-1}X_r)_{kj} \delta_{il}\, 
-\,\frac12\delta_{(s,r+1)} \,\, \delta_{kj} (X_r X_{r+1})_{il}\,, \label{cyida} \\
\br{(Y_r)_{ij},(Y_s)_{kl}}\,=\,& \frac12\delta_{(s,r-1)}\,\, \delta_{kj} (Y_rY_{r-1})_{il} \,  
 -\,\frac12\delta_{(s,r+1)}\,\, (Y_{r+1}Y_r)_{kj} \delta_{il}\, ,\label{cyida'}\\
 \br{(X_r)_{ij},(Y_s)_{kl}}\,=\,&\delta_{sr}\left( \delta_{kj}\delta_{il}
 +\frac{1}{2} (Y_rX_r)_{kj} \delta_{il} +\frac{1}{2} \delta_{kj} (X_rY_r)_{il} \right) \nonumber \\
 &-\frac12\delta_{(s,r-1)}\,\,(X_r)_{kj}  (Y_{r-1})_{il}\,+\frac12\delta_{(s,r+1)}\,\,(Y_{r+1})_{kj} (X_r)_{il}\,\,; \label{cyidb} 
        \end{align}
  \end{subequations}
on such entries with the functions $(V_{s,\alpha})_i$ and $(W_{s,\alpha})_i$, the quasi-Poisson bracket is given  by 
 \begin{subequations}
 \label{Eq:qPbrack2}
       \begin{align}
\br{(X_r)_{ij}, (W_{s,\alpha})_k}\,=\,& \frac12 \delta_{(s,r+1)}\, \delta_{kj} (X_r W_{r+1,\alpha})_i
-\frac12 \delta_{rs}\, (X_r)_{kj} (W_{r,\alpha})_i\,,\label{cyidd}\\
\br{(X_r)_{ij}, (V_{s,\alpha})_l}\,=\,& \frac12 \delta_{rs}\, (V_{r,\alpha} X_r)_j \delta_{il}
-\frac12 \delta_{(s,r+1)}\, (V_{r+1,\alpha})_j (X_r)_{il}\,,\label{cyidd'}\\
\br{(Y_r)_{ij}, (W_{s,\alpha})_k}\,=\,& \frac12 \delta_{rs}\,\delta_{kj} (Y_rW_{r,\alpha})_i 
-\frac12 \delta_{(s,r+1)}\, (Y_r)_{kj} (W_{r+1,\alpha})_i\,,\label{cyide} \\
\br{(Y_r)_{ij}, (V_{s,\alpha})_l}\,=\,& \frac12\delta_{(s,r+1)} (V_{r+1,\alpha} Y_r)_j \delta_{il}
-\frac12 \delta_{rs} (V_{r,\alpha})_j (Y_r)_{il}\,;\label{cyide'} 
       \end{align}
\end{subequations}
when restricted to the  functions $(V_{s,\alpha})_i$ and $(W_{s,\alpha})_i$, we finally get 
\begin{subequations}
 \label{Eq:qPbrack3}
       \begin{align}
\br{(V_{s,\alpha})_j,(V_{r,\beta})_l}\,=\,&-\frac12\, o(s,r)(V_{s,\alpha})_j (V_{r,\beta})_l \nonumber \\
&-\frac12 \,\delta_{sr}o(\alpha,\beta) \big(  (V_{r,\beta})_j (V_{s,\alpha})_l + (V_{s,\alpha})_j (V_{r,\beta})_l \big)\,, \label{cyidv}\\
\br{(W_{s,\alpha})_{i},(W_{r,\beta})_k}\,=\,&-\frac12\, o(s,r) (W_{r,\beta})_k (W_{s,\alpha})_i \nonumber \\
&-\frac12 \,\delta_{sr}o(\alpha,\beta) \big(  (W_{r,\beta})_k (W_{s,\alpha})_i + (W_{s,\alpha})_k (W_{r,\beta})_i \big)\,, \label{cyidw}\\
\br{(V_{s,\alpha})_j,(W_{r,\beta})_k}=\,& \frac12 o(s,r)\, (W_{r,\beta})_k (V_{s,\alpha})_j  \nonumber \\ 
&+ \delta_{sr}\delta_{\alpha \beta}\left(  \delta_{kj} +\frac12 (W_{r,\beta})_k (V_{s,\alpha})_j + \frac12 \delta_{kj} (V_{s,\alpha} W_{r,\beta}) \right) \nonumber \\
\,& + \frac12 \, \delta_{sr}o(\alpha,\beta) 
\big( \delta_{kj} (V_{s,\alpha} W_{r,\beta}) + (W_{r,\beta})_k (V_{s,\alpha})_j  \big) \,.\label{cyidu}
	\end{align}
\end{subequations}
In the last three expressions, we have introduced the \emph{ordering function} $o(-,-)$ which is defined as follows. If $s,r\in \Z_m$ are identified with their representatives in $\{0,\ldots,m-1\}$, then $o(s,r)=\operatorname{sgn}(r-s)$ is $+1$ if $s<r$, $-1$ if $s>r$, and $0$ if $s=r$. In the same way, we have denoted by $\delta_{sr}o(\alpha,\beta)$ the element $\operatorname{sgn}(\beta-\alpha)$ for $\alpha,\beta\in \{1,\ldots,d_s\}$ (it is taken to be zero if $d_s=0$ or $s\neq r$).

\begin{rem} \label{Rem:DBr}
 The reader may notice that the quasi-Poisson bracket written in terms of the regular functions \eqref{Eq:MatEntries} is always of the general form
 \begin{equation} \label{Eq:dbr}
  \br{A_{ij},B_{kl}}=\sum_{\mu\in J} (C_\mu)_{kj} (D_\mu)_{il}\,,
 \end{equation}
for some matrix-valued functions $A,B,C_\mu,D_\mu$ and a finite index set $J$. For calculations, it is therefore convenient to introduce the notation 
\begin{equation} \label{Eq:dbr-bis}
 \dgal{A,B}:=\sum_{\mu\in J} C_\mu \otimes D_\mu\,,
\end{equation} 
where the left-hand side of \eqref{Eq:dbr-bis} is regarded as a tensor product of two matrix-valued functions  whose entries are given by the left-hand side of \eqref{Eq:dbr};  
we recover \eqref{Eq:dbr} from $\br{A_{ij},B_{kl}}=\dgal{A,B}_{kj,il}$. This line of thoughts led to Van den Bergh's definition\footnote{For a slightly longer motivation of double brackets along the same lines, we refer to \cite[\S~5.1]{FG}.} of double brackets \cite{VdB1}. Let us simply note that, using the Leibniz rule, the quasi-Poisson bracket will always be of the form \eqref{Eq:dbr} for arbitrary matrices $A,B$ being products of the matrices \eqref{Eq:CycMQV-paraMat} parametrising $\MM_{Q_\dfat,\nfat}$. As noticed in \cite{VdB1}, this implies that 
\begin{equation} \label{Eq:br-traces}
 \br{\tr(A),B_{kl}}=\sum_{\mu\in J} (C_\mu D_\mu)_{kl}\,, \quad \br{\tr(A),\tr(B) }=\sum_{\mu\in J} \tr(C_\mu D_\mu)\,.
\end{equation}

\end{rem}

We can form the multiplicative quiver variety $\NN(Q_\dfat,<,\nfat,\qfat)$ as the GIT quotient for the $\Gl(\nfat)$ action \eqref{Eq:gact} on
the subspace of $\MM_{Q_\dfat,\nfat}^\bullet$ obtained by fixing the moment map \eqref{Eq:GeomMomap} to $(q_\infty,q_s \Id_{n_s})_{s\in \Z_m}$. The last condition is equivalent to 
the $m$ relations 
\begin{equation}
  (\Id_{n_s}+X_s Y_s)(\Id_{n_s}+Y_{s-1}X_{s-1})^{-1}=q_s \prod^{\longleftarrow}_{1\leqslant \alpha \leqslant d_s} (\Id_{n_s}+W_{s,\alpha}V_{s,\alpha})\,, \quad s \in \Z_m \,,
\label{GeomCys} 
\end{equation}
where the products on the right-hand sides are ordered from right to left with increasing indices. (We require that such a product is equal to $\Id_{n_s}$ when $d_s=0$.) The relation $\Phi_\infty=q_\infty$ just follows from the $m$ equalities \eqref{GeomCys} by taking determinant. The space $\NN(Q_\dfat,<,\nfat,\qfat)$ is equipped with a Poisson bracket, since it is constructed as a reduced Poisson variety following Proposition \ref{Pr:qPred}.

Finally, we can give a condition on $\qfat$ so that the multiplicative quiver variety is smooth. Introduce the following constants 
\begin{equation} \label{Eq:tparam}
  t_s:=\prod_{0\leq \tilde{s} \leq s} q_{\tilde{s}}\,,\,\, s=0,\ldots,m-1\,, \quad t:=t_{m-1}\,, \quad t_{-1}:=1\,,
\end{equation}
where we use the identification of $\Z_m$ with $\{0,\ldots,m-1\}$, noting however that $t_{m-1}\neq t_{-1}$. We say that an $m$-tuple $\tfat=(t_s)_{s\in \Z_m}\in (\CC^\times)^{\Z_m}$ is \emph{regular} whenever $t_{s_1}^{-1}t_{s_2}\neq t^k$ for any $k\in \Z$ and $-1 \leq s_1 < s_2 \leq m-1$ with $(s_1,s_2)\neq(-1,m-1)$, or whenever $t^k \neq 1$ for any $k\in \Z^\times$. This is equivalent to requiring that $\prod_{s\in \Z_m}q_s^{\alpha_s}\neq 1$ for any root $\alpha$ of the cyclic quiver on $m$ vertices seen as the Dynkin diagram of affine type $\widetilde{A}_{m-1}$. 
The next result is a direct adaptation of \cite[Proposition 4.5]{CF1}.
\begin{prop} \label{Pr:CyMQV}
  Assume that $\qfat$ is such that $\tfat$ defined through \eqref{Eq:tparam} is regular. 
  Then, provided that it is not empty, the multiplicative quiver variety $\NN(Q_\dfat,<,\nfat,\qfat)$ is a smooth Poisson variety of dimension $2\sum_{s\in \Z_m}n_s(n_{s+1}+d_s-n_s)$.
\end{prop}
Note that the variety $\MM_{Q_\dfat,\nfat}^\bullet$ has dimension $2\sum_{s\in \Z_m}n_s(n_{s+1}+d_s)$. 
If the multiplicative quiver variety is not empty and $\tfat$ regular, we have  $\sum_{s\in \Z_m}n_s^2$ independent relations \eqref{GeomCys} and by \cite[Theorem 1.10]{CBS} a free action of $\prod_{s\in \Z_m} \Gl(n_s)$ on the subvariety defined by \eqref{GeomCys}. Hence, simple dimension counting yields the dimension stated in Proposition \ref{Pr:CyMQV}.  
It is important to remark that the Poisson structure is non-degenerate by \ref{sss:MQV-constr}, and that $\tfat$ is regular for a generic choice of $\qfat$ satisfying \eqref{Eq:Cond-q}. 

\subsubsection{The case $(1,n\delta)$} \label{sss:Dim1n}

The varieties $\NN(Q_\dfat,<,\nfat,\qfat)$ described in \ref{sss:CycMQV-gen} have been considered in relation to integrable systems in \cite{CF1} for the case $\dfat=(1,0,\ldots,0)$ and in \cite{F1} for $\dfat=(d,0,\ldots,0)$ with $d\geq2$. In those works, it was remarked that one can construct interesting local coordinates on the subvarieties where the matrices $(X_s)_{s\in \Z_m}$ are invertible. Hence, since this subspace is clearly empty if  $n_s \neq n_{s+1}$ for some $s\in \Z_m$, we restrict our attention to the case
\begin{equation} \label{Eq:dim-ndelta}
 \nfat=(1,n \delta)=(1,\underbrace{n,\ldots,n}_{m\text{ times}})\,, \quad n \geq 1\,.
\end{equation}
Here, $\delta$ stands for the imaginary root $(1,\ldots,1)$ of the cyclic quiver. We use the notation 
\begin{equation} \label{Eq:Cnqm}
 \Cnqm:=\NN(Q_\dfat,<,(1,n\delta),\qfat)\,,
\end{equation}
and, when no confusion can arise, we simply denote $\Cnqm$ as $\Cnm$. 
\begin{prop} \label{Pr:CyMQVbis}
  Assume that $\qfat$ is such that $\tfat$ defined through \eqref{Eq:tparam} is regular. 
  Then $\Cnqm$ is a smooth variety of dimension $2n |\dfat|$ for $|\dfat|:=\sum_{s\in \Z_m}d_s$. 
\end{prop}
\begin{proof}
 By Proposition \ref{Pr:CyMQV}, it suffices to show that $\Cnqm$ is not empty. But the closed subvariety of $\Cnqm$ obtained by setting 
 \begin{equation*}
  V_{s,\alpha}=(0,\ldots,0)\,,\quad W_{s,\alpha}=(0,\ldots,0)^T \,,
 \end{equation*}
for all $(s,\alpha)\neq (0,1)$ coincides with the multiplicative quiver variety associated with the quiver $Q_{d'}$ for $d'=(1,0,\ldots,0)$. The last variety is studied in  \cite{CF1}, where it is shown that it is not empty. 
\end{proof}

It will be convenient for the subsequent sections to introduce some notation regarding the matrices parametrising $\Cnqm$ and the master space $\MM_{Q_{\dfat},\nfat}^\bullet$. We let 
\begin{equation*}
 \VV_\infty=\CC\,, \quad \VV_s=\CC^n \,\, \text{ for }\, s\in \Z_m\,, 
\end{equation*}
and we note that, by picking suitable bases on these vector spaces, $\MM_{Q_{\dfat},\nfat}^\bullet$ is parametrised by the  $2m+2|\dfat|$ matrices 
\begin{equation} \label{Eq:MatParam}
 \begin{aligned}
 &X_s \in \Hom(\VV_{s+1},\VV_s)\simeq \gl(n)\,, \quad 
Y_s \in \Hom(\VV_s,\VV_{s+1})\simeq \gl(n)\,, \\ 
&V_{s,\alpha}\in \Hom(\VV_s,\VV_\infty)\simeq \Mat(1\times n,\CC)\,, \quad
W_{s,\alpha}\in \Hom(\VV_\infty,\VV_s)\simeq\Mat(n\times 1,\CC)\,,
 \end{aligned} 
\end{equation}
subject to the condition that the following elements are invertible for all possible indices  
\begin{equation*}
  \Id_{\VV_s}+X_sY_s,\,\, \Id_{\VV_{s}}+Y_{s-1}X_{s-1},\,\ \Id_{\VV_s}+W_{s,\alpha}V_{s,\alpha}\in \End(\VV_s)\,, \quad 
  1+V_{s,\alpha}W_{s,\alpha}\in \End(\VV_\infty)\,.
\end{equation*}
Next, we form $\VV_\cyc=\oplus_{s\in \Z_m} \VV_s$ and let $1_\cyc=\Id_{\VV_{\cyc}}$. We can naturally see the matrices 
$(X_s,Y_s,V_{s,\alpha},W_{s,\alpha})$ as endomorphisms of $\VV_\infty\oplus\VV_\cyc$. Then, the action \eqref{Eq:gact} on  $\MM_{Q_{\dfat},\nfat}^\bullet$ is simply given by conjugation when realised through the diagonal embedding 
\begin{equation*}
 \CC^\times \times \Gl(n)^m \hookrightarrow \Gl(nm+1)\simeq \Gl(\VV_\infty\oplus\VV_\cyc)\,.
\end{equation*}
 We introduce the important matrices 
\begin{equation} \label{Eq:XYend}
 X:=\sum_{s\in \Z_m}X_s,\,\, Y:=\sum_{s\in \Z_m}Y_s,\,\, 1_\cyc+XY:=\sum_{s\in \Z_m} (\Id_{\VV_s}+X_sY_s)\, \in \End(\VV_\cyc)\simeq \gl(nm)\,.
\end{equation}
They can be seen as $m\times m$ block matrices with blocks of size $n\times n$ as 
\begin{align*}
X=\left( 
\begin{array}{ccccc}
 0_{n\times n}&X_0&0_{n\times n}&\cdots &0_{n\times n}\\
\vdots&0_{n\times n}&X_1&\ddots&\vdots \\
 \vdots&&\ddots&\ddots&0_{n\times n}\\
0_{n\times n}&&&\ddots&X_{m-2}\\
 X_{m-1}&0_{n\times n}&\cdots&\cdots&0_{n\times n}
\end{array}
\right) \,, \quad 
Y=\left( 
\begin{array}{ccccc}
 0_{n\times n}&\cdots&\cdots&0_{n\times n}&Y_{m-1}\\
Y_0&0_{n\times n}&&&0_{n\times n} \\
0_{n\times n}&Y_1&\ddots&&\vdots\\
\vdots&\ddots&\ddots&\ddots&\vdots\\
0_{n\times n}&\cdots&0_{n\times n}&Y_{m-2}&0_{n\times n}
\end{array}
\right) \,.
\end{align*} 
In that way, $1_\cyc+XY$ is block diagonal.  
We use the same notation in $\Cnqm$. 

Finally, we consider the subvariety $\MM_{Q_{\dfat},\nfat}^{\circ}\subset \MM_{Q_{\dfat},\nfat}^\bullet$ defined by the condition that $X\in \End(\VV_\cyc)$ is invertible, i.e. $\det(X)\neq 0$. We can then form the matrices  
\begin{equation} \label{Eq:Zend}
 Z:=\sum_{s\in \Z_m}Z_s\, \in \End(\VV_\cyc)\,,\quad Z_s:=Y_s+X_s^{-1}\,\in \Hom(\VV_s,\VV_{s+1})\,,
\end{equation}
which, due to the invertibility of $X_s$ and $\Id_{\VV_s}+X_sY_s$, are also invertible. 
The moment map condition \eqref{GeomCys} written on $\MM_{Q_{\dfat},\nfat}^{\circ}$ amounts to the $m$ relations 
\begin{equation}
  X_s Z_s X_{s-1}^{-1}Z_{s-1}^{-1}=q_s \prod^{\longleftarrow}_{1\leqslant \alpha \leqslant d_s} (\Id_{n}+W_{s,\alpha}V_{s,\alpha})\,, \quad s \in \Z_m \,.
\label{GeomCys-v2} 
\end{equation}
We denote by $\Cnqm^\circ$ the open subvariety of $\Cnqm$ defined by the condition $\det(X)\neq 0$, 
which coincides with the reduced Poisson variety obtained from $\MM_{Q_{\dfat},\nfat}^{\circ}$ by quasi-Hamiltonian reduction.  
We can summarise in the following diagram the relation between the different spaces introduced in this subsection: 
\begin{center}
 \begin{tikzpicture}
\node  (UL) at (-3,1) {$\MM_{Q_{\dfat},\nfat}^{\circ}$};
\node  (UM) at (0,1) {$\MM_{Q_{\dfat},\nfat}^{\bullet}$};
\node  (UR) at (3,1) {$\MM_{Q_{\dfat},\nfat}$};
\node  (LL) at (-3,-1) {$\Cnqm^\circ$};
\node  (LM) at (0,-1) {$\Cnqm$};
\path[->,dashed] (UL) edge   (LL) ;  \path[->,dashed] (UM) edge   (LM) ;
\path[right hook->] (UL) edge   (UM) ;  \path[right hook->] (UM) edge   (UR) ; \path[right hook->] (LL) edge   (LM) ;
 \end{tikzpicture}
\end{center}
Here, vertical arrows correspond to performing quasi-Hamiltonian reduction. 

\begin{prop} \label{Pr:CyMQVter}
  Assume that $\qfat$ is such that $\tfat$ defined through \eqref{Eq:tparam} is regular. 
  Then $\Cnqm^\circ$ is a smooth variety of dimension $2n |\dfat|$. 
\end{prop}
\begin{proof}
 If $\Cnqm^\circ$ is not empty, it is an open subset of $\Cnqm$ and the result follows from Proposition \ref{Pr:CyMQVbis}. As in the proof of Proposition \ref{Pr:CyMQVbis}, we have that $\Cnqm^\circ$ is not empty because it contains the subvariety corresponding to $d'=(1,0,\ldots,0)$ which is itself not empty by \cite{CF1}. 
\end{proof}

\begin{rem}
Setting $\widehat{\dfat}:=(|\dfat|,0,\ldots,0)$, 
the morphism $\psi_+ : \MM_{Q_{\dfat},\nfat}^{\circ} \longrightarrow \MM_{Q_{\widehat{\dfat}},\nfat}^{\circ}$ defined from  
\begin{align*} 
&\psi_+^\ast(X_s)=X_s, \quad \psi_+^\ast(Z_s)=Z_s, \\ 
&\psi_+^\ast(V_{0,\alpha})=V_{0,\alpha}, \quad 
\psi_+^\ast(W_{0,\alpha})=W_{0,\alpha}, \quad 1 \leq \alpha \leq d_0, \\
&\psi_+^\ast(V_{0,d_0+\ldots+d_{s-1}+\alpha})=V_{s,\alpha}X_{s-1}^{-1}\cdots X_{0}^{-1}, \\
&\psi_+^\ast(W_{0,d_0+\ldots+d_{s-1}+\alpha})=X_{0}\cdots X_{s-1}W_{s,\alpha}, \quad 1 \leq \alpha \leq d_s,\,\,  s\neq 0, 
\end{align*}
is an isomorphism of affine varieties.  
However $\psi_+$ is \emph{not} a Poisson morphism, since it would then preserve the moment map while $\Phi_s \neq \psi_+^\ast(X_sZ_s X_{s-1}^{-1} Z_{s-1}^{-1})$ on $\MM_{Q_{\dfat},\nfat}^{\circ}$ for any $s>0$ with $d_s\neq 0$. 
In particular, $\psi_+$ does not descend to a morphism between the corresponding reduced varieties $\Cnqm^\circ$ and $\mathcal{C}_{n,m,\qfat,\widehat{\dfat}}^\circ$.  
This observation (and the local Poisson structure discussed in Section \ref{S:Loc}) suggests the following: two such multiplicative quiver varieties $\Cnqm^\circ$ and $\mathcal{C}_{n,m,\qfat,\dfat'}^\circ$ with $|\dfat|=|\dfat'|$ are Poisson isomorphic if and only if $\dfat=\dfat'$. We are not aware of any reference discussing this conjectural result\footnote{We are grateful to an anonymous referee for raising this question.}. 
\end{rem}

\section{Poisson subalgebras before reduction and dynamics}  \label{S:Subalg} 

In this section, we investigate the quasi-Hamiltonian varieties $\MM_{Q_{\dfat},\nfat}^{\bullet}$ and $\MM_{Q_{\dfat},\nfat}^{\circ}$ defined in \ref{ss:CycMQV}, and we use the notation introduced there. We assume that $\nfat=(1,n\delta)$ as in \ref{sss:Dim1n}, though the results that do not require to work in $\MM_{Q_{\dfat},\nfat}^{\circ}$ (which is the subvariety of $\MM_{Q_{\dfat},\nfat}^{\bullet}$ defined by the condition that $X:=\sum_{s\in\Z_m} X_s$ is invertible) can be stated for arbitrary $\nfat=(1,n_0,\ldots,n_{m-1})$. 
Since we do not consider the reduced Poisson varieties obtained by quasi-Hamiltonian reduction in this section, it is not necessary to impose the regularity condition on $\qfat$ used in Proposition \ref{Pr:CyMQV}. 
Our goal is to study several Poisson algebras generated by regular $\Gl(\nfat)$-invariant global functions, which we will use to construct integrable systems after quasi-Hamiltonian reduction in Section \ref{S:Int}.

\begin{rem}[Spin indices] \label{Rem:Convention}
 When we consider a pair of indices $(s,\alpha)$, for example as index of $V_{s,\alpha}$ or $W_{s,\alpha}$, we assume that $s\in \Z_m$ and that $\alpha$ ranges over the set $\{1,\ldots,d_s\}$. We omit such labels when $d_s=0$. 
We call those pairs $(s,\alpha)$ the \emph{(admissible) spin indices}. 
Unless otherwise stated, we reserve $r,s,p,q$ for indices ranging over $\Z_m$, while the greek letters $\alpha,\beta,\gamma,\epsilon$ will appear as the second index in a pair of spin indices.
\end{rem}


\subsection{Poisson subalgebras with small centre} \label{ss:PoiDyn}

We construct several subalgebras of invariant functions on $\MM_{Q_{\dfat},\nfat}^{\bullet}$ which are closed under the quasi-Poisson bracket, which are therefore Poisson algebras due to their invariance. We give explicit formulas for the Hamiltonian flows associated with functions which are central with respect to the Poisson bracket.

\subsubsection{Construction of  Poisson subalgebras}

Let $\rL\in \End(\VV_\cyc)$ be one of the matrices $X,Y$ or $Z$, defined in \eqref{Eq:XYend} or \eqref{Eq:Zend} respectively. We introduce the commutative algebras
\begin{subequations}
\begin{align}
  \HHl=&\CC[\tr(\rL^k)\mid k\in\N]\,, \label{Eq:Palg-HL}\\ 
 \IIl=&\CC[\tr(\rL^k),\tr(W_{s,\alpha}V_{r,\beta}\rL^k)\mid k\in\N,\, r,s\in \Z_m,\, 1\leq \alpha\leq d_s,\, 1\leq \beta\leq d_r]\,, \label{Eq:Palg-IL}
\end{align}
\end{subequations}
and remark that $\HHl\subset \IIl$. 
Furthermore, these are subalgebras of the algebra of regular functions on $\MM_{Q_{\dfat},\nfat}^{\bullet}$ 
(or $\MM_{Q_{\dfat},\nfat}^{\circ}$ for $\rL=Z$), whose elements are $\Gl(\nfat)$-invariant for the regular action \eqref{Eq:gact}. 
Both algebras are finitely generated due to the Cayley-Hamilton theorem for $\rL$. 
Note also that some generators of $\HHl$ and $\IIl$ trivially vanish. Indeed, when $k$ is not divisible by $m$, we get that $\tr (\rL^k)=0$ since 
\begin{equation} \label{Eq:ThetaL}
\Id_{\VV_s}\circ \rL=\rL\circ \Id_{\VV_{s+\theta(\rL)}}\,, \quad \text{ for }\,\, \theta(\rL)=\left\{ 
\begin{array}{cl}
 +1&\text{ if }\rL=X\,, \\ -1&\text{ if }\rL=Y,Z\,,
\end{array} \right.
\end{equation}
in view of \eqref{Eq:MatParam} and \eqref{Eq:XYend}--\eqref{Eq:Zend}. 
We get in the same way that $\tr(W_{s,\alpha}V_{r,\beta}\rL^k)=V_{r,\beta}\rL^k W_{s,\alpha}$ identically vanishes if 
$k\theta(\rL)$ is not equal to $s-r$ modulo $m$.

\begin{prop} \label{Pr:DgIScy}
Consider $\MM_{Q_{\dfat},\nfat}^{\bullet}$ (or $\MM_{Q_{\dfat},\nfat}^{\circ}$ for $\rL=Z$) endowed with the quasi-Poisson bracket given by  \eqref{Eq:qPbrack1}--\eqref{Eq:qPbrack3}. 
Then the following results hold: 
\begin{enumerate}
  \item $\HHl$ is an abelian Poisson algebra; 
\item  For any $k,l \in \N$ and for any admissible spin indices $(s,\alpha)$ and $(r,\beta)$, 
$$\br{\tr (\rL^k),\tr (W_{s,\alpha} V_{r,\beta} \rL^l)}=0\,;$$ 
\item The algebra $\IIl$ is a Poisson algebra; 
\item $\HHl$ lies in the Poisson centre of $\IIl$.
\end{enumerate}
\end{prop}
\begin{proof}
The regular functions $\tr(\rL^k),\tr(W_{s,\alpha}V_{r,\beta}\rL^k)$ are traces of endomorphisms of $\VV_\cyc$, hence they are $\Gl(\nfat)$-invariant. Thus, the Jacobi identity is satisfied for any 3 such functions. Indeed, it is easy to see that the right-hand side of \eqref{Eq:JacPhi} vanishes if it is evaluated on invariant functions. So if we can show that $\HHl$ and $\IIl$ are closed under the quasi-Poisson bracket, they will be Poisson algebras. 

For part (1), it suffices to prove that $\br{\tr(\rL^k),\tr(\rL^l)}=0$ for $k,l$ divisible by $m$, since the equality is trivially satisfied otherwise. 
If $\rL=X$ and $k_0=\frac{k}{m}\in \N$, we note that $\tr(\rL^k)=m\,\tr((X_0\ldots X_{m-1})^{k_0})$. Thus, for $k=k_0m$ and $l=l_0m$ with $k_0,l_0\in \N$ we have 
\begin{align*}
\br{\tr(\rL^k),\tr(\rL^l)}=m^2kl \sum_{i,j,i',j'=1}^n 
(X_0\ldots X_{m-1})^{k_0-1}_{ji} (X_0\ldots X_{m-1})^{l_0-1}_{j'i'}
\br{(X_0\ldots X_{m-1})_{ij},(X_0\ldots X_{m-1})_{i'j'}}\,.
\end{align*}
The quasi-Poisson bracket that appears on the right-hand side can be computed  using \eqref{cyida} as 
\begin{align*}
&\sum_{a,b,a',b'=1}^n 
\sum_{r,s=0}^{m-1}
(X_0\ldots X_{r-1})_{ia} (X_0\ldots X_{s-1})_{i'a'}
\br{(X_r)_{ab},(X_s)_{a'b'}}
(X_{r+1}\ldots X_{m-1})_{bj},(X_{s+1}\ldots X_{m-1})_{b'j'} \\
&=\frac12 \delta_{ij'}(X_0\ldots X_{m-1})^2_{i'j}
-\frac12 (X_0\ldots X_{m-1})^2_{ij'} \delta_{i'j}
\,.
\end{align*}
We then deduce the desired identity. 
If $\rL=Y$ and $k_0=\frac{k}{m}\in \N$, we note that $\tr(\rL^k)=m\,\tr((Y_{m-1}\ldots Y_{0})^{k_0})$ and we get the same result from \eqref{cyida'}. We proceed also in a similar way if $\rL=Z$ because \eqref{cyida'} holds if we use the matrices $(Z_s)_{s\in \Z_m}$ instead of $(Y_s)_{s\in \Z_m}$.

For part (2), we prove the case $\rL=X$. The identity is trivial except if $k=k_0m$ and $l=l_0m+s-r$ for $k_0,l_0\in \N$. In that case, we  expand $\br{\tr (\rL^k),\tr (W_{s,\alpha} V_{r,\beta} \rL^l)}$ as 
\begin{equation*}
\Big\{m\tr\big((X_0\ldots X_{m-1})^{k_0}\big), \tr\big(W_{s,\alpha}V_{r,\beta} X_r\ldots X_{m-1}(X_0\ldots X_{m-1})^{l_0}X_0\ldots X_{s-1}\big) \Big\}\,.
\end{equation*}
Using \eqref{cyida} and \eqref{cyidd}--\eqref{cyidd'}, we can get that this equals zero.  The cases $\rL=Y,Z$ are obtained in a similar way. 

For part (3), we have to show that 
\begin{equation} \label{Eq:DgIScy-pf3}
 \br{\tr (W_{p,\gamma} V_{q,\epsilon}  \rL^k),\tr (W_{s,\alpha} V_{r,\beta} \rL^l)} \in \IIl\,,
\end{equation}
for any indices (omitting an index $s\in \Z_m$ if $d_s=0$)
\begin{equation*}
 k,l\geq0,\quad p,q,s,r\in \Z_m,\quad  1\leq \gamma \leq d_p,\,\, 1\leq \epsilon \leq d_q,\,\, 1\leq \alpha \leq d_s,\,\, 1\leq \beta \leq d_r\,.
\end{equation*}
We follow the idea of the proof of \cite[Lemma 5.12]{F1}. 
 By \eqref{Eq:br-traces}, we must have  
\begin{equation} 
 \br{\tr (W_{p,\gamma} V_{q,\epsilon}  \rL^k),\tr (W_{s,\alpha} V_{r,\beta} \rL^l)} = \sum_{\mu\in J} \tr( K_\mu)\,,
\end{equation}
for some matrices $K_\mu$ with $J$ a finite set. By inspection,  the quasi-Poisson bracket of any two entries of the matrices $(\rL, W_{s,\alpha},V_{r,\beta})$ can again be expressed in terms of entries of those matrices, and it is at most quadratic. In particular,  $K_\mu$ is a product of such matrices. These remarks imply that $\tr(K_\mu)$ takes one of the following forms 
\begin{equation} \label{Eq:TrickVW}
\begin{aligned}
\tr(K_\mu)=&\tr (\rL^{k_0}W_{s_0,\alpha_0} V_{r_0,\beta_0}  \rL^{k_1})=\tr (W_{s_0,\alpha_0} V_{r_0,\beta_0}  \rL^{k_0+k_1})\in \IIl\,, \quad \text{ or }\\
\tr(K_\mu)=&\tr (\rL^{k_0}W_{s_0,\alpha_0} V_{r_0,\beta_0}  \rL^{k_1}W_{s_1,\alpha_1} V_{r_1,\beta_1}  \rL^{k_2}) \\
=& (V_{r_0,\beta_0}  \rL^{k_1}W_{s_1,\alpha_1}) (V_{r_1,\beta_1}  \rL^{k_2}\rL^{k_0}W_{s_0,\alpha_0}) \\
=&\tr(W_{s_1,\alpha_1} V_{r_0,\beta_0}  \rL^{k_1}) \, \tr( \rL^{k_0+k_2}W_{s_0,\alpha_0}V_{r_1,\beta_1} )\in \IIl\,.
\end{aligned}
\end{equation}
This yields that \eqref{Eq:DgIScy-pf3} holds\footnote{The precise form of the bracket can be obtained from \cite[Lemma 3.2.8]{Fth}.}. 

Part (4) follows from the previous three results. 
\end{proof}

\begin{rem}
 If we work on the subvariety of $\MM_{Q_{\dfat},\nfat}^{\bullet}$ defined by the condition that $Y\in \End(\VV_\cyc)$ is invertible, we can also derive the above result for $\rL=X+Y^{-1}$ with $\theta(\rL)=+1$.  
We will omit to discuss any subsequent result when $\rL=X+Y^{-1}$, and we leave that case to the interested reader.  
\end{rem}

\subsubsection{Construction of another Poisson subalgebra}

We adapt the previous subsection to the case of the matrices $\Id_{\VV_s}+X_sY_s$, $s\in \Z_m$. 
We introduce the commutative algebras
\begin{subequations}
\begin{align}
  \HH_+=&\CC[\tr((\Id_{\VV_s}+X_sY_s)^k)\mid s\in \Z_m,\, k\in \N]\,, \label{Eq:Palg-H+}\\ 
 \II_+=&\CC[\tr((\Id_{\VV_s}+X_sY_s)^k),\tr(W_{p,\alpha}V_{q,\beta}(\Id_{\VV_s}+X_sY_s)^k)\mid k\in \N,\, p,q,s\in \Z_m,\, 1\leq \alpha\leq d_p,\, 1\leq \beta\leq d_q]\,. \label{Eq:Palg-I+}
\end{align}
\end{subequations}
Again, we have that $\HH_+\subset \II_+$ are $\Gl(\nfat)$-invariant subalgebras of the algebra of regular global functions, and they are finitely generated due to the Cayley-Hamilton theorem for the $m$ matrices $\Id_{\VV_s}+X_sY_s$.  
In fact, since $\det(\Id_{\VV_s}+X_sY_s)\neq 0$, this theorem  also implies that we can allow the exponent $k\geq 0$ to be negative in the definitions of $\HH_+,\II_+$ (this is also true for $\HHl,\IIl$ when $\rL=X,Z$ is restricted to $\MM_{Q_{\dfat},\nfat}^{\circ}$).
We easily see that $\tr(W_{p,\alpha}V_{q,\beta}(\Id_{\VV_s}+X_sY_s)^k)$ identically vanishes if 
$p,q$ are distinct from $s$. 
\begin{prop} \label{Pr:DgIScy-1XY}
Consider $\MM_{Q_{\dfat},\nfat}^{\bullet}$ endowed with the quasi-Poisson bracket given by  \eqref{Eq:qPbrack1}--\eqref{Eq:qPbrack3}. 
Then the following results hold: 
\begin{enumerate}
  \item $\HH_+$ is an abelian Poisson algebra; 
\item  For any $k,l \in \N$, $s,r\in \Z_m$, and for any admissible spin indices $(p,\alpha)$ and $(q,\beta)$, 
\begin{equation*}
  \br{\tr (\Id_{\VV_s}+X_sY_s)^k,\tr \big(W_{p,\alpha} V_{q,\beta} (\Id_{\VV_r}+X_rY_r)^l\big)}=0\, ;
\end{equation*}
\item The algebra $\II_+$ is a Poisson algebra; 
\item $\HH_+$ lies in the Poisson centre of $\II_+$.
\end{enumerate}
\end{prop}
\begin{proof}
 It suffices to adapt the proof of Proposition \ref{Pr:DgIScy}. It is helpful to note the following identities
  \begin{subequations}
 \label{Eq:qPbrack-1XY}
       \begin{align}
\br{(\Id_{\VV_r}+X_rY_r)_{ij},(\Id_{\VV_s}+X_sY_s)_{kl}}\,=\,&\frac12\delta_{sr}\,
\left( (\Id_{\VV_s}+X_sY_s)^2_{kj} \delta_{il}- \delta_{kj} (\Id_{\VV_s}+X_sY_s)^2_{il}\right)\,, \label{cyid-1XYa} \\
\br{(\Id_{\VV_r}+X_rY_r)_{ij},(V_{s,\alpha})_{l}}\,=\,&\frac12\delta_{sr}\,
\left( (V_{s,\alpha}(\Id_{\VV_r}+X_rY_r))_{j} \delta_{il}- (V_{s,\alpha})_{j} (\Id_{\VV_r}+X_rY_r)_{il}\right)\,, \label{cyid-1XYv}\\
\br{(\Id_{\VV_r}+X_rY_r)_{ij},(W_{s,\alpha})_{l}}\,=\,&\frac12\delta_{sr}\,
\left( \delta_{kj} ((\Id_{\VV_r}+X_rY_r) W_{s,\alpha})_{i}- (\Id_{\VV_r}+X_rY_r)_{kj} (W_{s,\alpha})_{i}\right)\,, \label{cyid-1XYw}
        \end{align}
  \end{subequations}
which can be derived from \eqref{Eq:qPbrack1} and \eqref{Eq:qPbrack2}.
\end{proof}

\subsubsection{Dynamics} \label{sss:PoiDyn-flow}

Consider one of the Poisson algebras $\HHl$ or $\HH_+$ defined previously, which is abelian by Proposition \ref{Pr:DgIScy} or \ref{Pr:DgIScy-1XY}. 
Each generator $F$ of the algebra defines a Hamiltonian vector field $\br{F,-}$ on $\MM_{Q_{\dfat},\nfat}^{\bullet}$ using the quasi-Poisson bracket. 
Since the abelian Poisson structure of the algebra is obtained by restriction of the quasi-Poisson bracket,  
the flows of these Hamiltonian vector fields are commuting on $\MM_{Q_{\dfat},\nfat}^{\bullet}$. We will write them explicitly in the remainder of this subsection. 

We start with the case $\rL=Y$. We will repeatedly use the notation \eqref{Eq:XYend}.
\begin{prop} \label{Pr:floYcy}
Given the initial condition in $\MM_{Q_{\dfat},\nfat}^{\bullet}$ 
\begin{equation*}
 X(0):=\sum_{s\in\Z_m} X_s(0),\quad Y(0):=\sum_{s\in\Z_m}Y_s(0),\quad V_{s,\alpha}(0),\quad W_{s,\alpha}(0)\,,
\end{equation*}
the flow at time $t$ defined by the Hamiltonian function $\frac1k \tr Y^k$ for $k=k_0m$, $k_0\geq1$, is given by 
\begin{equation} \label{Eq:floYcy}
  \begin{aligned}
       X(t)=& X(0)\, \exp(-t Y(0)^k) + Y(0)^{-1}[\exp(-t Y(0)^k)-1_{\cyc}]\,, \\
Y(t)=& Y(0)\,, \quad  W_{s,\alpha}(t)= W_{s,\alpha}(0)\,, \quad V_{s,\alpha}(t)= V_{s,\alpha}(0)\,.
  \end{aligned}
\end{equation}
\end{prop} 
Note that the expression for $X(t)$ can be written as a series using only non-negative powers of $Y(0)$, so that we do not need $Y(0)$ to be invertible and the flow is well-defined in $\MM_{Q_{\dfat},\nfat}^{\bullet}$. 
\begin{proof}
Using \eqref{Eq:qPbrack1} and \eqref{Eq:qPbrack2}, we note that the Hamiltonian vector field $d/dt:= \br{\tr Y^k/k,-}$ associated with $\frac1k \tr Y^k$ is such that 
\begin{equation} \label{Eq:floY-inf}
  \frac{d X}{dt}=-Y^{k-1}-XY^k\,, \quad \frac{d Y}{dt}=0\,,\quad   \frac{d V_{s,\alpha}}{dt}=0\,, \quad \frac{d W_{s,\alpha}}{dt}=0\,.
\end{equation}
(Here, we use the notation that for a matrix-valued function $A$, $dA/dt$ is the matrix-valued function with entry $(i,j)$ given by $dA_{ij}/dt=\frac{1}{k}\br{\tr Y^k,A_{ij}}$.) 
Given the initial condition, one verifies that the unique solution to the ordinary differential equation \eqref{Eq:floY-inf} is given by \eqref{Eq:floYcy}. 
\end{proof}

Next, we work with $\rL=Z$ in $\MM_{Q_{\dfat},\nfat}^{\circ}$. 
\begin{prop} \label{Pr:floZcy}
Given the initial condition in $\MM_{Q_{\dfat},\nfat}^{\circ}$
\begin{equation*}
 X(0):=\sum_{s\in\Z_m} X_s(0),\quad Z(0):=\sum_{s\in\Z_m}Z_s(0),\quad V_{s,\alpha}(0),\quad W_{s,\alpha}(0)\,,
\end{equation*}
the flow at time $t$ defined by the Hamiltonian function $\frac1k \tr Z^k$ for $k=k_0m$, $k_0\geq1$, is given by 
\begin{equation}  \label{Eq:floZcy}
 \begin{aligned}
    &X(t)= X(0)\, \exp(-t Z(0)^k)\,, \quad  Z(t)= Z(0)\,, \\
 &V_{s,\alpha}(t)= V_{s,\alpha}(0)\,, \quad  W_{s,\alpha}(t)= W_{s,\alpha}(0)\,.
 \end{aligned} 
\end{equation} 
\end{prop}
\begin{proof}
Using \eqref{Eq:qPbrack1} and \eqref{Eq:qPbrack2}, we can compute the quasi-Poisson brackets of the form $\br{(Z_{s})_{ij},-}$ where $Z_s=Y_s+X_s^{-1}$. We can then note that the Hamiltonian vector field $d/dt:= \br{\tr Z^k/k,-}$ associated with $\frac1k \tr Z^k$ is such that 
\begin{equation*}
  \frac{d X}{dt}=-XZ^k\,, \quad \frac{d Z}{dt}=0\,,\quad   
  \frac{d V_{s,\alpha}}{dt}=0\,, \quad \frac{d W_{s,\alpha}}{dt}=0\,,
\end{equation*}
and we can conclude similarly to the proof of Proposition \ref{Pr:floYcy}.
\end{proof}

The case $\rL=X$ can be obtained by symmetry from the previous two results, in $\MM_{Q_{\dfat},\nfat}^{\bullet}$ or $\MM_{Q_{\dfat},\nfat}^{\circ}$ respectively. 
\begin{prop} \label{Pr:floXcy}
Given the initial condition in $\MM_{Q_{\dfat},\nfat}^{\bullet}$ 
\begin{equation*}
 X(0):=\sum_{s\in\Z_m} X_s(0),\quad Y(0):=\sum_{s\in\Z_m}Y_s(0),\quad V_{s,\alpha}(0),\quad W_{s,\alpha}(0)\,,
\end{equation*}
the flow at time $t$ defined by the Hamiltonian function $\frac1k \tr X^k$ for $k=k_0m$, $k_0\geq1$, is given by 
\begin{equation} \label{Eq:floXcy-a}
  \begin{aligned}
Y(t)=& Y(0)\, \exp(t X(0)^k) + X(0)^{-1}[\exp(t X(0)^k)-1_{\cyc}]\,, \\
X(t)=& X(0)\,, \quad  W_{s,\alpha}(t)= W_{s,\alpha}(0)\,, \quad V_{s,\alpha}(t)= V_{s,\alpha}(0)\,.
  \end{aligned}
\end{equation}
If $X(0)$ is invertible, the flow can be restricted to $\MM_{Q_{\dfat},\nfat}^{\circ}$ where, upon introducing $Z(0):=Y(0)+X(0)^{-1}$, it takes the form 
\begin{equation} \label{Eq:floXcy-b}
  \begin{aligned}
    &X(t)= X(0)\,, \quad  Z(t)= Z(0)  \exp(t X(0)^k)\,, \\
 &V_{s,\alpha}(t)= V_{s,\alpha}(0)\,, \quad  W_{s,\alpha}(t)= W_{s,\alpha}(0)\,.
  \end{aligned}
\end{equation}
\end{prop}

Finally, we can consider the flows corresponding to the functions generating the abelian Poisson algebra $\HH_+$ defined by \eqref{Eq:Palg-H+}.
\begin{prop} \label{Pr:floTcy}
Given the initial condition in $\MM_{Q_{\dfat},\nfat}^{\bullet}$ 
\begin{equation*}
 X(0):=\sum_{s\in\Z_m} X_s(0),\quad Y(0):=\sum_{s\in\Z_m}Y_s(0),\quad V_{s,\alpha}(0),\quad W_{s,\alpha}(0)\,,
\end{equation*}
the flow at time $t$ defined by the Hamiltonian function $\frac1k \tr (\Id_{\VV_r}+X_rY_r)^k$ for $r\in \Z_m$ and $k\geq 1$ is given by 
\begin{equation}  \label{Eq:floTcy}
 \begin{aligned}
    X(t)=& \exp\big(-t [\Id_{\VV_r}+ X_r(0)Y_r(0)]^k\big)\, X(0) \,, \\
    Y(t)=& Y(0)\,\exp\big(t [\Id_{\VV_r}+ X_r(0)Y_r(0)]^k\big) \,, \\
 V_{s,\alpha}(t)=& V_{s,\alpha}(0)\,, \quad  W_{s,\alpha}(t)= W_{s,\alpha}(0)\,.
 \end{aligned} 
\end{equation} 
\end{prop}
\begin{proof}
Put $T_r:= (\Id_{\VV_r}+X_rY_r)$. 
Using \eqref{Eq:qPbrack1} and \eqref{Eq:qPbrack2}, we find the quasi-Poisson brackets  $\br{(T_r)_{ij},-}$. 
We can then note that the Hamiltonian vector field $d/dt:= \br{\tr T_r^k/k,-}$ associated with $\frac1k \tr T_r^k$ is such that 
\begin{equation*}
  \frac{d X}{dt}=-T_r^k X\,, \quad \frac{d Y}{dt}=Y T_r^k\,,\quad   \frac{d V_{s,\alpha}}{dt}=0\,, \quad \frac{d W_{s,\alpha}}{dt}=0\,.
\end{equation*}
In particular, $\frac{d T_r}{dt}=0$ yields $T_r(t)=\Id_{\VV_r}+X_r(0)Y_r(0)$ and we can conclude.
\end{proof}

Note that in \eqref{Eq:floTcy} only the submatrices $X_s(t)$ and $Y_s(t)$ with $s=r$ are not constant. If $X(0)$ is invertible, we can work in $\MM_{Q_{\dfat},\nfat}^{\circ}$ where the flow is completely determined by 
\begin{equation*}
    X(t)= \exp\big(-t T_r(0)^k\big)\, X(0) \,, \quad  T(t)= T(0)\,, \quad 
 V_{s,\alpha}(t)= V_{s,\alpha}(0)\,, \quad  W_{s,\alpha}(t)= W_{s,\alpha}(0)\,, 
\end{equation*}
for $T=\sum_{s\in \Z_m}T_s$ with $T_s:=\Id_{\VV_s}+X_sY_s$.

\begin{rem} \label{Rem:Dim}
As pointed out at the beginning of the section, the results that hold on $\MM_{Q_{\dfat},\nfat}^{\bullet}$ can be stated for any dimension vector $\nfat=(1,n_0,\ldots,n_{m-1})$.
That is, for $\rL=X,Y$, the algebra of functions $\IIl$ \eqref{Eq:Palg-IL} is a Poisson algebra  which contains $\HHl$ \eqref{Eq:Palg-HL} in its center, and we can explicitly write down the flows of the Hamiltonian vector fields associated with the generators of $\HHl$. These statements are also true for $\II_+$ \eqref{Eq:Palg-I+} with $\HH_+$ \eqref{Eq:Palg-H+}.  
We focused on the dimension vector $\nfat=(1,n\delta)$ because we can show in that case that these algebras descend to degenerately integrable systems on the corresponding multiplicative quiver varieties, cf. \ref{ss:DIS} (we can also relate them to generalised spin RS systems, cf. \ref{sss:SpinRS}).  
However, since flows can be integrated explicitly, this hints that degenerate integrability may also hold after reduction 
for \emph{any choice} of $\nfat$. It is an open problem to investigate this issue. 
A possible approach to this challenging question would be to consider reflection functors, in analogy with the work of Silantyev \cite{Si} in the additive case of quiver varieties.  
Establishing Liouville integrability of the families considered in the next \ref{ss:AbDyn} after reduction for arbitrary $\nfat$ is also an open problem.
\end{rem}


\subsection{Abelian Poisson subalgebras} \label{ss:AbDyn}

We construct several abelian Poisson subalgebras of invariant functions on $\MM_{Q_{\dfat},\nfat}^{\bullet}$ as in \ref{ss:PoiDyn}, and write the corresponding Hamiltonian flows. The construction of the algebras relies on a chain of subquivers of $\overline{Q}_\dfat$.  

\subsubsection{Embeddings of quivers}  \label{sss:Embed}

\begin{rem}[Order on spin indices] Recall from Remark \ref{Rem:Convention} that an admissible spin index is a pair $(s,\alpha)$ such that $s\in \Z_m$, $1\leq \alpha \leq d_s$, and we omit such pairs when $d_s=0$. 
Using the identification between $\Z_m$ and $\{0,\ldots,m-1\}$ introduced in \ref{sss:CycMQV-not}, we get a total order on the spin indices as follows. We set that $(s,\alpha)<(r,\beta)$ whenever $s<r$ as elements of $\Z_m\simeq \{0,\ldots,m-1\}$, or  when $\alpha< \beta$ if $s=r$.  

Let $\Ord(\dfat)=\{(s,\alpha) \mid s\in I,\,\,1\leq \alpha \leq d_s\}$ be the set of admissible spin indices, which is a totally ordered set by the above construction. This set has cardinality $|\dfat|:=\sum_{s\in \Z_m}d_s$. If we consider $\{1,\ldots,|\dfat|\}$ with its natural total order, there is a unique map $\rho:\{1,\ldots,|\dfat|\}\to \Ord(\dfat)$ preserving total orders on both sets. The map $\rho$ satisfies $\rho(1)=(0,1)$ and $\rho(|\dfat|)=(\bar{s},d_{\bar{s}})$, 
where $\bar{s}\in \Z_m$ is defined to be the largest element satisfying $d_{\bar{s}}\neq 0$.  
Furthermore, for any $b=1,\ldots,|\dfat|-1$, if $\rho(b)=(s,\alpha)$ we have 
\begin{equation*}
\rho(b+1)=\left\{ 
 \begin{array}{ll}
 (s,\alpha+1) & \text{ if } \alpha<d_s\,, \\
 (s_+,1)&\text{ if }\alpha=d_s \text{ and }s_+:=\min\{r\in \Z_m \mid r>s,\, d_r\neq 0\}\,.
 \end{array}
\right. 
\end{equation*}
\end{rem}

We use the order on the spin indices to define a subquiver  $\overline{Q}_b^{\res}$ of $\overline{Q}_\dfat$ for any $b\in \{0\}\cup \{1,\ldots,|\dfat|\}$ with the same vertex set. 
First, the arrow set of $\overline{Q}_0^{\res}$ is obtained by only considering the (doubled) cyclic quiver, i.e. $\overline{Q}_0^{\res}$ has $2m$ arrows  $x_s:s \to s+1$ and  $y_s:s+1 \to s$ for $s \in \Z_m$. 
Next, we get $\overline{Q}_1^{\res}$  by adding to $\overline{Q}_0^{\res}$ the two arrows 
\begin{equation*}
 v_{0,1}:\infty \to 1, \quad w_{0,1}: 1 \to \infty.
\end{equation*}
Then, we define recursively $\overline{Q}_{b}^{\res}$ from $\overline{Q}_{b-1}^{\res}$ by adding the two arrows 
\begin{equation*}
 v_{s,\alpha}:\infty \to s, \quad w_{s,\alpha}: s \to \infty, 
\end{equation*}
for $\rho(b)=(s,\alpha)$. We end up with $\overline{Q}_{|\dfat|}^{\res}=\overline{Q}_\dfat$.

For each subquiver $\overline{Q}^{\res}_b$, $b\in \{0,1,\ldots,|\dfat|\}$, we consider at each vertex the ordering obtained by restricting \eqref{CyOrd0}--\eqref{CyOrdinf}  to the arrows of $\overline{Q}^{\res}_b$. Then, as a simple application of \ref{sss:CycMQV-gen}, we can define the quasi-Hamiltonian $\Gl(\nfat)$-variety $\MM_{Q_b^{\res},\nfat}^{\bullet}$ described using $2m+2b$ matrices 
\begin{equation}
 X_s,\,\, Y_s,\,\, V_{r,\beta},\,\, W_{r,\beta},\quad \text{ with }s\in \Z_m \text{ and } (r,\beta)\leq \rho(b)\,;
\end{equation}
only the $2m$ matrices $(X_s,Y_s)_{s\in \Z_m}$ parametrise $\MM_{Q_0^{\res},\nfat}^\bullet$ when $b=0$. It is crucial for us to remark several consequences of these embeddings of quivers. To begin with, there is a chain of inclusions  
\begin{equation}\label{Eq:chainAlg}
\CC[\MM_{Q_0^{\res},\nfat}^{\bullet}] \hookrightarrow \ldots \hookrightarrow 
\CC[\MM_{Q_b^{\res},\nfat}^{\bullet}] \hookrightarrow \CC[\MM_{Q_{b+1}^{\res},\nfat}^{\bullet}]
\hookrightarrow \ldots \hookrightarrow \CC[\MM_{Q_{|\dfat|}^{\res},\nfat}^{\bullet}]=\CC[\MM_{Q_\dfat,\nfat}^{\bullet}] 
\,,
\end{equation}
where at the step  $\CC[\MM_{Q_b^{\res},\nfat}^{\bullet}] \hookrightarrow \CC[\MM_{Q_{b+1}^{\res},\nfat}^{\bullet}]$ we add the $2n$ generators $(V_{\rho(b+1)})_j,(W_{\rho(b+1)})_j$ with $1\leq j \leq n$. 
Geometrically, this is obtained by regarding $\MM_{Q_b^{\res},\nfat}^{\bullet}$ as the closed subvariety of $\MM_{Q_{b+1}^{\res},\nfat}^{\bullet}$ defined by the conditions $V_{\rho(b+1)}=0_{1\times n}$ and $W_{\rho(b+1)}=0_{n\times 1}$. 

For each $b\in \{0,1,\ldots,|\dfat|\}$, we can recursively define the moment map $\Phi^{(b)}:=(\Phi^{(b)}_\infty , \Phi^{(b)}_s)$ on $\MM_{Q_b^{\res},\nfat}^{\bullet}$, using the chain \eqref{Eq:chainAlg} 
by setting at $p=(X_s,Y_s,V_{r,\beta},W_{r,\beta})_{s\in \Z_m}^{(r,\beta)\leq \rho(b)}\in \MM_{Q_b^{\res},\nfat}^{\bullet}$
\begin{equation*} 
 \Phi_\infty^{(0)}(p)=1\,, \qquad 
\Phi_s^{(0)}(p)=  (\Id_{n_s}+X_s Y_s)(\Id_{n_s}+Y_{s-1}X_{s-1})^{-1}\,, \quad s \in \Z_m \,, 
\end{equation*}
and for any $1\leq c \leq b$  
\begin{subequations}  
 \begin{align}
\Phi^{(c)}_\infty(p) =&\Phi^{(c-1)}_\infty(p)\,  (1+V_{\rho(c)} W_{\rho(c)})\,, \\
\Phi^{(c)}_s(p)=& \left\{ 
\begin{array}{ll}
\Phi^{(c-1)}_s(p)\,  (\Id_{n_s}+W_{\rho(c)}V_{\rho(c)})^{-1} &\text{ if } \rho(c)=(s,\alpha) \text{ for some }\alpha\in \{1,\ldots,d_s\}\,, \\
\Phi^{(c-1)}_s(p)\,, &\text{ if }  \rho(c)=(r,\alpha)\text{ for some }r\in \Z_m\setminus\{s\}\,.
\end{array} \right.
 \end{align}
 \label{Eq:chainGeomMomap} 
\end{subequations}
Finally, because the total order on the set of arrows of $\overline{Q}_b^{\res}$ is induced by that for $\overline{Q}_\dfat$, we deduce that the embeddings in \eqref{Eq:chainAlg} respect the quasi-Poisson brackets. This means that the quasi-Poisson bracket on $\MM_{Q_b^{\res},\nfat}^{\bullet}$ 
is simply given by \eqref{Eq:qPbrack1} together with the identities in \eqref{Eq:qPbrack2}--\eqref{Eq:qPbrack3} with spin indices $(s,\alpha)$ or $(r,\beta)$ that are less or equal to $\rho(b)$. 
This leads us to the following result.

\begin{lem} \label{Lem:MomapChain}
 Fix $b\in \{0,1,\ldots,|\dfat|\}$. Then, for any $s,r\in \Z_m$, we have on $\MM_{Q_\dfat,\nfat}^{\bullet}$
 { \allowdisplaybreaks  
 \begin{subequations} \label{Eq:ChainMomapCond-1}
  \begin{align}
\br{(\Phi^{(b)}_s)_{ij},(X_r)_{kl}}=&\frac12 \delta_{r+1,s}\left((X_r)_{kj}(\Phi^{(b)}_s)_{il} + (X_r\Phi^{(b)}_s)_{kj}\delta_{il} \right) \nonumber \\
&-\frac12 \delta_{rs} \left(\delta_{kj} (\Phi^{(b)}_sX_r)_{il} + (\Phi^{(b)}_s)_{kj} (X_r)_{il} \right) \label{Eq:ChainMomapCond-s-X} \,, \\
\br{(\Phi^{(b)}_s)_{ij},(Y_r)_{kl}}=&\frac12 \delta_{rs}\left((Y_r)_{kj}(\Phi^{(b)}_s)_{il} + (Y_r\Phi^{(b)}_s)_{kj}\delta_{il} \right) \nonumber \\
&-\frac12 \delta_{r+1,s} \left(\delta_{kj} (\Phi^{(b)}_sY_r)_{il} + (\Phi^{(b)}_s)_{kj} (Y_r)_{il} \right) \label{Eq:ChainMomapCond-s-Y}  \,, \\
\br{\Phi^{(b)}_\infty,(X_r)_{kl}}=0\,, &\qquad \br{\Phi^{(b)}_\infty,(Y_r)_{kl}}=0\,,  
  \end{align}
 \end{subequations}
}  
while for any $s\in\Z_m$ and  $(r,\beta)\leq \rho(b)$, we have  on $\MM_{Q_\dfat,\nfat}^{\bullet}$ 
 \begin{subequations}  \label{Eq:ChainMomapCond-2}
  \begin{align}
\br{(\Phi^{(b)}_s)_{ij},(V_{r\beta})_{l}}=&\frac12 \delta_{rs}\left((V_{r,\beta})_{j}(\Phi^{(b)}_s)_{il} + (V_{r,\beta}\Phi^{(b)}_s)_{j}\delta_{il} \right)   \,, \\
\br{(\Phi^{(b)}_s)_{ij},(W_{r,\beta})_{k}}=&-\frac12 \delta_{rs} \left(\delta_{kj} (\Phi^{(b)}_sW_{r,\beta})_{i} + (\Phi^{(b)}_s)_{kj} (W_{r,\beta})_{i} \right) \,,\\
\br{\Phi^{(b)}_\infty,(V_{r,\beta})_{l}}=&- \Phi^{(b)}_\infty (V_{r,\beta})_{l} \,, \quad 
\br{\Phi^{(b)}_\infty,(W_{r,\beta})_{k}}= (W_{r,\beta})_{k} \Phi^{(b)}_\infty \,. 
  \end{align}
 \end{subequations}
Furthermore, if $c\in \{0,1,\ldots,|\dfat|\}$ is such that $c<b$, we have for any $r,s\in \Z_m$ 
\begin{equation} \label{Eq:ChainMomapCond-3}
\begin{aligned}
 \br{(\Phi^{(b)}_s)_{ij} , (\Phi^{(c)}_r)_{kl}} =&
\frac12 \delta_{rs} \left( (\Phi^{(c)}_r)_{kj}(\Phi^{(b)}_s)_{il} 
+ (\Phi^{(c)}_r \Phi^{(b)}_s)_{kj} \delta_{il}\right) \\
&-\frac12 \delta_{rs}  \left(  \delta_{kj} (\Phi^{(b)}_s \Phi^{(c)}_r)_{il}  + (\Phi^{(b)}_s)_{kj} (\Phi^{(c)}_r)_{il}\right)\,,
\end{aligned}
\end{equation}
and $ \br{\Phi^{(b)}_\infty , \Phi^{(c)}_\infty}=0$.
\end{lem}
\begin{proof}
 Since the embeddings in \eqref{Eq:chainAlg} respect the quasi-Poisson bracket, it suffices to prove the stated identities on $\MM_{Q_b^{\res},\nfat}^{\bullet}$. It is a consequence of \cite[Proposition 7.13.2]{VdB1} that these are equivalent to the moment map property \eqref{momapScheme} for $\Phi^{(b)}$. For clarity, let us nevertheless reproduce Van den Bergh's argument to derive \eqref{Eq:ChainMomapCond-s-X}. 
 
Using \eqref{momapScheme} with the evaluation function $F:=F_{s;ij}$ that returns the entry $(i,j)$ of the\footnote{We work in full generality with $\nfat=(1,n_0,\ldots,n_{m-1})$ and not just with the case of interest $\nfat=(1,n\delta)$. This has the advantage of making the dependence on $s\in \Z_m$ transparent.} 
 $\Gl(n_s)$ component of an element of $\Gl(\nfat)$ (hence $(\Phi^{(b)})^\ast F=(\Phi^{(b)}_s)_{ij}$), we can write 
\begin{equation}
 \br{(\Phi^{(b)}_s)_{ij},(X_r)_{kl}}=\frac12 \sum_{a\in A}(\Phi^{(b)})^\ast\big((f_a^L+f_a^R)(F_{s;ij})\big)\,\, (f^a)_M\big((X_r)_{kl}\big)\,.
\end{equation}
It is easy to see that $f_a^L(F)=0=f_a^R(F)$ unless $f_a\in\gl(n_s)$. Thus, it suffices to sum over $a=(a_1,a_2)$ for $1\leq a_1,a_2\leq n_s$ with the following bases that are dual under the trace pairing of $\gl(n_s)$:  
\begin{equation*}
 f_{(a_1,a_2)}=E_{a_1,a_2}\,, \quad f^{(a_1,a_2)}=E_{a_2,a_1}\,.
\end{equation*}
(Here $E_{a,a'}\in \gl(n_s)$ is the elementary matrix with only nonzero entry equal to $+1$ in position $(a,a')$.)
We have from \eqref{EqinfLR} that at a point $z=(z_\infty,z_0,\ldots,z_{m-1})\in \Gl(\nfat)$, 
\begin{equation*}
f_{(a_1,a_2)}^L(F_{s;ij})(z)=(z_s f_{(a_1,a_2)})_{ij}=\sum_{a_3=1}^{n_s} (z_s)_{ia_3} (f_{(a_1,a_2)})_{a_3 j}=\delta_{a_2j} (z_s)_{i a_1}\,,
\end{equation*}
which yields $f_{(a_1,a_2)}^L(F_{s;ij})=\delta_{a_2j}\, F_{s;i a_1}$ as an identity of regular functions on $\Gl(\nfat)$. 
We can obtain $f_{(a_1,a_2)}^R(F_{s;ij})=\delta_{i a_1}\, F_{s;a_2 j}$ in the same way. 
Therefore 
\begin{equation} \label{Eq:PfMomapChain}
 \br{(\Phi^{(b)}_s)_{ij},(X_r)_{kl}}=\frac12 \sum_{a_1,a_2=1}^{n_s}\big(\delta_{a_2j}\, (\Phi^{(b)}_s)_{i a_1} + \delta_{i a_1}\, (\Phi^{(b)}_s)_{a_2 j} \big)\,\, (f^{(a_1,a_2)})_M\big((X_r)_{kl}\big)\,.
\end{equation}
Using \eqref{EqinfVectM}, we can compute that 
\begin{equation*}
(f^{(a_1,a_2)})_M((X_r)_{kl})
=\delta_{s,r+1}(X_r f^{(a_1,a_2)})_{kl} - \delta_{rs} (f^{(a_1,a_2)} X_r)_{kl} 
=\delta_{s,r+1}(X_r)_{k a_2} \delta_{a_1 l} -   \delta_{rs} \delta_{k a_2}  (X_r)_{a_1 l}\,,
\end{equation*}
which can be substituted in \eqref{Eq:PfMomapChain} to give \eqref{Eq:ChainMomapCond-s-X}. 

We check \eqref{Eq:ChainMomapCond-3} in the same way by noticing that $\Gl(\nfat)$ acts on $\Phi^{(c)}_r$ by $g\cdot \Phi^{(c)}_r = g_r \Phi^{(c)}_r g_r^{-1}$. Alternatively, this is proven by induction using \eqref{Eq:chainGeomMomap}. The last identity is also obtained in that way.
\end{proof}

\begin{lem} \label{Lem:MomapChain-Bis}
 Fix $b\in \{0,1,\ldots,|\dfat|\}$. Then, for any $s\in\Z_m$ and  $(r,\beta)> \rho(b)$, we have  on $\MM_{Q_\dfat,\nfat}^{\bullet}$ 
{ \allowdisplaybreaks  
 \begin{subequations}  \label{Eq:ChainMomapCond-4}
  \begin{align}
\br{(\Phi^{(b)}_s)_{ij},(V_{r,\beta})_{l}}=&\frac12 \delta_{rs}\left(-(V_{r,\beta})_{j}(\Phi^{(b)}_s)_{il} + (V_{r,\beta}\Phi^{(b)}_s)_{j}\delta_{il} \right)   \,, \\
\br{(\Phi^{(b)}_s)_{ij},(W_{r,\beta})_{k}}=&\frac12 \delta_{rs} \left(\delta_{kj} (\Phi^{(b)}_sW_{r,\beta})_{i} - (\Phi^{(b)}_s)_{kj} (W_{r,\beta})_{i} \right) \,,\\
\br{\Phi^{(b)}_\infty,(V_{r,\beta})_{l}}=&0 \,, \quad \br{\Phi^{(b)}_\infty,(W_{r,\beta})_{k}}= 0\,. 
  \end{align}
 \end{subequations}
}   
\end{lem}
\begin{proof}
 Let $c\in \{1,\ldots,|\dfat|\}$ be such that $\rho(c)=(r,\beta)$. Since $c>b$ by assumption, we have 
\begin{equation}
(\Phi^{(b)}_s)_{ij}\in \CC[\MM_{Q_{b}^{\res},\nfat}^{\bullet}] \subset \CC[\MM_{Q_{c-1}^{\res},\nfat}^{\bullet}] \,.
\end{equation}
Now, we remark that due to our construction, the quasi-Hamiltonian structure on $\MM_{Q_{c}^{\res},\nfat}^{\bullet}$ is obtained by fusion of 
$\MM_{Q_{c-1}^{\res},\nfat}^{\bullet}$ and the Van den Bergh space 
\begin{equation}
 \MM_{(1,n_r)}=\{V_{r,\beta}\in \Mat(1\times n_r,\CC), W_{r,\beta}\in \Mat(n_r\times 1,\CC)\mid 1+V_{r,\beta}W_{r,\beta}\neq 0\}\,.
\end{equation}
It is then an exercise to check that we can get \eqref{Eq:ChainMomapCond-4} after using Proposition \ref{Pr:qPfus} twice. 

As an alternative proof, we know the quasi-Poisson bracket between $(\Phi^{(c)}_s)_{ij},\Phi^{(c)}_\infty$ and $(V_{r,\beta})_{l},(W_{r,\beta})_{k}$ from Lemma \ref{Lem:MomapChain} as $\rho(c)=(r,\beta)$. Then, we can check \eqref{Eq:ChainMomapCond-4} by descending induction on $b=0,\ldots,c-1$. It suffices to combine the recurrence relation \eqref{Eq:chainGeomMomap} for the moment map and the fact that we know the quasi-Poisson bracket between the entries of $V_{\rho(b)},W_{\rho(b)}$ and $V_{r,\beta},W_{r,\beta}$ by \eqref{Eq:qPbrack3}. 
\end{proof}
 
\begin{rem} \label{Rem:OpenChain}
Denote by $\MM_{Q_b^{\res},\nfat}^{\circ}\subset \MM_{Q_b^{\res},\nfat}^{\bullet}$ the subvariety defined by the condition that $X\in \End(\VV_{\cyc})$ is invertible. 
The variety $\MM_{Q_{|\dfat|}^{\res},\nfat}^{\circ}$ coincides with  $\MM_{Q_\dfat,\nfat}^{\circ}$ as defined in \ref{sss:Dim1n}, and the chain of inclusions of coordinate rings  \eqref{Eq:chainAlg} can be restricted to these open subvarieties. 
For later use, we respectively denote by $\Cnqm^{\res,b}$ and $\Cnqm^{\circ,\res,b}$ the varieties obtained from $\MM_{Q_b^{\res},\nfat}^{\bullet}$ and   $\MM_{Q_b^{\res},\nfat}^{\circ}$ after performing quasi-Hamiltonian reduction at the value $\qfat \cdot \Id=(q_\infty, q_s \Id_n)_{s\in \Z_m}$ of the moment map as in \eqref{Eq:Cnqm}, where $q_\infty$ satisfies \eqref{Eq:Cond-q}.  
\end{rem}

\subsubsection{Construction of abelian subalgebras}

Let $\rL\in \End(\VV_\cyc)$ be one of the matrices $Y$ or $Z$. 
For each $b\in \{0,1,\ldots,|\dfat|\}$, the moment map $\Phi^{(b)}$ of $\MM_{Q_b^{\res},\nfat}^{\bullet}$ defined inductively by \eqref{Eq:chainGeomMomap} can be regarded as a matrix-valued function on $\MM_{Q_{\dfat},\nfat}^{\bullet}$ (or $\MM_{Q_{\dfat},\nfat}^{\circ}$ for $\rL=Z$). 
We define the element 
\begin{equation}\label{Eq:LjEnd}
 \rL^{(b)}:=\Phi^{(b)}\rL=\sum_{s\in \Z_m} \Phi^{(b)}_s \rL_{s-1} \in \End(\VV_\cyc)\,,
\end{equation}
and  introduce the commutative algebra
\begin{equation}   \label{Eq:Palg-LL}
  \LLl=\CC[\tr((\rL^{(b)})^k)\mid k\in \N,\, b=0,1,\ldots,|\dfat|]\,.
\end{equation}
Note as in \ref{ss:PoiDyn} that this algebra is finitely generated by $\Gl(\nfat)$-invariant regular functions. Furthermore, a regular function $\tr((\rL^{(b)})^k)$ is trivially zero except when $k$ is divisible by $m$.

\begin{prop} \label{Pr:IScymom}
Consider $\MM_{Q_{\dfat},\nfat}^{\bullet}$ (or $\MM_{Q_{\dfat},\nfat}^{\circ}$ for $\rL=Z$) endowed with the quasi-Poisson bracket given by  \eqref{Eq:qPbrack1}--\eqref{Eq:qPbrack3}. Then $\LLl$ is an abelian Poisson algebra. 
\end{prop}
\begin{proof}
We have to show that the functions $\big(\tr ((\rL^{(b)})^k)\big)_{b=0,1,\ldots,|\dfat|}^{k\in \N}$ are in involution. 
It suffices to prove that for any $k=k_0m$ and $l=l_0m$ with $k_0,l_0\geq0$, and for any $b,c\in\{0,1,\ldots,|\dfat|\}$ with $b\geq c$, we have 
$\br{ \tr ((\rL^{(b)})^k) , \tr ((\rL^{(c)})^l)}=0$. It is useful to note that we can write  
\begin{equation*}
\tr ((\rL^{(b)})^k) =m \tr\big((\rL_{m-1}\Phi^{(b)}_{m-1}\ldots\rL_0\Phi^{(b)}_0)^{k_0}\big) ,\quad 
\tr ((\rL^{(c)})^l)=m\tr\big((\rL_{m-1}\Phi^{(c)}_{m-1}\ldots\rL_0\Phi^{(c)}_0)^{l_0}\big)\,.
\end{equation*}
Consider $\rL=Y$, the case $\rL=Z$ being similar with the matrix $Z_s$ replacing $Y_s$ in all the computations below. We have  
\begin{equation*}
\begin{aligned}
 \br{(Y_s\Phi^{(b)}_s)_{ij},(Y_r \Phi^{(c)}_r)_{kl}}=& 
\quad \sum_{j',l'}   \br{(Y_s)_{ij'},(Y_r)_{kl'}} (\Phi^{(b)}_s)_{j'j} (\Phi^{(c)}_r)_{l'l}\\
 &+ \sum_{j',k'} (Y_r )_{kk'} \br{(Y_s)_{ij'},(\Phi^{(c)}_r)_{k'l}}  (\Phi^{(b)}_s)_{j'j} \\
 &+ \sum_{i',l'} (Y_s)_{ii'}\br{(\Phi^{(b)}_s)_{i'j},(Y_r)_{kl'}} \Phi^{(c)}_r)_{l'l} \\
 &+ \sum_{i',k'} (Y_s)_{ii'}(Y_r)_{kk'}\br{(\Phi^{(b)}_s)_{i'j},( \Phi^{(c)}_r)_{k'l}}\,.
\end{aligned}
\end{equation*}
The four quasi-Poisson brackets in the right-hand side of the last formula can be computed using \eqref{cyida'}, \eqref{Eq:ChainMomapCond-s-Y} twice, and \eqref{Eq:ChainMomapCond-3}, respectively. 
We get after simplifications that 
\begin{equation*}
\begin{aligned}
 \br{(Y_s\Phi^{(b)}_s)_{ij},(Y_r \Phi^{(c)}_r)_{kl}}=& 
\quad \frac12 \delta_{rs}\left[(Y_r\Phi^{(c)}_r)_{kj} (Y_s\Phi^{(b)}_s)_{il} - (Y_s\Phi^{(b)}_s)_{kj}(Y_r\Phi^{(c)}_r)_{il}\right] \\
&+\frac12 \delta_{s,r-1}\, (Y_r\Phi^{(c)}_r Y_s\Phi^{(c)}_s)_{kj} \delta_{il} 
- \frac12 \delta_{s,r+1}\, \delta_{kj} (Y_s\Phi^{(b)}_s Y_r\Phi^{(c)}_r)_{il}\,.
\end{aligned}
\end{equation*}
Using the Leibniz rule, we get the matrix identity  
\begin{equation} \label{Eq:trLbLc}
 \br{ \tr ((Y^{(b)})^k) , Y_r \Phi^{(c)}_r} =k\, \left( Y_r \Phi^{(c)}_r (Y^{(b)})^k   -  (Y^{(b)})^k Y_r \Phi^{(c)}_r \right)\,.
\end{equation} 
Note that \eqref{Eq:trLbLc} is of the standard form $\br{f,B}=BC-CB$ for some regular function $f$ and matrix-valued functions $B,C$. This particular form entails $\br{f,B^l}=B^lC-CB^l$ for any $l\geq 1$ by the Leibniz rule, then $\br{f, \tr(B^l)}=0$ by taking trace. 
Thus, we can  conclude that $\br{ \tr ((Y^{(b)})^k) , \tr ((Y^{(c)})^l)}=0$. 
\end{proof}

In a similar way, we can consider the commutative subalgebras of $\CC[\MM_{Q_{\dfat},\nfat}^{\bullet}]$ 
\begin{subequations}
\begin{align}
   \LL_+^{\cyc}&=\CC[\tr\big((\Phi^{(b)} (1_{\cyc}+XY)^{-1})^k\big)\mid k\in\N,\, b=0,1,\ldots,|\dfat|]\,, \label{Eq:Palg-cL+} \\
  \LL_+&=\CC[\tr\big((\Phi_s^{(b)} (\Id_{\VV_s}+X_sY_s)^{-1})^k\big)\mid k\in\N,\,s\in \Z_m,\, b=0,1,\ldots,|\dfat|]\,. \label{Eq:Palg-L+}
\end{align}
\end{subequations}
We have $\LL_+^{\cyc}\subset \LL_+$ in view of the identity 
\begin{equation*}
 \tr\big((\Phi^{(b)} (1_{\cyc}+XY)^{-1})^k\big)=\sum_{s\in \Z_m} \tr\big((\Phi_s^{(b)} (\Id_{\VV_s}+X_sY_s)^{-1})^k\big)\,.
\end{equation*}
Note that the trace in the left-hand side is taken over $\VV_{\cyc}=\bigoplus_{s\in\Z_m} \VV_s$, while on the right-hand side we take the trace over each $\VV_s$, $s\in \Z_m$.

\begin{prop} \label{Pr:IScymomBis}
Consider $\MM_{Q_{\dfat},\nfat}^{\bullet}$ endowed with the quasi-Poisson bracket given by  \eqref{Eq:qPbrack1}--\eqref{Eq:qPbrack3}. Then $\LL_+$ and $\LL_+^{\cyc}$ are  abelian Poisson algebras. 
\end{prop}
\begin{proof}
 It suffices to prove the result for $\LL_+$ by adapting the proof of Proposition \ref{Pr:IScymom} with $k,l\geq0$ not necessarily divisible by $m$. One needs to use  the quasi-Poisson brackets 
\begin{equation*}
 \begin{aligned}
\br{(\Id_{\VV_s}+X_sY_s)^{-1}_{ij},(\Id_{\VV_r}+X_rY_r)^{-1}_{kl}}=&
\frac12 \delta_{sr}\left[\delta_{kj} (\Id_{\VV_s}+X_sY_s)^{-2}_{il} - (\Id_{\VV_s}+X_sY_s)^{-2}_{kj}\delta_{il} \right]\,, \\
\br{(\Phi_s^{(b)})_{ij},(\Id_{\VV_r}+X_rY_r)^{-1}_{kl}}=& \quad  
\frac12 \delta_{sr}\left[ (\Id_{\VV_r}+X_rY_r)^{-1}_{kj} (\Phi_s^{(b)})_{il} + ((\Id_{\VV_r}+X_rY_r)^{-1}\Phi_s^{(b)})_{kj}\delta_{il}
\right]\\
&-\frac12 \delta_{sr}\left[ \delta_{kj} (\Phi_s^{(b)}(\Id_{\VV_r}+X_rY_r)^{-1})_{il} + (\Phi_s^{(b)})_{kj}(\Id_{\VV_r}+X_rY_r)^{-1}_{il}
\right]\,,
 \end{aligned}
\end{equation*}
which are obtained from \eqref{cyid-1XYa}  and \eqref{Eq:ChainMomapCond-1}. 
For $b\leq c$, these equalities together with \eqref{Eq:ChainMomapCond-3} allow to derive the matrix identity  
\begin{equation} \label{Eq:trL-1XY}
 \begin{aligned}
&\br{ \tr \big((\Phi^{(b)}_s (\Id_{\VV_s}+X_sY_s)^{-1})^k\big), \Phi_r^{(c)} (\Id_{\VV_r}+X_rY_r)^{-1}} \\
 =&\quad k \delta_{rs}\, \Phi_r^{(c)} (\Id_{\VV_r}+X_rY_r)^{-1}\left(\Phi^{(b)}_s (\Id_{\VV_s}+X_sY_s)^{-1}\right)^k \\  
 &-k \delta_{rs}   \left(\Phi^{(b)}_s (\Id_{\VV_s}+X_sY_s)^{-1}\right)^k \Phi_r^{(c)} (\Id_{\VV_r}+X_rY_r)^{-1} \,,  
 \end{aligned}
\end{equation}
from which we get that $\tr \big((\Phi^{(b)}_s (\Id_{\VV_s}+X_sY_s)^{-1})^k\big)$ and $\tr \big((\Phi^{(c)}_r (\Id_{\VV_r}+X_rY_r)^{-1})^k\big)$ are in involution. 
\end{proof}

\begin{rem}
We can adapt Proposition \ref{Pr:IScymomBis} to the case where we use $1_{\cyc}+XY$ instead of its inverse by inverting all the involved matrices. This result is in fact easy to obtain because 
\begin{equation}
 \tr\big(((\Phi_s^{(b)})^{-1} (\Id_{\VV_s}+X_sY_s))^k\big)=\tr\big((\Phi_s^{(b)} (\Id_{\VV_s}+X_sY_s)^{-1})^{-k}\big)
\end{equation}
belongs to $\LL_+$ by the Cayley-Hamilton theorem. 
In the same way, if we want to adapt Proposition \ref{Pr:IScymom} to the case  $\rL=X$, it is necessary to replace the matrix $\Phi^{(b)}$ by its inverse $(\Phi^{(b)})^{-1}$ in the construction of $\LLl$  in order to get an abelian Poisson algebra. Since this is a tedious adaptation which we will not use, we do not consider the case $\rL=X$ any further. 
\end{rem}

\subsubsection{Dynamics} \label{sss:AbDyn-flow}

We show that we can integrate the Hamiltonian vector field associated with any of the generators of the algebras $\LLl$ or $\LL_+$ that we have just considered. 
We start with $\LLl$ in the case $\rL=Y$. 
We recall that we are interested in the symmetric functions of $Y^{(b)}:=\Phi^{(b)} Y$ regarded as elements of $\End(\VV_{\cyc})$. 
\begin{prop} \label{Pr:floYcy2}
Given the initial condition in $\MM_{Q_{\dfat},\nfat}^{\bullet}$ 
\begin{equation*}
 X(0):=\sum_{s\in\Z_m} X_s(0),\quad Y(0):=\sum_{s\in\Z_m}Y_s(0),\quad V_{s,\alpha}(0),\quad W_{s,\alpha}(0)\,,
\end{equation*}
the flow at time $t$ defined by the Hamiltonian function $\frac1k \tr ((Y^{(b)})^k)$ for $k=k_0m$, $k_0\geq1$, 
and $b\in \{0,1,\ldots,|\dfat|\}$ is given by 
\begin{equation*}
\begin{aligned}
   &X(t)= \exp\left(-t Y^{(b)}(0)^k\right)X(0) + Y^{(b)}(0)^{-1}[\exp(-t Y^{(b)}(0)^k)-1_{\cyc}]\Phi^{(b)}(0)\,, \quad 
Y(t)= Y(0)\,, \\
&V_{s,\alpha}(t)= V_{s,\alpha}(0)\exp\left( t Y(0)Y^{(b)}(0)^{k-1}\Phi^{(b)}(0)\right), \quad  (s,\alpha) \leq \rho(b),\\
&W_{s,\alpha}(t)= \exp\left(-t Y(0)Y^{(b)}(0)^{k-1}\Phi^{(b)}(0)\right)W_{s,\alpha}(0),\quad (s,\alpha) \leq \rho(b), \\ 
&V_{s,\alpha}(t)= V_{s,\alpha}(0)\,, \,\, W_{s,\alpha}(t)= W_{s,\alpha}(0)\,,\,\, (s,\alpha) > \rho(b)\,.
\end{aligned}
\end{equation*} 
\end{prop}

Note that to understand the dynamics of $V_{s,\alpha}(t),W_{s,\alpha}(t)$ with $(s,\alpha) \leq \rho(b)$, these matrices must be regarded as elements of 
$\Hom(\VV_{\cyc},\VV_\infty)$ and $\Hom(\VV_\infty,\VV_{\cyc})$ respectively. 
This means that we should write the dynamics of $V_{s,\alpha}(t)$ in $\Hom(\VV_{s},\VV_\infty)$ as 
\begin{equation*}
\begin{aligned}
  V_{s,\alpha}(t)=& V_{s,\alpha}(0)\exp\left( t A(0)\right), \quad \text{ where }\\
A(0):=& (Y_{s-1}(0)\Phi^{(b)}_{s-1}(0)\ldots Y_{0}(0)\Phi^{(b)}_{0}(0) Y_{m-1}(0)\Phi^{(b)}_{m-1}(0)\ldots Y_{s}(0)\Phi^{(b)}_{s}(0))^{k_0} 
\end{aligned}
\end{equation*}
for $(s,\alpha) \leq \rho(b)$ and $k=k_0m$. 
A similar formula holds for $W_{s,\alpha}(t)$ in $\Hom(\VV_\infty,\VV_{s})$. 

\begin{proof}[Proof of Proposition \ref{Pr:floYcy2}]
Using \eqref{Eq:qPbrack1}, \eqref{Eq:qPbrack2} and Lemma \ref{Lem:MomapChain}, we can find that for $(r,\alpha)\leq \rho(b)$, 
 { \allowdisplaybreaks  
  \begin{align*}
\br{(\Phi^{(b)}_sY_{s-1})_{ij},(X_r)_{kl}}=&
\delta_{r,s-1}\left(-\delta_{kj}(\Phi^{(b)}_s)_{il}-\frac12 \delta_{kj}(\Phi^{(b)}_sY_{s-1}X_r)_{il} + \frac12 (X_r\Phi^{(b)}_sY_{s-1})_{kj}\delta_{il} \right)  \\
&-\frac12 \delta_{rs} (\Phi^{(b)}_sY_{s-1})_{kj} (X_r)_{il} -\frac12 \delta_{r,s-2}\,(X_r)_{kj}(\Phi_s^{(b)}Y_{s-1})_{il}  \,,\\
\br{(\Phi^{(b)}_sY_{s-1})_{ij},(Y_r)_{kl}}=&
-\frac12\delta_{r,s-1}\left( (Y_r)_{kj}(\Phi^{(b)}_sY_{s-1})_{il} + (\Phi^{(b)}_sY_{s-1})_{kj} (Y_r)_{il} \right)  \\
&+\frac12 \delta_{rs} (Y_r\Phi^{(b)}_sY_{s-1})_{kj} \delta_{il} +\frac12 \delta_{r,s-2}\,\delta_{kj}(\Phi_s^{(b)}Y_{s-1}Y_r)_{il}  \,,\\
\br{(\Phi^{(b)}_sY_{s-1})_{ij},(W_{r,\alpha})_{k}}=&
-\delta_{rs}\left( (Y_{s-1})_{kj}(\Phi^{(b)}_sW_{r,\alpha})_{i} +\frac12 (\Phi^{(b)}_sY_{s-1})_{kj} (W_{r,\alpha})_{i} \right)  \\
&+\frac12 \delta_{r,s-1}\,\delta_{kj}(\Phi_s^{(b)}Y_{s-1}W_{r,\alpha})_{i}  \,, \\
\br{(\Phi^{(b)}_sY_{s-1})_{ij},(V_{r,\alpha})_{l}}=&
\delta_{rs}\left( (V_{r,\alpha}Y_{s-1})_{j}(\Phi^{(b)}_s)_{il} + \frac12 (V_{r,\alpha}\Phi^{(b)}_sY_{s-1})_{j} \delta_{il} \right)  \\
&-\frac12 \delta_{r,s-1}\,(V_{r,\alpha})_{j}(\Phi_s^{(b)}Y_{s-1})_{il}  \,.
  \end{align*}
}  
If $(r,\alpha)>\rho(b)$, we use Lemma \ref{Lem:MomapChain-Bis} instead and we get 
 \begin{equation*} 
  \begin{aligned}
\br{(\Phi^{(b)}_sY_{s-1})_{ij},(W_{r,\alpha})_{k}}=&
-\frac12\delta_{rs}  (\Phi^{(b)}_sY_{s-1})_{kj} (W_{r,\alpha})_{i} 
+\frac12 \delta_{r,s-1}\,\delta_{kj}(\Phi_s^{(b)}Y_{s-1}W_{r,\alpha})_{i}  \,, \\
\br{(\Phi^{(b)}_sY_{s-1})_{ij},(V_{r,\alpha})_{l}}=&
\frac12\delta_{rs} (V_{r,\alpha}\Phi^{(b)}_sY_{s-1})_{j} \delta_{il} 
-\frac12 \delta_{r,s-1}\,(V_{r,\alpha})_{j}(\Phi_s^{(b)}Y_{s-1})_{il}  \,.
  \end{aligned}
 \end{equation*}
For $k=k_0m$, we can use these identities and the vector field 
\begin{equation} \label{Eq:floY2-vf}
 \frac{d}{dt}:=\frac1k\br{\tr ((Y^{(b)})^k),-}
 =\frac{1}{k_0}\br{\tr \big( (\Phi^{(b)}_{m-1}Y_{m-2}\ldots \Phi^{(b)}_1 Y_{0} \Phi^{(b)}_0Y_{m-1})^{k_0}\big),-}
\end{equation}
to write that 
\begin{equation}
\begin{aligned} \label{Eq:floY2-Mot}
 \frac{d X}{dt}=&-(Y^{(b)})^{k-1}\Phi^{(b)}-(Y^{(b)})^kX\,, \quad  \frac{d Y}{dt}=0\,, \\
 \frac{d W_{r,\alpha}}{dt}=&-Y(Y^{(b)})^{k-1}\Phi^{(b)}W_{r,\alpha}\,, \quad 
 \frac{d V_{r,\alpha}}{dt}=V_{r,\alpha} Y(Y^{(b)})^{k-1}\Phi^{(b)}\,, \quad \text{ if }(r,\alpha)\leq \rho(b)\,, \\
  \frac{d W_{r,\alpha}}{dt}=&0\,, \quad 
 \frac{d V_{r,\alpha}}{dt}=0\,, \quad \text{ if } (r,\alpha)> \rho(b)\,.
\end{aligned}
\end{equation} 
(Recall that for a matrix-valued function $A$, $dA/dt$ is the matrix-valued function with entry $(i,j)$ given by $dA_{ij}/dt=\frac{1}{k}\br{\tr ((Y^{(b)})^k),A_{ij}}$.)
These equations can be integrated to yield the solution from the statement provided that 
\begin{equation}\label{Eq:pffloYcy2}
 \frac{d Y^{(b)}}{dt}=0\,, \quad \frac{d \Phi^{(b)}}{dt}=0\,.
\end{equation}
The two identities in \eqref{Eq:pffloYcy2} follow from \eqref{Eq:trLbLc} with $c=b$ together with $d Y/dt=0$. 
\end{proof}

In the case $\rL=Z$ with $Z^{(b)}:=\Phi^{(b)} Z$, the result takes place in $\MM_{Q_{\dfat},\nfat}^{\circ}$. 
\begin{prop} \label{Pr:floZcy2}
Given the initial condition in $\MM_{Q_{\dfat},\nfat}^{\circ}$
\begin{equation*}
 X(0):=\sum_{s\in\Z_m} X_s(0),\quad Z(0):=\sum_{s\in\Z_m}Z_s(0),\quad V_{s,\alpha}(0),\quad W_{s,\alpha}(0)\,,
\end{equation*}
the flow at time $t$ defined by the Hamiltonian function $\frac1k \tr ((Z^{(b)})^k)$ for $k=k_0m$, $k_0\geq1$, 
and $b\in \{0,1,\ldots,|\dfat|\}$ is given by 
\begin{equation*}
\begin{aligned}
   &X(t)= \exp\left(-t Z^{(b)}(0)^k\right)X(0)\,, \quad Z(t)= Z(0)\,, \\
&V_{s,\alpha}(t)= V_{s,\alpha}(0)\exp\left( t Z(0)Z^{(b)}(0)^{k-1}\Phi^{(b)}(0)\right), \quad  (s,\alpha) \leq \rho(b),\\
&W_{s,\alpha}(t)= \exp\left(-t Z(0)Z^{(b)}(0)^{k-1}\Phi^{(b)}(0)\right)W_{s,\alpha}(0),\quad (s,\alpha) \leq \rho(b), \\ 
&V_{s,\alpha}(t)= V_{s,\alpha}(0)\,, \,\, W_{s,\alpha}(t)= W_{s,\alpha}(0)\,,\,\, (s,\alpha) > \rho(b)\,.
\end{aligned}
\end{equation*} 
\end{prop}
\begin{proof}
 It suffices to repeat the proof of Proposition \ref{Pr:floYcy2}. We can simply replace each occurrence of the matrix $Y$ by $Z$ in \eqref{Eq:floY2-Mot}, with the only difference that the first equality becomes  $dX/dt=\frac1k \br{\tr ((Z^{(b)})^k),X}=-(Z^{(b)})^k X$. 
\end{proof}

To write down elements of $\LL_+$, we adopt the compact notation $T_s:=(\Id_{\VV_s}+X_sY_s)$, $T=\sum_{s\in\Z_m} T_s$, and $(T^{(b)}_s)^{-1}:=\Phi^{(b)}T_s^{-1}$, $(T^{(b)})^{-1}=\sum_{s\in\Z_m} (T_s^{(b)})^{-1}$ for any $b\in \{0,1,\ldots,|\dfat|\}$. We work in $\MM_{Q_{\dfat},\nfat}^{\circ}$ to simplify the statement of the next proposition, since the invertibility of $X$ allows to characterise a point in terms of the data $(X,T,V_{s,\alpha},W_{s,\alpha})$.
\begin{prop} \label{Pr:floTcy2}
Given the initial condition in $\MM_{Q_{\dfat},\nfat}^{\circ}$ 
\begin{equation*}
 X(0):=\sum_{s\in\Z_m} X_s(0),\quad T(0):=\sum_{s\in\Z_m}T_s(0),\quad V_{s,\alpha}(0),\quad W_{s,\alpha}(0)\,,
\end{equation*}
the flow at time $t$ defined by the Hamiltonian function $\frac1k \tr ((T^{(b)}_r)^{-k})$ for $k\geq1$, $r\in \Z_m$ 
and $b\in \{0,1,\ldots,|\dfat|\}$ is given by 
\begin{equation*}
\begin{aligned}
   &X(t)= \, X(0)\exp\left(t T_r^{-1}(0) (T^{(b)}_r(0))^{-(k-1)}\Phi_r^{(b)}\right)\,, \quad T(t)= T(0)\,, \\
&V_{r,\alpha}(t)=V_{r,\alpha}(0)\exp\left(t T_r^{-1}(0)(T^{(b)}_r(0))^{-(k-1)}\Phi_r^{(b)}(0)\right), 
\quad  (r,\alpha) \leq \rho(b),\\
&W_{r,\alpha}(t)= \exp\left(-t T^{-1}_r(0)(T^{(b)}_r(0))^{-(k-1)}\Phi^{(b)}_r(0)\right) W_{r,\alpha}(0),
\quad (r,\alpha) \leq \rho(b), \\ 
&V_{s,\alpha}(t)= V_{s,\alpha}(0)\,, \,\, W_{s,\alpha}(t)= W_{s,\alpha}(0)\,,\,\, 
\text{if }s=r \text{ with }(r,\alpha) > \rho(j) \text{ or if }s\neq r\,.
\end{aligned}
\end{equation*}
\end{prop}
\begin{proof}
 It suffices again to repeat the proof of Proposition \ref{Pr:floYcy2}. To compute the quasi-Poisson brackets of the form $\br{((T^{(b)}_r)^{-1})_{ij},-}$, we need Lemmas \ref{Lem:MomapChain} and \ref{Lem:MomapChain-Bis}, as well as \eqref{Eq:qPbrack-1XY}.  We can use these brackets  and the vector field 
\begin{equation}
 \frac{d}{dt}:=\frac1k\br{\tr ((T^{(b)}_r)^{-k}),-}
\end{equation} 
to derive 
 \begin{equation*}
\begin{aligned}
    \frac{dX}{dt}=&X T^{-1}_r (T^{(b)}_r)^{-(k-1)}\Phi_r^{(b)}\,, \quad 
\frac{dT}{dt}= 0 \,, \\
\frac{dV_{s,\alpha}}{dt}=&\delta_{sr}\,V_{s,\alpha} T_r^{-1} (T^{(b)}_r)^{-(k-1)} \Phi^{(b)}_r\,, \quad 
\frac{dW_{s,\alpha}}{dt}=-\delta_{sr}\,T_r^{-1} (T^{(b)}_r)^{-(k-1)} \Phi^{(b)}_r W_{s,\alpha}\,, \quad (s,\alpha) \leq \rho(b)\,, \\
\frac{dV_{s,\alpha}}{dt}=&0\,, \quad \frac{dW_{s,\alpha}}{dt}=0\,, \quad (s,\alpha) > \rho(b)\,,
\end{aligned}
\end{equation*}
as well as $d\Phi^{(b)}/dt=0$. We can conclude by integrating these identities to the unique solution given in the statement for the chosen initial condition. 
\end{proof}

We know that the various flows presented in Propositions \ref{Pr:floYcy2}, \ref{Pr:floZcy2} and \ref{Pr:floTcy2} are commuting because the corresponding Hamiltonian functions pairwise commute under the quasi-Poisson bracket due to Propositions \ref{Pr:IScymom} and \ref{Pr:IScymomBis}. (However, two flows taken from two different propositions need not commute.) Let us nevertheless check this result explicitly in the case of Proposition \ref{Pr:floYcy2} since it is certainly not obvious by just looking at the given evolution equations. 

Let $k,l\geq1$ be divisible by $m$, and fix $b,c\in \{0,1,\ldots,|\dfat|\}$. Without loss of generality, assume that $b<c$. 
Denote by $d/dt$ and $d/d\tau$ the Hamiltonian vector fields \eqref{Eq:floY2-vf} associated with the functions $\frac1k \tr ((Y^{(b)})^k)$ and $\frac1l \tr ((Y^{(c)})^l)$, respectively. The corresponding flows commute if we can check that 
\begin{equation} \label{Eq:CheckVF}
 \frac{d}{dt}\frac{d}{d\tau}-\frac{d}{d\tau}\frac{d}{dt} = 0 
\end{equation}
when evaluated on the matrices $(X,Y,V_{s,\alpha},W_{s,\alpha})$ parametrising $\MM_{Q_{\dfat},\nfat}^{\bullet}$. 
From \eqref{Eq:floY2-Mot} both sides of \eqref{Eq:CheckVF} trivially vanish when applied to $Y$ or $V_{s,\alpha},W_{s,\alpha}$ with $(s,\alpha)>\rho(c)>\rho(b)$. To check the other cases, let us note from the recurrence relation for the moment map \eqref{Eq:chainGeomMomap} that we can write 
\begin{equation} \label{Eq:CheckVF-Phi}
 \sum_{s\in\Z_m}\Phi^{(c)}_s=\big(\sum_{s\in\Z_m}\Phi^{(b)}_s\big) (1_{\cyc}+W_{\rho(b+1)}V_{\rho(b+1)})^{-1}\ldots (1_{\cyc}+W_{\rho(c)}V_{\rho(c)})^{-1}\,,
\end{equation}
because $c>b$. Noting that $\frac{d}{dt}(1_{\cyc}+W_{\rho(b')}V_{\rho(b')})=0$ for any $b'>b$ due to \eqref{Eq:floY2-Mot}, we can combine \eqref{Eq:CheckVF-Phi} with \eqref{Eq:floY2-Mot}--\eqref{Eq:pffloYcy2} to conclude that 
\begin{equation}
 \frac{d}{dt}\sum_{s\in\Z_m}\Phi^{(c)}_s=0\,,\quad \frac{dY^{(c)}}{dt}=0\,.
\end{equation}
In the same way, \eqref{Eq:CheckVF-Phi} with \eqref{Eq:floY2-Mot}--\eqref{Eq:pffloYcy2} lead to 
\begin{equation}
\begin{aligned}
 \frac{d}{d\tau} \sum_{s\in\Z_m}\Phi^{(b)}_s
 =&\frac{d}{d\tau}\big(\sum_{s\in\Z_m}\Phi^{(c)}_s\big) (1_{\cyc}+W_{\rho(c)}V_{\rho(c)})\ldots (1_{\cyc}+W_{\rho(b+1)}V_{\rho(b+1)}) \\
=&-(Y^{(c)})^l \Phi^{(b)} + \Phi^{(b)} Y(Y^{(c)})^{l-1}\Phi^{(c)}\,,
\end{aligned}
\end{equation}
because we have for any $b'<c$  
\begin{align*}
 \frac{d}{d\tau}(1_{\cyc}+W_{\rho(b')}V_{\rho(b')})= &
(1_{\cyc}+W_{\rho(b')}V_{\rho(b')}) Y(Y^{(c)})^{l-1}\Phi^{(c)} \\
&- Y(Y^{(c)})^{l-1}\Phi^{(c)}(1_{\cyc}+W_{\rho(b')}V_{\rho(b')})\,.
\end{align*}
This, in turn, implies that $dY^{(b)}/d\tau=-(Y^{(c)})^l Y^{(b)} + Y^{(b)} (Y^{(c)})^{l}$. 

Let us now apply the left-hand side of \eqref{Eq:CheckVF} to $X$. We have on the one hand, 
\begin{equation*}
  \begin{aligned}
\frac{d}{dt}\frac{d}{d\tau}X=&
-\frac{d}{dt}(Y^{(c)})^{l-1}\Phi^{(c)}-\frac{d}{dt} (Y^{(c)})^l X = -(Y^{(c)})^l \frac{dX}{dt} \\
=&(Y^{(c)})^l(Y^{(b)})^{k-1}\Phi^{(b)}+(Y^{(c)})^l (Y^{(b)})^k X\,.
  \end{aligned}
\end{equation*}
On the other hand,  
\begin{equation*}
  \begin{aligned}
\frac{d}{d\tau}\frac{d}{dt}X=&
-\frac{d(Y^{(b)})^{k-1}}{d\tau}\Phi^{(b)}-(Y^{(b)})^{k-1}\frac{d\Phi^{(b)}}{d\tau}
-\frac{d(Y^{(b)})^k}{d\tau} X-(Y^{(b)})^k\frac{dX}{d\tau} \\
=&-\left(-(Y^{(c)})^l (Y^{(b)})^{k-1} + (Y^{(b)})^{k-1} (Y^{(c)})^{l}\right)\Phi^{(b)}\\
&-(Y^{(b)})^{k-1} \left(-(Y^{(c)})^l \Phi^{(b)} + Y^{(b)} (Y^{(c)})^{l-1}\Phi^{(c)}\right) \\
&-\left(-(Y^{(c)})^l (Y^{(b)})^{k} + (Y^{(b)})^{k} (Y^{(c)})^{l}\right) X\\
&-(Y^{(b)})^k \left(-(Y^{(c)})^{l-1}\Phi^{(c)}-(Y^{(c)})^lX\right) \\
=&(Y^{(c)})^l(Y^{(b)})^{k-1}\Phi^{(b)}+ (Y^{(c)})^l (Y^{(b)})^{k} X\,,
  \end{aligned}
\end{equation*}
so that the left-hand side of \eqref{Eq:CheckVF} applied to $X$ is indeed zero. In the same way, we can check that 
\begin{equation*}
  \begin{aligned}
\frac{d}{dt}\frac{d}{d\tau}W_{s,\alpha}=&0=\frac{d}{d\tau}\frac{d}{dt}W_{s,\alpha}\,, \quad \text{ if }b<(s,\alpha)\leq c\,, \\
\frac{d}{dt}\frac{d}{d\tau}V_{s,\alpha}=&0=\frac{d}{d\tau}\frac{d}{dt}V_{s,\alpha}\,, \quad \text{ if }b<(s,\alpha)\leq c\,, \\
\frac{d}{dt}\frac{d}{d\tau}W_{s,\alpha}=&Y(Y^{(c)})^l(Y^{(b)})^{k-1}\Phi^{(b)}W_{s,\alpha}
=\frac{d}{d\tau}\frac{d}{dt}W_{s,\alpha}\,, \quad \text{ if }(s,\alpha)\leq b<c\,, \\
\frac{d}{dt}\frac{d}{d\tau}V_{s,\alpha}=&V_{s,\alpha}Y(Y^{(b)})^k(Y^{(c)})^{l-1}\Phi^{(c)}
=\frac{d}{d\tau}\frac{d}{dt}V_{s,\alpha}\,, \quad \text{ if }(s,\alpha)\leq b<c\,.
  \end{aligned}
\end{equation*}
Hence \eqref{Eq:CheckVF} holds and the two vector fields are commuting, as expected.

\section{Local structure}  \label{S:Loc}

We assume that we work with the dimension vector $(1,n\delta)$ as in  \ref{sss:Dim1n}, and that the regularity condition on $\qfat$ from Propositions \ref{Pr:CyMQVbis}-\ref{Pr:CyMQVter} is satisfied. 
Thus, the multiplicative quiver variety $\Cnqm$ (and its open subset $\Cnqm^\circ$) obtained by quasi-Hamiltonian reduction from $\MM_{Q_{\dfat},\nfat}^{\bullet}$ (and from $\MM_{Q_{\dfat},\nfat}^{\circ}$) is smooth, of dimension $2n|\dfat|$ for $|\dfat|:=\sum_{s\in \Z_m}d_s$.

From now on, we simply denote the varieties $\Cnqm$ and $\Cnqm^\circ$ as $\Cnm$ and $\Cnm^\circ$, respectively. We will derive the local structure of these spaces  following ideas from \cite{CF2,F1}. 
Let us emphasise that we work with the structure of complex algebraic variety of $\Cnm$, hence we consider the Zariski topology throughout this section. Since the variety is smooth, we could recast all the results by considering $\Cnm$ as a complex manifold and work with the corresponding classical complex topology. Indeed, all the functions involved are polynomials, hence analytic. When dealing with integrability in Section \ref{S:Int}, we will directly work with the structure of complex manifold of $\Cnm$.


\subsection{Coordinates on an open subset} \label{ss:Loc-coord}

\subsubsection{Alternative characterisation of $\Cnm^\circ$}  \label{ss:Loc-char}
Recall from \eqref{Eq:Zend} that, on the quasi-Hamiltonian variety $\MM_{Q_{\dfat},\nfat}^{\circ}\subset \MM_{Q_{\dfat},\nfat}^{\bullet}$ defined by the condition that $X=\sum_{s\in\Z_m} X_s$ is invertible, 
we can introduce $Z\in \Gl(nm)$ by $Z=\sum_{s\in\Z_m} Z_s$ for $Z_s=Y_s+X_s^{-1}$.  
For each $s\in \Z_m$, we also introduce 
\begin{equation}
\As^s\in \Mat(n\times d_s,\CC)\,, \quad \Cs^s\in \Mat(d_s\times n,\CC)\,, 
\end{equation}
which we refer to as the $s$-th \emph{spin matrices}, and that are defined entry-wise by 
  \begin{equation} \label{AsCsm}
\As_{i\alpha}^s=\left(W_{s,\alpha}\right)_i \, ,\quad
\Cs_{\alpha j}^s=\left(V_{s,\alpha}(\Id_n+W_{s,\alpha-1}V_{s,\alpha-1})\ldots (\Id_n+W_{s,1}V_{s,1})Z_{s-1}\right)_j\,.
  \end{equation}
(Note that $\Cs_{1 j}^s=(V_{s,1}Z_{s-1})_j$.)
We remark that the rows of the matrix $\Cs^s$ satisfy the recurrence relation
\begin{equation}\label{Eq:RecCs}
 \Cs_{\alpha j}^s= \sum_{\lambda=1}^{\alpha-1} (V_{s,\alpha}W_{s,\lambda})\Cs_{\lambda,j}^s + (V_{s,\alpha}Z_{s-1})_j\,,
\end{equation}
which is obtained as a special case of the following identity with $1\leq \beta\leq \alpha-1$
\begin{equation} \label{Eq:usefulCs}
 \sum_{\lambda=1}^{\beta} (V_{s,\alpha}W_{s,\lambda})\Cs_{\lambda j}^s + (V_{s,\alpha}Z_{s-1})_j
 =\Big(V_{s,\alpha} (\Id_n+W_{s,\beta}V_{s,\beta})(\Id_n+W_{s,1}V_{s,1})Z_{s-1}\Big)_j\,.
\end{equation} 
Denote by $E_{\alpha \beta}^{sr}\in \Mat(d_s \times d_r,\CC)$ the elementary matrix with entry $+1$ at $(\alpha,\beta)$ and zero everywhere else. 
Then, \eqref{AsCsm} entails for $s\in\Z_m$ and $1\leq \alpha \leq d_s$  
\begin{equation} \label{Eq:InvCs}
 Z_{s-1}+\sum_{\lambda=1}^\alpha \As^s E_{\lambda\lambda}^{ss} \Cs^s = (\Id_n+W_{s,\alpha}V_{s,\alpha})\ldots (\Id_n+W_{s,1}V_{s,1})Z_{s-1}\,,
\end{equation}
which is invertible because the matrices $(\Id_n+W_{s,\lambda}V_{s,\lambda})$ and $Z_{s-1}$ are invertible. From these observations, we can see that $\MM_{Q_{\dfat},\nfat}^{\circ}$ can be equivalently parametrised by the $4m$ matrices 
\begin{equation*}
X_s,Z_s\in \gl(n)\,, \quad 
\As^s\in \Mat(n\times d_s,\CC)\,, \quad \Cs^s\in \Mat(d_s\times n,\CC)\,, \quad s\in \Z_m,
\end{equation*}
subject to the invertibility conditions 
\begin{equation} \label{EqZCyc-inv}
 \det(X_s)\neq0,\quad \det(Z_s)\neq0,\quad \det\Big(Z_{s-1}+\sum_{1\leq \lambda\leq \alpha} \As^s E_{\lambda\lambda}^{ss} \Cs^s \Big)\neq 0,\quad 
\text{for } 1\leq \alpha \leq d_s,\,\, s\in \Z_m.
\end{equation}
Using \eqref{Eq:InvCs} with $\alpha=d_s$, the moment map equation \eqref{GeomCys} can be written in the form  
\begin{equation}
  \label{EqZCyc}
 X_sZ_sX_{s-1}^{-1}=q_sZ_{s-1}+ q_s \As^s \Cs^s\,, \quad \text{for all }s \in \Z_m\,.
\end{equation}
Thanks to this alternative characterisation, we have that the variety $\Cnm^\circ$ is parametrised by equivalence classes $[(X,Z,\As^s,\Cs^s)]$ of tuples $(X,Z,\As^s,\Cs^s)\in\MM_{Q_{\dfat},\nfat}^{\circ}$ satisfying \eqref{EqZCyc-inv}--\eqref{EqZCyc}, modulo the action 
\begin{equation} \label{Eq:ActCy}
 g\cdot (X,Z,\As^s,\Cs^s)=(g X g^{-1},g Z g^{-1},g_s \As^s,\Cs^s g_{s-1}^{-1})\,, \quad g=(g_s)_{s\in\Z_m}\in \Gl(n \delta)\,.
\end{equation}
Here, we identify $\Gl(n \delta)=\{1\}\times \Gl(n \delta)\subset \Gl(\nfat)$ with $\Gl(\nfat)/\CC^\times$, see the short discussion after \eqref{Eq:gact}.

\subsubsection{The slice} \label{ss:Loc-slice}

Introduce for each $s\in \Z_m$ the matrix $\Xs_s:=X_0\ldots X_s$. Then $\Cnm^\circ\subset \Cnm$ is the subvariety defined by the condition that the product $\Xs:=\Xs_{m-1}$ is invertible in each equivalence class.
With the notation of \ref{ss:Loc-char}, we can act on a representative $(X,Z,\As^s,\Cs^s)$ of a point in $\Cnm^\circ$ by $g_1=(\Id_n,\Xs_0, \ldots, \Xs_{m-2})\in \Gl(n\delta)$ so that $g_1\cdot X_s=\Id_n$ for all $s\neq m-1$ while $g_1 \cdot X_{m-1}=\Xs$. 

Assume furthermore that $\Xs$ is diagonalisable with eigenvalues $(x_1,\ldots,x_n)$ taking values in 
\begin{equation} \label{Eqhreg}
\hreg=\{x=(x_1,\ldots,x_n)\in\CC^n \, \mid \, x_i\ne 0\,,\ x_i\ne x_j\,,\ x_i\ne t x_j\ \text{for all}\ i\ne j\}\,.
\end{equation} 
Then, there exists $U\in \Gl(n)$ such that $U \Xs U^{-1}=\diag(x_1,\ldots,x_n)$. Acting by 
$g_\h=(U,\ldots,U)$, we get that $g_\h g_1\cdot X_s=\Id_n$ for all $s\neq m-1$ while $g_\h g_1 \cdot X_{m-1}=\diag(x_1,\ldots,x_n)$. Now, consider the open subspace defined by the condition that $a_{i} \neq0$ for all $i=1,\ldots,n$, where $a_{i} :=\sum_{\alpha=1}^{d_0} (g_\h g_1 \cdot\As^0)_{i \alpha}$. We can form the matrix $A=\diag(a_1^{-1},\ldots,a_n^{-1})$, then define $g_a = (A,\ldots,A)$. We find that $\sum_\alpha (g_a g_\h g_1 \cdot \As^0)_{i \alpha}=1$ for each $i=1,\ldots,n$. 

\begin{rem} \label{rem:Cnmp}
We define the subspace $\Cnm'\subset \Cnm^\circ$ as the subvariety where we can always perform the last two algebraic operations.  
In other words, each point of $\Cnm'$ admits a representative where 
\begin{equation}
 X_0,\ldots,X_{m-2}=\Id_n\,, \quad X_{m-1}=\diag(x_1,\ldots,x_n)\,, \quad  
\sum_{\alpha=1}^{d_0} \As^0_{i \alpha}=1 \text{ for }1\leq i \leq n\,,
\end{equation} 
with $(x_1,\ldots,x_n)\in \hreg$ \eqref{Eqhreg}. 
To see that $\Cnm'$ is not empty, note by adapting the proof of Proposition \ref{Pr:CyMQVbis} that it contains a copy of the variety corresponding to $d'=(1,0,\ldots,0)$ which is not empty by \cite{CF1}. 

It is important to remark that, in the case $\dfat=d'$ of \cite{CF1}, $\Cnm^\circ$ is irreducible by \cite{Ob} and thus $\Cnm'$ is a dense subset of $\Cnm^\circ$. We believe that a similar result holds for arbitrary $\dfat$. Hence, the construction of local coordinates that we provide below gives a set of local coordinates on the irreducible component of $\Cnm^\circ$ containing $\Cnm'$, which we conjecture to be $\Cnm^\circ$ itself. 
\end{rem}

We do a final transformation to have all the matrices $(X_s)_{s\in\Z_m}$  constituting $X$ in the same form. Consider $\lambda_i\in \CC^\times$ such that $\lambda_i^m=x_i$. In particular, $\lambda_i^m \neq \lambda_j^m$ and $\lambda_i^m \neq t \lambda_j^m$ for each $i \neq j$. Then acting by $g_\lambda=(\Id_n, \Lambda, \ldots, \Lambda^{m-2})$, where $\Lambda=\diag(\lambda_1^{-1},\ldots,\lambda_n^{-1})$, we have that 
\begin{equation}\label{Eq:gauge}
  (\hat{X},\hat{Z},\hat{\As}^s,\hat{\Cs}^s):=g_\lambda g_a g_\h g_1 \cdot (X,Z,\As^s,\Cs^s)\,,\quad s\in\Z_m\,,
\end{equation}
satisfies $\hat{X}_s=\diag(\lambda_1,\ldots,\lambda_n)$ for each $s\in \Z_m$, 
and  $\sum_{\alpha=1}^{d_0}\hat{\As}^{0}_{i\alpha}=+1$ for each $i=1,\ldots,n$.

\begin{lem} \label{L:gauge}
  The choice of gauge given by \eqref{Eq:gauge} completely determines the representative up to an action by the finite group $W=S_n \ltimes\Z^n_m$. 
\end{lem}
\begin{proof}
It is clear that the choice of diagonal form of $(\hat{X}_s)_s$ depends on both the ordering of $(x_1,\ldots,x_n)$, and the choice of $m$-th root of unity. 
To be precise, let $P_\sigma\in \Gl(n)$ be the permutation matrix corresponding to $\sigma\in S_n$ while we set $M_{\kfat}=\diag(\zeta^{k_1},\ldots,\zeta^{k_n})\in \Gl(n)$ for $\kfat=(k_1,\ldots,k_n)\in \Z^n_m$ where $\zeta$ is a fixed primitive $m$-th root of unity.
Then the action of an element $(\sigma,\kfat) \in S_n \ltimes\Z^n_m$ is represented by the matrix $g=(g_s)_{s\in\Z_m}$, where $g_s:= P_\sigma M_{\kfat}^{-s}$. 
The action of such a $g$ maps $\hat{X}_s$ to $\diag(\zeta^{k_{\sigma(1)}}\lambda_{\sigma(1)},\ldots,\zeta^{k_{\sigma(n)}}\lambda_{\sigma(n)})$ for any $s\in\Z_m$, so each $\hat{X}_s$ remains in the required form. To see that $\sum_{\alpha=1}^{d_0}(g \cdot\hat{\As}^{0})_{i\alpha}=1$, remark that the component of $g$ acting on $\VV_0$ is $g_0=P_\sigma M_{\kfat}^{-0}=P_\sigma$, hence  
$\sum_{\alpha=1}^{d_0}(g \cdot\hat{\As}^{0})_{i\alpha}=
\sum_{\alpha=1}^{d_0}(\hat{\As}^{0})_{\sigma^{-1}(i),\alpha}=1$. 
We can conclude because any $g=(g_s)_{s\in\Z_m}\in \Gl(n \delta)$ that maps a representative of the form \eqref{Eq:gauge} to another one is such that $g_s= P_\sigma M_{\kfat}^{-s}$ for some $\sigma\in S_n$ and $\kfat\in \Z^n_m$. 
\end{proof}

Note that if we set  $\aaa_i^{s\alpha}=\hat{\As}^{s}_{i\alpha}$ and $\ccc_j^{s \alpha}=\hat{\Cs}^s_{\alpha j}$, 
then the $(2|\dfat|+1)n$ variables $(\lambda_i,\aaa_i^{s\alpha},\ccc_i^{s\alpha})$ under the $n$ constraints $\sum_{\alpha=1}^{d_0}\aaa_i^{0\alpha}=1$ (and some invertibility conditions) characterise a representative $(\hat{X},\hat{Z},\hat{\As}^s,\hat{\Cs}^s)$ of the form \eqref{Eq:gauge} uniquely up to the action of $W$. Indeed, such variables are defined from $(\hat{X},\hat{\As}^s,\hat{\Cs}^s)$, and they determine the matrices  $(\hat{Z}_s)_{s\in\Z_m}$.  
To see this,  consider the following regular functions 
\begin{equation} \label{Eq:spin-g}
  \gfat_{ij}^s:=\sum_{\alpha=1}^{d_s}\aaa_i^{s\alpha} \ccc_j^{s\alpha}\,,\quad i,j=1,\ldots,n,\,\,\,  s \in \Z_m\,,
\end{equation}
which are the entries of the matrix $\hat{\As}^s \hat{\Cs}^s$ (assuming $d_s >0$, otherwise we set $\gfat_{ij}^s=0$). 
Then, for any $r=0,\ldots,m-1$, 
\begin{equation} \label{hatZij}
  (\hat{Z}_r)_{ij}=\sum_{s=0}^r \frac{t_r}{t_{s-1}}\frac{\lambda_i^{m+(s-r-1)}\lambda_j^{-(s-r-1)}}{\lambda_i^m - t \lambda_j^m} \gfat_{ij}^s 
+\sum_{s=r+1}^{m-1}\frac{tt_r}{t_{s-1}}\frac{\lambda_i^{s-r-1}\lambda_j^{m-(s-r-1)}}{\lambda_i^m - t \lambda_j^m} \gfat_{ij}^s \,.
\end{equation}
(Recall that each $t_s$ is defined in \eqref{Eq:tparam}, and note that $t_r/t_{-1}=t_r$ appears in the term  $s=0$.) 
To show that \eqref{hatZij} holds, remark that in our choice of gauge, if we multiply \eqref{EqZCyc} at entry $(i,j)$ by $\frac{t}{t_s}(\lambda_j/\lambda_i)^{m-s-1}$ and take the sum over $s\in\Z_m$ of all such equations, we get 
\begin{equation*}
  \left(\frac{\lambda_i}{\lambda_j} - t \frac{\lambda_j^{m-1}}{\lambda_i^{m-1}}  \right) (\hat{Z}_{m-1})_{ij}
= \sum_{s=0}^{m-1}\frac{t}{t_{s-1}}\left(\frac{\lambda_j}{\lambda_i}\right)^{m-s-1} \gfat_{ij}^s\,.
\end{equation*}
This yields in particular 
\begin{equation}
  (\hat{Z}_{m-1})_{ij}=\sum_{s=0}^{m-1} \frac{t}{t_{s-1}}\frac{\lambda_i^{s}\lambda_j^{m-s}}{\lambda_i^m - t \lambda_j^m} \gfat_{ij}^s \,. 
\end{equation}
This is exactly \eqref{hatZij} for $r=m-1$. We can then use  relation \eqref{EqZCyc} for $s=0$ to get $\hat{Z}_0$,  and finally get the other matrices by induction. 

This discussion allows us to introduce coordinates on $\Cnm'$ as follows. 
Consider the open affine subvariety of $\CC^n\times \CC^{|\dfat|n}\times \CC^{|\dfat|n}$ with coordinates $(\lambda_i,\aaa_i^{s\alpha},\ccc_i^{s\alpha})$ defined by the condition that $(\lambda_i^m)_{1\leq i \leq n}\in \hreg$ \eqref{Eqhreg}. 
Define $\hloc$ from the closed subvariety $\{\sum_{\alpha=1}^{d_0} \aaa_i^{0\alpha} \mid i=1,\ldots,n\}$ by requiring  that the matrices $(\hat{Z}_s)_{s\in\Z_m}$ defined by \eqref{hatZij} are invertible and, 
furthermore, if we introduce the matrices $(\hat{\As}^s,\hat{\Cs}^s)_{s\in\Z_m}$ by $\hat{\As}^s_{i\alpha}=\aaa_i^{s\alpha}$, $\hat{\Cs}^s_{\alpha i}=\ccc_i^{s\alpha}$, then the matrices appearing in \eqref{Eq:InvCs} are invertible. 
By construction,  $\hloc$ defined in this way characterises a representative in the form \eqref{Eq:gauge} of a point in $\Cnm'$. 

If $\zeta$ is an $m$-th root of unity,  an element  $\kfat=(k_1,\ldots,k_n)\in \Z^n_m$ acts on $\hloc$ by  
\begin{equation}
 \kfat \cdot (\lambda_i,\aaa_i^{s\alpha},\ccc_i^{s\alpha})=(\zeta^{k_i} \lambda_i,\zeta^{-s k_i}\aaa_i^{s\alpha},\zeta^{(s-1)k_i}\ccc_i^{s\alpha})\,.
\end{equation}
In particular $\kfat \cdot \ccc_i^{0\alpha}=\zeta^{(m-1)k_i}\ccc_i^{0\alpha}$. 
This implies that $\kfat \cdot \gfat_{ij}^s =\zeta^{s(k_j-k_i)-k_j} \gfat_{ij}^s$ for $s \neq 0$ (and $\kfat \cdot \gfat_{ij}^0 =\zeta^{(m-1)k_j} \gfat_{ij}^0$), so that  $g_{ii}^s\lambda_i$ is $\Z^n_m$-invariant. We also have an action of $S_n$ on  $\hloc$ which is the obvious one obtained by permutation. Altogether, this defines an action of $W=S_n \ltimes\Z^n_m$ on $\hloc$. 
The next result is then obtained when combined together with  Lemma \ref{L:gauge}. 

\begin{prop} \label{Pr:diffeo}
There is a $W$-covering $\xi:\hloc\to \Cnm'$ given by 
 \begin{equation*}
\xi (\lambda_i,\aaa_i^{s\alpha},\ccc_i^{s\alpha}) = [(X_s,Z_s,\As^s,\Cs^s)]\,,
 \end{equation*}
where $X_s=\diag(\lambda_1,\ldots,\lambda_n)$, $\As^s_{i\alpha}=\aaa_i^{s\alpha}$, $\Cs^s_{\alpha i}=\ccc_i^{s\alpha}$ while $Z_s$ is defined through \eqref{hatZij}. 
In particular, the induced map $\xi^{(W)}:\hloc/W\to \Cnm'$ is an isomorphism of algebraic varieties.
\end{prop}


\subsection{Poisson brackets of the coordinates} \label{ss:Loc-PB}

Recall that the multiplicative quiver variety $\Cnm$ (hence any of its open subvarieties) admits a Poisson bracket obtained by quasi-Hamiltonian reduction from $\MM_{Q_{\dfat},\nfat}^{\bullet}$, see \ref{sss:CycMQV-gen}. We use Proposition \ref{Pr:diffeo} to induce the Poisson bracket on the local coordinates $(\lambda_i,\aaa_i^{s\alpha},\ccc_i^{s\alpha})$ of $\hloc$. We also derive some identities that will be useful in \ref{ss:Loc-express} to write down locally some of the (induced) Hamiltonian functions for which the associated flows could be integrated over $\MM_{Q_{\dfat},\nfat}^{\bullet}$ before quasi-Hamiltonian reduction, see Section \ref{S:Subalg}.  

\subsubsection{General result} 

We first define an antisymmetric biderivation, denoted $\brloc{-,-}$, on $\hloc$. It turns out that $\brloc{-,-}$ satisfies the Jacobi identity, hence it is a Poisson bracket as we will see in Proposition \ref{Pr:CyPoi}. We can simply introduce $\brloc{-,-}$ by specifying it on the coordinates of $\hloc$, and then extend it by using the Leibniz rule in each argument and the antisymmetry property. The reader can check that the operation is well-defined by remarking that we have 
\begin{equation*}
 \sum_{\alpha=1}^{d_0} \brloc{\aaa_i^{0\alpha},-} = 0\,, \quad \text{ for } i=1,\ldots,n\,.
\end{equation*}

Below, we take $1\leq i,j\leq n$, and we recall that we defined the ordering function $o(-,-)$ after \eqref{Eq:qPbrack1}--\eqref{Eq:qPbrack3}. For $s\in\Z_m$, an element $Z^s_{ij}$ is defined through \eqref{hatZij} (we remove the hat for simplicity).
For any admissible spin indices $(s,\alpha)$ and $(r,\beta)$ (see Remark \ref{Rem:Convention}), we set 
\begin{equation}
  \brloc{\lambda_i,\lambda_j}=0, \quad \brloc{\lambda_i,\aaa^{s\alpha}_j}=0, \quad 
\brloc{\lambda_i,\ccc_j^{r\beta}}=\frac1m \delta_{ij}\lambda_i \ccc_j^{r\beta}. \label{EqCyFirst}
\end{equation}
For any $1\leq \beta, \epsilon\leq d_0$, we set 
\begin{equation} \label{Eq:CC00}
\begin{aligned}
     \brloc{\ccc_j^{0\epsilon},\ccc_i^{0\beta}}=&
\frac12 \delta_{(i \neq j)} \frac{\lambda_j^m+\lambda_i^m}{\lambda_j^m-\lambda_i^m} 
(\ccc_j^{0\epsilon}\ccc_i^{0\beta} + \ccc_i^{0\epsilon} \ccc_j^{0\beta})  \\
&+(Z_{m-1})_{ij} \ccc_i^{0 \beta} - (Z_{m-1})_{ji} \ccc_j^{0 \epsilon}  + \frac12 o(\epsilon,\beta)  (\ccc_i^{0\epsilon} \ccc_j^{0\beta} -\ccc_j^{0\epsilon}\ccc_i^{0\beta})  \\
&+\ccc_i^{0\beta} \sum_{\lambda=1}^{\epsilon-1} \aaa_i^{0\lambda} (\ccc_j^{0\lambda} - \ccc_j^{0\epsilon}) 
- \ccc_j^{0\epsilon} \sum_{\mu=1}^{\beta-1} \aaa_j^{0\mu} (\ccc_i^{0\mu} - \ccc_i^{0\beta})\,.
\end{aligned}
\end{equation}
For any $q\in \Z_m\setminus\{0\}$, $1\leq \epsilon\leq d_q$ and  $1\leq \beta\leq d_0$, we set 
\begin{equation} \label{Eq:CCq0}
\begin{aligned}
     \brloc{\ccc_j^{q\epsilon},\ccc_i^{0\beta}}=&\frac{q-m}{m}\delta_{ij}\ccc_j^{q\epsilon}\ccc_i^{0\beta} 
+ \delta_{(i \neq j)} \frac{\lambda_i^m}{\lambda_j^m-\lambda_i^m} 
\left(\ccc_j^{q\epsilon}\ccc_i^{0\beta} +\frac{\lambda_j^q}{\lambda_i^q} \ccc_i^{q\epsilon} \ccc_j^{0\beta}\right) \\
&- (Z_{m-1})_{ji} \ccc_j^{q \epsilon} 
- \ccc_j^{q\epsilon} \sum_{\mu=1}^{\beta-1} \aaa_j^{0\mu} (\ccc_i^{0\mu} - \ccc_i^{0\beta})\,.
\end{aligned}
\end{equation}
For any $r\in \Z_m\setminus\{0\}$, $1\leq \beta\leq d_r$ and  $1\leq \epsilon\leq d_0$, we set 
\begin{equation} \label{Eq:CC0r}
\begin{aligned}
     \brloc{\ccc_j^{0\epsilon},\ccc_i^{r\beta}}=&\frac{m-r}{m}\delta_{ij}\ccc_j^{0\epsilon}\ccc_i^{r\beta} 
+ \delta_{(i \neq j)} \frac{\lambda_j^m}{\lambda_j^m-\lambda_i^m} 
\left(\ccc_j^{0\epsilon}\ccc_i^{r\beta} +\frac{\lambda_j^{-r}}{\lambda_i^{-r}} \ccc_i^{0\epsilon} \ccc_j^{r\beta}\right)  \\
&+(Z_{m-1})_{ij} \ccc_i^{r \beta} 
+\ccc_i^{r\beta} \sum_{\lambda=1}^{\epsilon-1} \aaa_i^{0\lambda} (\ccc_j^{0\lambda} - \ccc_j^{0\epsilon}) \,.
\end{aligned}
\end{equation}
For any $q\in \Z_m\setminus\{0\}$ and $1\leq \beta, \epsilon\leq d_q$, we set 
\begin{equation} \label{Eq:CCqq}
\begin{aligned}
     \brloc{\ccc_j^{q\epsilon},\ccc_i^{q\beta}}=&
\frac12  \delta_{(i \neq j)} \frac{\lambda_j^m+\lambda_i^m}{\lambda_j^m-\lambda_i^m} 
\left(\ccc_j^{q\epsilon}\ccc_i^{q\beta} + \ccc_i^{q\epsilon} \ccc_j^{q\beta}\right) 
+\frac12 o(\epsilon,\beta)\left( \ccc_i^{q\epsilon} \ccc_j^{q\beta} - \ccc_j^{q\epsilon}\ccc_i^{q\beta} \right)\,.
\end{aligned}
\end{equation}
For any $0<r<q\leq m-1$, and for any $1\leq \beta\leq d_r$, $1\leq \epsilon\leq d_q$, we set 
\begin{equation} \label{Eq:CCq>r}
\begin{aligned}
     \brloc{\ccc_j^{q\epsilon},\ccc_i^{r\beta}}\stackrel{q > r}{\,=\,}&\,\,\,
\delta_{ij}\frac{q-r}{m} \ccc_j^{q\epsilon}\ccc_i^{r\beta} 
+ \delta_{(i \neq j)} \frac{\lambda_j^m}{\lambda_j^m-\lambda_i^m} \ccc_j^{q\epsilon}\ccc_i^{r\beta} 
+ \delta_{(i \neq j)} \frac{\lambda_i^m}{\lambda_j^m-\lambda_i^m} \frac{\lambda_j^{q-r}}{\lambda_i^{q-r}}  \ccc_i^{q\epsilon} \ccc_j^{r\beta}   \,.
\end{aligned}
\end{equation}
For any $0<q<r\leq m-1$, and for any $1\leq \beta\leq d_r$, $1\leq \epsilon\leq d_q$, we set 
\begin{equation} \label{Eq:CCq<r}
\begin{aligned}
     \brloc{\ccc_j^{q\epsilon},\ccc_i^{r\beta}}\stackrel{q < r}{\,=\,}&\,\,\,
\delta_{ij}\frac{q-r}{m} \ccc_j^{q\epsilon}\ccc_i^{r\beta} 
+ \delta_{(i \neq j)} \frac{\lambda_i^m}{\lambda_j^m-\lambda_i^m} \ccc_j^{q\epsilon}\ccc_i^{r\beta} 
+ \delta_{(i \neq j)} \frac{\lambda_j^m}{\lambda_j^m-\lambda_i^m} \frac{\lambda_j^{q-r}}{\lambda_i^{q-r}}  \ccc_i^{q\epsilon} \ccc_j^{r\beta}   \,.
\end{aligned}
\end{equation}
For any $1\leq \alpha, \epsilon\leq d_0$, we set 
\begin{equation} \label{Eq:CA00}
\begin{aligned}
     \brloc{\ccc_j^{0\epsilon},\aaa_i^{0\alpha}}=&
\frac12 \delta_{(i \neq j)} \frac{\lambda_j^m+\lambda_i^m}{\lambda_j^m-\lambda_i^m}  \ccc_j^{0\epsilon}  (\aaa_j^{0\alpha}-\aaa_i^{0\alpha})  -\delta_{(\alpha<\epsilon)}\aaa_i^{0\alpha} \ccc_j^{0\epsilon} 
+\delta_{\epsilon \alpha} \sum_{\lambda=1}^{\epsilon-1} \aaa_i^{0\lambda} \ccc_j^{0\lambda} \\
&+\delta_{\epsilon \alpha}(Z_{m-1})_{ij}-\aaa_i^{0\alpha} (Z_{m-1})_{ij}
-\aaa_i^{0\alpha} \sum_{\lambda=1}^{\epsilon-1}\aaa_i^{0\lambda} (\ccc_j^{0\lambda}-\ccc_j^{0\epsilon})  \\
&+\frac12 \ccc_j^{0\epsilon}  \sum_{\kappa=1}^{d_0} o(\alpha,\kappa)
 (\aaa_i^{0\alpha} \aaa_j^{0\kappa} +\aaa_j^{0\alpha}\aaa_i^{0\kappa}) \,.
\end{aligned}
\end{equation}
For any $s\in \Z_m\setminus\{0\}$,  $1\leq \epsilon\leq d_0$ and $1\leq \alpha\leq d_s$, we set 
\begin{equation} \label{Eq:CA0s}
\begin{aligned}
     \brloc{\ccc_j^{0\epsilon},\aaa_i^{s\alpha}}=&\frac{s-m}{m}\delta_{ij}\ccc_j^{0\epsilon}\aaa_i^{s\alpha} 
+ \delta_{(i \neq j)} \frac{\lambda_i^m}{\lambda_j^m-\lambda_i^m}\frac{\lambda_j^s}{\lambda_i^s} \ccc_j^{0\epsilon}  \aaa_j^{s\alpha} 
- \delta_{(i \neq j)} \frac{\lambda_j^m}{\lambda_j^m-\lambda_i^m} \ccc_j^{0\epsilon}  \aaa_i^{s\alpha} \\
&-\aaa_i^{s\alpha} (Z_{m-1})_{ij}
-\aaa_i^{s\alpha} \sum_{\lambda=1}^{\epsilon-1}\aaa_i^{0\lambda} (\ccc_j^{0\lambda}-\ccc_j^{0\epsilon}) \,.
\end{aligned}
\end{equation}
For any $q\in \Z_m\setminus\{0\}$ and  $1\leq \epsilon, \alpha\leq d_q$, we set 
\begin{equation} \label{Eq:CAqq}
\begin{aligned}
     \brloc{\ccc_j^{q\epsilon},\aaa_i^{q\alpha}}=&\frac{q}{m}\delta_{ij}\ccc_j^{q\epsilon}\aaa_i^{q\alpha} 
+ \delta_{(i \neq j)} \frac{\lambda_i^m}{\lambda_j^m-\lambda_i^m}  \ccc_j^{q\epsilon}
\left( \frac{\lambda_j^q}{\lambda_i^q}  \aaa_j^{q\alpha} - \aaa_i^{q\alpha} \right)  \\
& - \delta_{(\alpha < \epsilon)} \ccc_j^{q\epsilon}\aaa_i^{q \alpha} 
+\delta_{\alpha \epsilon}\sum_{\lambda=1}^{\epsilon-1}\aaa_i^{q\lambda} \ccc_j^{q\lambda} + \delta_{\alpha \epsilon}  (Z_{q-1})_{ij}\,.
\end{aligned}
\end{equation}
For any $q\in \Z_m\setminus\{0\}$, $1\leq \epsilon\leq d_q$ and $1\leq \alpha\leq d_0$, we set 
\begin{equation} \label{Eq:CAq0}
\begin{aligned}
     \brloc{\ccc_j^{q\epsilon},\aaa_i^{0\alpha}}=& 
\frac12 \delta_{(i \neq j)} \frac{\lambda_j^m+\lambda_i^m}{\lambda_j^m-\lambda_i^m}
\ccc_j^{q\epsilon} \left( \aaa_j^{0\alpha} - \aaa_i^{0\alpha} \right)  \\
&+\frac12 \ccc_j^{q\epsilon}  \sum_{\kappa=1}^{d_0} o(\alpha,\kappa)
 (\aaa_i^{0\alpha} \aaa_j^{0\kappa} +\aaa_j^{0\alpha}\aaa_i^{0\kappa}) \,.
\end{aligned}
\end{equation}
For any $q,s\in \Z_m \setminus \{0\}$ with $q \neq s$, and for any $1\leq \epsilon\leq d_q$, $1\leq \alpha\leq d_s$, we set 
\begin{equation} \label{Eq:CAqNs}
\begin{aligned}
     \brloc{\ccc_j^{q\epsilon},\aaa_i^{s\alpha}}\stackrel{q \neq s}{\,=\,}&\frac{s}{m}\delta_{ij}\ccc_j^{q\epsilon}\aaa_i^{s\alpha} 
+ \delta_{(i \neq j)} \frac{\lambda_i^m}{\lambda_j^m-\lambda_i^m}  \ccc_j^{q\epsilon}
\left( \frac{\lambda_j^s}{\lambda_i^s}  \aaa_j^{s\alpha} - \aaa_i^{s\alpha} \right) 
 - \delta_{(s < q)} \ccc_j^{q\epsilon}\aaa_i^{q \alpha} \,.
\end{aligned}
\end{equation}
For any $1\leq \alpha, \gamma\leq d_0$, we set 
\begin{equation} \label{Eq:AA00}
\begin{aligned}
     \brloc{\aaa_j^{0\gamma},\aaa_i^{0\alpha}}=&
\frac12 \delta_{(i \neq j)} \frac{\lambda_j^m+\lambda_i^m}{\lambda_j^m-\lambda_i^m}  
\left(\aaa_j^{0\gamma}\aaa_i^{0\alpha} + \aaa_i^{0\gamma}\aaa_j^{0\alpha} - \aaa_j^{0\gamma}\aaa_j^{0\alpha} - \aaa_i^{0\gamma}\aaa_i^{0\alpha}\right)  \\
&+ \frac12 o(\alpha,\gamma) (\aaa_j^{0\gamma}\aaa_i^{0\alpha}+\aaa_i^{0\gamma}\aaa_j^{0\alpha})\\  
&+ \frac12 \aaa_i^{0\alpha} \sum_{\sigma=1}^{d_0} o(\gamma,\sigma)
 (\aaa_j^{0\gamma} \aaa_i^{0\sigma}+ \aaa_i^{0\gamma} \aaa_j^{0\sigma} )  \\
&- \frac12 \aaa_j^{0\gamma} \sum_{\kappa=1}^{d_0} o(\alpha,\kappa)
 (\aaa_i^{0\alpha} \aaa_j^{0\kappa} +\aaa_j^{0\alpha}\aaa_i^{0\kappa}) \,.
\end{aligned}
\end{equation}
For any $p\in \Z_m\setminus \{0\}$ and $1\leq \alpha, \gamma\leq d_p$, we set
\begin{equation} \label{Eq:AApp}
\begin{aligned}
     \brloc{\aaa_j^{p\gamma},\aaa_i^{p\alpha}}=&
\frac12 \delta_{(i \neq j)} \frac{\lambda_j^m+\lambda_i^m}{\lambda_j^m-\lambda_i^m}  
\left(\aaa_j^{p\gamma}\aaa_i^{p\alpha} + \aaa_i^{p\gamma}\aaa_j^{p\alpha}\right) 
+ \frac12 o(\alpha,\gamma) (\aaa_j^{p\gamma}\aaa_i^{p\alpha}+\aaa_i^{p\gamma}\aaa_j^{p\alpha})\\  
&- \delta_{(i \neq j)} \frac{\lambda_i^m}{\lambda_j^m-\lambda_i^m} \frac{\lambda_j^p}{\lambda_i^p}  \aaa_j^{p\gamma}\aaa_j^{p\alpha}
- \delta_{(i \neq j)} \frac{\lambda_j^m}{\lambda_j^m-\lambda_i^m} \frac{\lambda_i^p}{\lambda_j^p} \aaa_i^{p\gamma}\aaa_i^{p\alpha}\,.
\end{aligned}
\end{equation}
For any $s\in \Z_m\setminus \{0\}$, $1\leq \gamma\leq d_0$ and $1\leq \alpha\leq d_s$, we set 
\begin{equation} \label{Eq:AA0s}
\begin{aligned}
     \brloc{\aaa_j^{0\gamma},\aaa_i^{s\alpha}}=& 
\frac12 \delta_{(i \neq j)} \frac{\lambda_j^m+\lambda_i^m}{\lambda_j^m-\lambda_i^m}  
\left(\aaa_j^{0\gamma}\aaa_i^{s\alpha} - \aaa_i^{0\gamma}\aaa_i^{s\alpha}\right) \\
&+\delta_{(i \neq j)} \frac{\lambda_i^m}{\lambda_j^m-\lambda_i^m} \frac{\lambda_j^s}{\lambda_i^s}  \left( \aaa_i^{0\gamma}\aaa_j^{s\alpha} -  \aaa_j^{0\gamma}\aaa_j^{s\alpha} \right) \\
&+ \frac12 \aaa_i^{s\alpha} \sum_{\sigma=1}^{d_0} o(\gamma,\sigma) 
 (\aaa_j^{0\gamma} \aaa_i^{0\sigma} +\aaa_i^{0\gamma}\aaa_j^{0\sigma})\,.
\end{aligned}
\end{equation}
For any $p\in \Z_m\setminus \{0\}$, $1\leq \gamma\leq d_p$ and $1\leq \alpha\leq d_0$, we set 
\begin{equation} \label{Eq:AAp0}
\begin{aligned}
     \brloc{\aaa_j^{p\gamma},\aaa_i^{0\alpha}}=& 
\frac12 \delta_{(i \neq j)} \frac{\lambda_j^m+\lambda_i^m}{\lambda_j^m-\lambda_i^m}  
\left(\aaa_j^{p\gamma}\aaa_i^{0\alpha} - \aaa_j^{p\gamma}\aaa_j^{0\alpha}\right) \\
&+\delta_{(i \neq j)} \frac{\lambda_j^m}{\lambda_j^m-\lambda_i^m} \frac{\lambda_j^{-p}}{\lambda_i^{-p}}  \left( \aaa_i^{p\gamma}\aaa_j^{0\alpha} -  \aaa_i^{p\gamma}\aaa_i^{0\alpha} \right) \\
&- \frac12 \aaa_j^{p \gamma} \sum_{\kappa=1}^{d_0} o(\alpha,\kappa)
 (\aaa_i^{0\alpha} \aaa_j^{0\kappa} +\aaa_j^{0\alpha}\aaa_i^{0\kappa}) \,.
\end{aligned}
\end{equation}
For any $0<p<s \leq m-1$, and for any $1\leq \gamma\leq d_p$,  $1\leq \alpha\leq d_s$,  we set 
\begin{equation} \label{Eq:AAp<s}
\begin{aligned}
     \brloc{\aaa_j^{p\gamma},\aaa_i^{s\alpha}}   \stackrel{p<s}{\,=\,}&\,\,\,
-\delta_{ij} \aaa_j^{p\gamma}\aaa_i^{s\alpha}
+ \delta_{(i \neq j)} \frac{\lambda_i^m}{\lambda_j^m-\lambda_i^m} 
\left(\aaa_j^{p\gamma}\aaa_i^{s\alpha} + \frac{\lambda_j^{s-p}}{\lambda_i^{s-p}}\aaa_i^{p\gamma}\aaa_j^{s\alpha} \right) \\
&-\delta_{(i \neq j)} \frac{\lambda_i^m}{\lambda_j^m-\lambda_i^m} \frac{\lambda_j^{s}}{\lambda_i^{s}}  \aaa_j^{p\gamma}\aaa_j^{s\alpha}  
-\delta_{(i \neq j)} \frac{\lambda_j^m}{\lambda_j^m-\lambda_i^m} \frac{\lambda_j^{-p}}{\lambda_i^{-p}}  \aaa_i^{p\gamma}\aaa_i^{s\alpha}  \,.
\end{aligned}
\end{equation}
For any $0<s<p \leq m-1$, and for any $1\leq \gamma\leq d_p$,  $1\leq \alpha\leq d_s$,  we set
\begin{equation} \label{Eq:AAp>s}
\begin{aligned}
 \brloc{\aaa_j^{p\gamma},\aaa_i^{s\alpha}}   \stackrel{p>s}{\,=\,}&\,\,\,
\delta_{ij} \aaa_j^{p\gamma}\aaa_i^{s\alpha}
+ \delta_{(i \neq j)} \frac{\lambda_j^m}{\lambda_j^m-\lambda_i^m} 
\left(\aaa_j^{p\gamma}\aaa_i^{s\alpha} + \frac{\lambda_j^{s-p}}{\lambda_i^{s-p}}\aaa_i^{p\gamma}\aaa_j^{s\alpha} \right) \\
&-\delta_{(i \neq j)} \frac{\lambda_i^m}{\lambda_j^m-\lambda_i^m} \frac{\lambda_j^{s}}{\lambda_i^{s}}  \aaa_j^{p\gamma}\aaa_j^{s\alpha}  
-\delta_{(i \neq j)} \frac{\lambda_j^m}{\lambda_j^m-\lambda_i^m} \frac{\lambda_j^{-p}}{\lambda_i^{-p}}  \aaa_i^{p\gamma}\aaa_i^{s\alpha}  \,.
\end{aligned}
\end{equation}

Using these formulas and the map $\xi$ from Proposition \ref{Pr:diffeo}, we can establish the next result which is proven in \ref{ss:Loc-proof}. 
\begin{prop} \label{Pr:CyPoi}
The regular map $\xi: \hloc\to \Cnm'$ is a morphism of Poisson varieties for the Poisson bracket $\brloc{-,-}$ defined on $\hloc$ by \eqref{EqCyFirst}--\eqref{Eq:AAp>s}, and the Poisson bracket $\br{-,-}$ inherited by $\Cnm'$ from the one on $\Cnm$.  
\end{prop}

Observe that the Poisson bracket $\brloc{-,-}$ is $W$-invariant for the action on $\hloc$ introduced before Proposition \ref{Pr:diffeo}. 
This means that if $f_1,f_2$ are any two coordinate functions from $(\lambda_i,\aaa_i^{s,\alpha},\ccc_i^{s,\alpha})$, we have  
\begin{equation*}
 \brloc{w\cdot f_1,w\cdot f_2}=w\cdot \brloc{f_1,f_2}\,, \quad \text{ for all }w\in W\,.
\end{equation*}

\begin{cor}\label{cor:CyPoi}
 The map $\xi^{(W)}: \hloc/W\to \Cnm'$ is an isomorphism of Poisson varieties. 
\end{cor}

\begin{rem}
Consider the $2nd_0+n$ variables $(\lambda_i,\aaa_i^\alpha:=\aaa_i^{0\alpha},\ccc_i^\alpha:=\ccc_i^{0\alpha})_{1\leq i\leq n}^{1\leq\alpha\leq d_0}$ with 
 the Poisson brackets \eqref{EqCyFirst}, \eqref{Eq:CC00}, \eqref{Eq:CA00} and \eqref{Eq:AA00}. It is easy to check that this takes the form of the Poisson structure associated with the spin trigonometric RS system from \cite{CF2}, written as in \cite[Proposition 4.1]{F1} for $B:=Z_{m-1}$ and $x_i:=\lambda_i^m$. This happens because the Poisson brackets that we will derive in Lemma \ref{LemPoiCy} (for a suitable choice of regular functions on $\Cnm^\circ$)  can be restricted to the case $\dfat=(d_0,0,\ldots,0)$ where they are written exactly in that form in \cite[(B.1)-(B.2)]{F1}. 
\end{rem}

\subsubsection{Generalised Arutyunov-Frolov Poisson brackets} 

We can explicitly determine the Poisson bracket on $\hloc$ in terms of the regular functions $\lambda_i$ and $\gfat^s_{ij}$ \eqref{Eq:spin-g}  which fully characterise the entries of the matrices $(X,Z)$ in the slice constructed in \ref{ss:Loc-slice}. As a first step, we see from \eqref{EqCyFirst} and \eqref{Eq:spin-g} that 
\begin{equation} \label{Eq:PBlambg}
 \brloc{\lambda_i,\lambda_k}=0\,, \quad \brloc{\lambda_i,\gfat_{kl}^s}=\frac1m \delta_{il} \lambda_i \gfat_{kl}^s\,,
\end{equation}
for any $1\leq i,k,l\leq n$ and $s\in\Z_m$. 
Next, by combining \eqref{Eq:spin-g} and the expressions \eqref{Eq:CC00}, \eqref{Eq:CA00}, \eqref{Eq:AA00} of $\brloc{-,-}$, we arrive at the following result.

 \begin{lem} \label{ggMed00}
  For any $1\leq \epsilon,\gamma\leq d_0$ and $i,j,k,l=1,\ldots,n$, 
 \begin{equation*}
  \begin{aligned}
   \brloc{\ccc_j^{0\epsilon},\gfat^0_{kl}}=&(Z_{m-1})_{kj}\ccc^{0\epsilon}_l-(Z_{m-1})_{jl}\ccc_j^{0\epsilon}
 +((Z_{m-1})_{lj}-(Z_{m-1})_{kj})\gfat^0_{kl} \\
&+\frac12 \delta_{(j\neq k)}\frac{\lambda^m_j+\lambda^m_k}{\lambda^m_j-\lambda^m_k}\ccc_j^{0\epsilon} 
 (\gfat^0_{jl}-\gfat^0_{kl}) 
+\frac12 \delta_{(j\neq l)}\frac{\lambda^m_j+\lambda^m_l}{\lambda^m_j-\lambda^m_l}(\ccc_j^{0\epsilon} \gfat^0_{kl}+\ccc_l^{0\epsilon}\gfat^0_{kj}) \\
& +\frac12 \ccc_l^{0\epsilon} \gfat^0_{kj}-\frac12 \ccc_j^{0\epsilon} \gfat^0_{jl}  
+\gfat^0_{kl} \sum_{\lambda=1}^{\epsilon-1} (\ccc_j^{0\lambda}-\ccc_j^{0\epsilon})(\aaa_l^{0\lambda}-\aaa_k^{0\lambda})\,, 
  \end{aligned}
 \end{equation*}
 \begin{equation*}
  \begin{aligned}
 \brloc{\aaa_i^{0\gamma},\gfat^0_{kl}}=& \aaa_i^{0\gamma} (Z_{m-1})_{il}-\aaa_k^{0\gamma} (Z_{m-1})_{il}
 +\frac12 \delta_{(i\neq k)}\frac{\lambda^m_i+\lambda^m_k}{\lambda^m_i-\lambda^m_k}(\aaa_k^{0\gamma}-\aaa_i^{0\gamma})(\gfat^0_{il}-\gfat^0_{kl}) \\ 
&+\frac12 \delta_{(i\neq l)}\frac{\lambda^m_i+\lambda^m_l}{\lambda^m_i-\lambda^m_l}\gfat^0_{kl}(\aaa_{l}^{0\gamma}-\aaa_i^{0\gamma})
 +\frac12 \aaa_i^{0\gamma} \gfat^0_{il}-\frac12 \aaa_k^{0\gamma} \gfat^0_{il}  \\
 &+\frac12 \sum_{\sigma=1}^{d_0} o(\gamma,\sigma) \gfat^0_{kl}[\aaa_i^{0\gamma}(\aaa_k^{0\sigma}-\aaa_l^{0\sigma})
 +\aaa_i^{0\sigma}(\aaa_k^{0\gamma}-\aaa_l^{0\gamma})]\,.
  \end{aligned}
 \end{equation*}
 \end{lem}
Let us note the identity
\begin{equation}
(Z_{m-1})_{ij}+\frac12 \gfat^0_{ij}=\sum_{s=1}^{m-1}\frac{t}{t_{s-1}}\frac{\lambda_i^{s}\lambda_j^{m-s}}{\lambda_i^m-t\lambda_j^m}\gfat_{ij}^s + \frac12 \frac{\lambda_i^m+t\lambda_j^m}{\lambda_i^m-t\lambda_j^m}\gfat_{ij}^0\,, 
\label{Eq:Zg0}  
\end{equation}
 which follows directly from \eqref{hatZij}. 
Using \eqref{Eq:spin-g} again together with Lemma \ref{ggMed00}, we can get by straightforward computations that 
 \begin{equation}
  \begin{aligned}
   &\brloc{\gfat^0_{ij},\gfat^0_{kl}}=
 \frac{1}{2}\gfat^0_{ij}\gfat^0_{kl} \left[\frac{\lambda^m_i+\lambda^m_k}{\lambda^m_i-\lambda^m_k} + 
 \frac{\lambda^m_j+\lambda^m_l}{\lambda^m_j-\lambda^m_l}+\frac{\lambda^m_k+\lambda^m_j}{\lambda^m_k-\lambda^m_j} + 
 \frac{\lambda^m_l+\lambda^m_i}{\lambda^m_l-\lambda^m_i}\right] \\
 &\quad+\frac{1}{2}\gfat^0_{il}\gfat^0_{kj} \left[\frac{\lambda^m_i+\lambda^m_k}{\lambda^m_i-\lambda^m_k} + 
 \frac{\lambda^m_j+\lambda^m_l}{\lambda^m_j-\lambda^m_l}+\frac{\lambda^m_k+t\lambda^m_j}{\lambda^m_k-t\lambda^m_j} - 
 \frac{\lambda^m_i+t\lambda^m_l}{\lambda^m_i-t\lambda^m_l}\right] \\
 &\quad+\frac{1}{2}\gfat^0_{ij}\gfat^0_{il} \left[\frac{\lambda^m_k+\lambda^m_i}{\lambda^m_k-\lambda^m_i} 
 +\frac{\lambda^m_i+t\lambda^m_l}{\lambda^m_i-t\lambda^m_l} \right]+\frac{1}{2}\gfat^0_{ij}\gfat^0_{jl} 
 \left[\frac{\lambda^m_j+\lambda^m_k}{\lambda^m_j-\lambda^m_k} -\frac{\lambda^m_j+t\lambda^m_l}{\lambda^m_j-t\lambda^m_l} \right] \\
 &\quad+\frac{1}{2}\gfat^0_{kj}\gfat^0_{kl} \left[\frac{\lambda^m_k+\lambda^m_i}{\lambda^m_k-\lambda^m_i} 
 -\frac{\lambda^m_k+t\lambda^m_j}{\lambda^m_k-t\lambda^m_j} \right]+\frac{1}{2}\gfat^0_{lj}\gfat^0_{kl} 
 \left[\frac{\lambda^m_i+\lambda^m_l}{\lambda^m_i-\lambda^m_l} +\frac{\lambda^m_l+t\lambda^m_j}{\lambda^m_l-t\lambda^m_j} \right] \\
&\quad+\sum_{s=1}^{m-1}\frac{t}{t_{s-1}}\frac{\lambda_i^s \lambda_l^{m-s}}{\lambda_i^m-t \lambda_l^m}\gfat^s_{il} (\gfat_{ij}^0-\gfat_{kj}^0) 
+\sum_{s=1}^{m-1}\frac{t}{t_{s-1}}\frac{\lambda_k^s \lambda_j^{m-s}}{\lambda_k^m-t \lambda_j^m}\gfat^s_{kj} (\gfat_{il}^0-\gfat_{kl}^0) \\
&\quad-\sum_{s=1}^{m-1}\frac{t}{t_{s-1}}\frac{\lambda_j^s \lambda_l^{m-s}}{\lambda_j^m-t \lambda_l^m}\gfat^s_{jl} \gfat_{ij}^0  
+\sum_{s=1}^{m-1}\frac{t}{t_{s-1}}\frac{\lambda_l^s \lambda_j^{m-s}}{\lambda_l^m-t \lambda_j^m}\gfat^s_{lj} \gfat_{kl}^0 
\,. \label{Eq:Pbarg00}
  \end{aligned}
 \end{equation}
(Here, we take the convention that all the terms with a vanishing denominator should be omitted, 
e.g. the first term  stands for $\frac{1}{2}\gfat^0_{ij}\gfat^0_{kl}\delta_{(i\neq k)}\frac{\lambda^m_i+\lambda^m_k}{\lambda^m_i-\lambda^m_k}$.) 
Repeating similar computations to derive the Poisson brackets $\brloc{\gfat^r_{ij},\gfat^s_{kl}}$ for any $r,s\in \Z_m$, we can obtain analogous expressions\footnote{For the explicit form of these Poisson brackets, see \cite[\S5.2.4]{Fth}.} which are quadratic in the functions $\{\gfat_{ab}^t\mid a,b=i,j,k,l;\, t\in \Z_m\}$ and rational in $(\lambda_i)_{i=1}^n$. 
We obtain in this way the following result. 
\begin{cor} \label{Cor:PoilgCyc}
The commutative algebra generated by all the elements 
\begin{equation*}
\lambda_i,\,\,\gfat_{ij}^s,\quad i,j=1,\ldots,n,\quad s\in \Z_m\,, 
\end{equation*}
is, after localisation at  $S=\{\lambda_i^m-\lambda_j^m,\,\,\lambda_i^m-t\lambda_j^m\mid i\neq j\}$, a Poisson algebra under $\brloc{-,-}$.  
\end{cor}

\begin{rem}
 In the case $\dfat=(d,0,\ldots,0)$, only the coordinates $(\lambda_i,\aaa_i^{0\alpha},\ccc_i^{0\alpha})$ are not trivially zero (and so do the functions $\gfat_{ij}^0$). It can be checked that by plugging 
 \begin{equation}
x_i:=\lambda_i^m,\quad  a_i^\alpha:=\aaa_i^{0\alpha},\quad b_i^\alpha:=\ccc_i^{0\alpha}, 
 \end{equation}
the Poisson brackets \eqref{EqCyFirst}, \eqref{Eq:CC00}, \eqref{Eq:CA00} and \eqref{Eq:AA00} take the form given in \cite[Proposition 2.3]{CF2}. 
To do so, one needs to remark that $\gfat_{ij}^s$ vanishes if $s\neq 0$ in that case, hence our $\gfat_{ij}^0$ and $(Z_{m-1})_{ij}$ respectively become $f_{ij}$ and $Z_{ij}$ from \cite[(2.9)]{CF2}. 
Furthermore, under these conditions the Poisson bracket $\brloc{\gfat^0_{ij},\gfat^0_{kl}}$ \eqref{Eq:Pbarg00} takes a form conjectured by Arutyunov and Frolov \cite{AF} in relation to the spin Ruijsenaars-Schneider system; this conjecture was proven in \cite{CF2}. 
Thus, we can see the Poisson bracket \eqref{Eq:Pbarg00} as a generalisation of the Arutyunov-Frolov Poisson bracket in the presence of several types of spin variables $(\aaa_i^{s\alpha},\ccc_i^{s\alpha})$.
\end{rem}

\subsection{Proof of Proposition \ref{Pr:CyPoi}} \label{ss:Loc-proof} 

The idea of the proof consists in checking that, for a suitable set of regular functions $F_1,F_2$ on $\Cnm'$, we have that the identity 
\begin{equation} \label{EqPfCyFF}
 \xi^\ast\br{F_1,F_2}=\brloc{\xi^\ast F_1,\xi^\ast F_2}
\end{equation}
is identically satisfied. 

\subsubsection{Outline of the proof}
It is convenient to consider the following regular functions on $\Cnm^\circ$  
\begin{equation} \label{Eq:Cyfg}
  f^k:= \tr(X^k)\,, \quad g^l_{s \alpha,r \beta}:=\tr(\As^s E_{\alpha \beta}^{sr} \Cs^r X^l)\,,
\end{equation} 
for any $k,l\in \N$, $s,r\in \Z_m$, $\alpha=1,\ldots,d_s$, and $\beta=1,\ldots,d_r$. Clearly, $f^k=0$ if $k$ does not satisfy $k\underset{m}{\equiv} 0$, while $g^l_{s \alpha,r \beta}=0$ whenever the condition $l\underset{m}{\equiv} s-(r-1)$ is not satisfied. 
(The symbol $\underset{m}{\equiv}$ means that we take the equality  modulo $m$, i.e. we consider the identity in $\Z_m$.)  
We will compute the Poisson bracket $\br{-,-}$  on  $\Cnm^\circ$ between these functions $(f^k,g^l_{s \alpha,r \beta})$ in \ref{sss:Loc-GlobBr}.
In the meantime, we note that on $\Cnm'$ we can use the expressions given in \eqref{Eq:Cyfg} and the map $\xi$ from Proposition \ref{Pr:diffeo} to write  locally  
\begin{equation}\label{xiCyfg}
  \xi^\ast f^k=m\sum_{i=1}^n \lambda_i^k\,, \quad 
\xi^\ast g^l_{s \alpha,r \beta}=\sum_{i=1}^n \aaa_i^{s \alpha} \ccc_i^{r \beta}\lambda_i^l\,, \quad 
\sum_{\alpha=1}^{d_0}\xi^\ast g^{l'}_{0 \alpha,r \beta}=\sum_{i=1}^n  \ccc_i^{r \beta}\lambda_i^{l'}\,, 
\end{equation} 
assuming that $k\underset{m}{\equiv} 0$, while  $l\underset{m}{\equiv} s-(r-1)$ and $l'\underset{m}{\equiv} 1-r$. From these expressions, it is not difficult to see that the differentials of the regular functions (taken with any possible indices)
\begin{equation} \label{xiCyfgGen}
  f^k\,,\quad \sum_{1\leq \alpha\leq d_0} g^{l'}_{0\alpha, r \beta}\,,\quad  g^l_{s \alpha,01}\,,
\end{equation}
generate the cotangent space at a generic point of $\Cnm'$ where the $\ccc_i^{01}$ are nonzero. (This is a non-empty condition as it contains the case $d'=(1,0,\ldots,0)$ from \cite{CF1}.) It then follows that if we can check that \eqref{EqPfCyFF} holds for such functions, then \eqref{EqPfCyFF} holds for any two regular functions $F_1,F_2$ over $\Cnm'$. Thus $\xi$ intertwines the antisymmetric biderivation $\brloc{-,-}$ and the Poisson bracket $\br{-,-}$. 
Finally, for $W=S_n \ltimes\Z^n_m$ note that at a generic point of $\Cnm'$ we can write a $W$-symmetric function on $\hloc$ in terms of the functions \eqref{xiCyfg}. This implies that $\brloc{-,-}$ satisfies the Jacobi identity when restricted to $W$-symmetric functions. 
Hence, since $S_n \ltimes\Z^n_m$ is finite, it must satisfies the Jacobi identity identically. 

To finish the proof, we only need to check that \eqref{EqPfCyFF} holds on the functions \eqref{xiCyfgGen}. The Poisson brackets of these functions will be computed in \ref{sss:Loc-GlobBr}. We will then be exhaustive and explain\footnote{We will give the full computations only in one case, because the other cases are all derived in a similar way through tedious calculations. Details of these derivations can be found in \cite[\S5.2.3]{Fth}.} how to check the identities \eqref{EqCyFirst}--\eqref{Eq:AAp>s} one after another through \ref{sss:Loc-Lambda}--\ref{sss:Loc-AA}.

\subsubsection{Brackets of global functions} \label{sss:Loc-GlobBr}

To compute the Poisson bracket between the functions $(f^k,g^l_{s \alpha,r \beta})$ given in \eqref{Eq:Cyfg}, it suffices to take representatives (denoted in the same way) over the quasi-Poisson variety $\MM_{Q_{\dfat},\nfat}^{\circ}$ where we can compute their quasi-Poisson bracket before projecting it to $\Cnm^\circ$. The next result is found in this way. 

\begin{lem} \label{LemPoiCy}
 For any $k,l\in\N$ and spin indices $(s,\alpha)$, $(r,\beta)$, $(p,\gamma)$ and $(q,\epsilon)$, 
{ \allowdisplaybreaks  
\begin{subequations}
 \begin{align}
&\br{f^k,f^l}
=0\,, \quad 
\br{f^k,g^l_{s \alpha,r \beta}}
=\,k\, g^{k+l}_{s \alpha,r \beta}\,, \\
&\br{g^k_{p\gamma,q \epsilon},g^l_{s\alpha,r \beta}}
=\,\frac12 \left(\sum_{v=1}^k-\sum_{v=1}^l \right) 
\tr(\As^s E^{sr}_{\alpha \beta} \Cs^r X^v \As^p E^{pq}_{\gamma\epsilon} \Cs^q X^{k+l-v}) \nonumber \\
&\qquad\quad+\frac12 \left(\sum_{v=1}^k-\sum_{v=1}^l \right) 
\tr(\As^s E^{sr}_{\alpha \beta} \Cs^r X^{k+l-v} \As^p E^{pq}_{\gamma\epsilon} \Cs^q X^v)\nonumber \\
&\qquad\quad+\frac12 [o(p,r)-o(p,s)+o(q,s)-o(q,r)] 
\tr (\As^s E^{sq}_{\alpha \epsilon} \Cs^q X^k \As^p E^{pr}_{\gamma\beta} \Cs^r X^{l}) \nonumber\\
&\qquad\quad+\frac12 \delta_{ps} o(\alpha,\gamma) 
[\tr(\As^p E^{pq}_{\gamma \epsilon} \Cs^q X^k \As^s E^{sr}_{\alpha \beta} \Cs^r X^{l}) 
+ \tr(\As^s E^{sq}_{\alpha \epsilon} \Cs^q X^k \As^p E^{pr}_{\gamma\beta} \Cs^r X^{l})] \nonumber\\
&\qquad\quad+\frac12 \delta_{qr}o(\epsilon,\beta)
[\tr(\As^s E^{sr}_{\alpha \beta} \Cs^r X^{k} \As^p E^{pq}_{\gamma \epsilon} \Cs^q X^l ) 
- \tr(\As^s E^{sq}_{\alpha \epsilon} \Cs^q X^k \As^p E^{pr}_{\gamma\beta} \Cs^r X^{l})] \nonumber\\
&\qquad\quad+\frac12 \delta_{qs}[o(\epsilon,\alpha)+\delta_{\alpha \epsilon}] 
\tr(\As^s E^{sq}_{\alpha \epsilon} \Cs^q X^k \As^p E^{pr}_{\gamma\beta} \Cs^r X^{l}) \nonumber\\
&\qquad\quad-\frac12 \delta_{pr} [o(\beta,\gamma)+\delta_{\beta \gamma}] 
\tr(\As^s E^{sq}_{\alpha \epsilon} \Cs^q X^k \As^p E^{pr}_{\gamma\beta} \Cs^r X^{l}) \nonumber\\
&\qquad\quad+\delta_{qs}\delta_{\alpha \epsilon} \tr(Z X^k \As^p E^{pr}_{\gamma\beta} \Cs^r X^{l})
+\delta_{qs}\delta_{\alpha \epsilon} \sum_{\lambda=1}^{\epsilon-1} \tr(\As^s E^{ss}_{\lambda \lambda}\Cs^s X^k \As^p E^{pr}_{\gamma\beta} \Cs^r X^{l})  \nonumber \\ 
&\qquad\quad-\delta_{pr}\delta_{\beta \gamma} \tr(Z X^l \As^s E^{sq}_{\alpha \epsilon} \Cs^q X^k)
-\delta_{pr}\delta_{\beta \gamma} \sum_{\mu=1}^{\beta-1} \tr(\As^s E^{sq}_{\alpha \epsilon} \Cs^q X^k \As^p E^{pp}_{\mu \mu} \Cs^p X^l)\,.
 \end{align}
\end{subequations}
} 
In order for the elements on which we evaluate the Poisson bracket to be nonzero, we need $k\underset{m}{\equiv} 0$ for $f^k$, $l\underset{m}{\equiv} 0$ for $f^l$, while 
$l\underset{m}{\equiv} s-(r-1)$ for $g^l_{s \alpha,r \beta}$, and 
$k\underset{m}{\equiv} p-(q-1)$ for $g^k_{p\gamma,q \epsilon}$.
\end{lem} 
\begin{proof}
It suffices to establish these identities before reduction for the quasi-Poisson bracket on $\MM_{Q_{\dfat},\nfat}^{\circ}$ using the definition of the invariant regular functions given in \eqref{Eq:Cyfg}. Thus, we need the quasi-Poisson bracket between the evaluation functions given by the entries of the matrices  
\begin{equation}
 X_s\in \Mat(n\times n,\CC),\,\, \As^s\in \Mat(n\times d_s,\CC),\,\, \Cs^s\in \Mat(d_s\times n,\CC),\quad s\in \Z_m\,.
\end{equation}
Recall that the matrices $\As^s,\Cs^s$ are defined through \eqref{AsCsm}. 
The quasi-Poisson bracket $\br{(X_r)_{ij},(X_s)_{kl}}$ is given in \eqref{cyida}, while \eqref{cyidd} and \eqref{cyidw} can be rewritten as 
\begin{subequations}
 \begin{align}
\br{(X_s)_{ij}, \As^r_{k\beta}}\,=\,& \frac12 \delta_{(s,r-1)}\, \delta_{kj} (X_s \As^r)_{i\beta}
-\frac12 \delta_{rs}\, (X_s)_{kj} \As^r_{i\beta}\,,\\ 
\br{\As^s_{i\alpha},\As^r_{k\beta}}\,=\,&-\frac12\, o(s,r) \As^r_{k\beta} \As^s_{i\alpha} 
-\frac12 \,\delta_{sr}o(\alpha,\beta) \big(  \As^r_{k\beta} \As^s_{i\alpha} + \As^s_{k\alpha} \As^r_{i\beta} \big)\,.
 \end{align}
\end{subequations}
Using that the entries $\Cs^s_{\alpha j}$ have an inductive (in $\alpha=1,\ldots,d_s$) definition through \eqref{Eq:RecCs}, we can obtain by induction from \eqref{Eq:qPbrack1}--\eqref{Eq:qPbrack3} that 
\begin{subequations}
 \begin{align}
\br{(X_s)_{ij}, \Cs^r_{\beta l}}\,=\,& \frac12 \delta_{(s,r-1)}\, (\Cs^r X_s)_{\beta j} \delta_{il} 
+\frac12 \delta_{(s,r-2)}\, \Cs^r_{\beta j} (X_s)_{il}\,,\\ 
\br{\As^s_{i\alpha},\Cs^r_{\beta l}}\,=\,&\frac12\, o(s,r) \As^s_{i\alpha} \Cs^r_{\beta l}
-\frac12 \delta_{(s,r-1)} (\Cs^r \As^s)_{\beta\alpha} \delta_{il} \nonumber \\
&+\frac12 \,\delta_{sr}\big[o(\alpha,\beta)-\delta_{\alpha\beta}\big]   \As^s_{i\alpha}\Cs^r_{\beta l} 
-\delta_{sr}\delta_{\alpha\beta} (Z_{s-1})_{il} - \delta_{sr}\delta_{\alpha\beta}\sum_{\lambda=1}^{\beta-1} \As^s_{i\lambda} \Cs^r_{\lambda l}\,.
 \end{align}
\end{subequations}
The expression $\br{\Cs^s_{\alpha j},\Cs^r_{\beta l}}$ is the more cumbersome to derive. We first need to prove inductively using \eqref{Eq:RecCs} that 
\begin{subequations}
 \begin{align}
\br{(Z_s)_{ij}, \Cs^r_{\beta l}}\,=\,& -\frac12 \delta_{(s,r-2)}\, (\Cs^r Z_s)_{\beta j} \delta_{il} 
+\frac12 \delta_{(s,r-1)}\, \Cs^r_{\beta j} (Z_s)_{il}\,,\\ 
\br{(V_{s,\alpha})_{j},\Cs^r_{\beta l}}\,=\,&-\frac12\, o(s,r) (V_{s,\alpha}) \Cs^r_{\beta l}
+\frac12 \delta_{(s,r-1)} \Cs^r_{\beta j} (V_{s,\alpha})_l \nonumber \\
&-\frac12 \,\delta_{sr}\big[o(\alpha,\beta)+\delta_{\alpha\beta}\big]  (V_{s,\alpha})_j \Cs^r_{\beta l} \,.
 \end{align}
\end{subequations}
Then, gathering the different quasi-Poisson brackets of the form $\br{-,\Cs^r_{\beta l}}$ obtained so far, we can again use \eqref{Eq:RecCs} to prove inductively that 
\begin{equation}
 \br{\Cs^s_{\alpha j},\Cs^r_{\beta l}}\,=\,-\frac12\, o(s,r) \Cs^s_{\alpha j} \Cs^r_{\beta l}  
+\frac12 \,\delta_{sr}o(\alpha,\beta) \big(  \Cs^r_{\beta j} \Cs^s_{\alpha l} - \Cs^s_{\alpha j} \Cs^r_{\beta l} \big)\,.
\end{equation}

Once the quasi-Poisson brackets have been obtained, it is a routine but tedious computation to get the identities from the statement of the lemma  on $\MM_{Q_{\dfat},\nfat}^{\circ}$. 
The explicit computations (which use the formalism of double brackets, see Remark \ref{Rem:DBr}) are given in the proof of Lemma 3.2.7 in \cite{Fth}.
\end{proof}

We can compose the result of Lemma \ref{LemPoiCy} with the map $\xi$ as follows. First, we note  that 
\begin{equation} \label{brapCyff}
  \xi^\ast \br{f^k,f^l}
=\,0\,, \quad 
\xi^\ast \br{f^k,g^l_{s \alpha,r \beta}}
=\,k\, \sum_{i=1}^n \aaa_i^{s \alpha} \ccc_i^{r \beta}\lambda_i^{k+l}\,.
\end{equation}
(We follow the conventions of Lemma \ref{LemPoiCy} for all the indices under consideration in \eqref{brapCyff} and \eqref{brapCygg}.) 
Moreover, we also obtain 
{ \allowdisplaybreaks  
\begin{equation} \label{brapCygg}
\begin{aligned}
  \xi^\ast \br{g^k_{p\gamma,q \epsilon},g^l_{s\alpha,r \beta}}
=\,&\frac12 
\left(\sum_{v=1,\ldots,k}^\bullet -\sum_{v=1,\ldots,l}^\bullet \right) \sum_{i,j=1}^n
(\aaa_j^{s \alpha} \ccc_i^{r \beta} \lambda_i^v \aaa_i^{p \gamma}\ccc_j^{q \epsilon} \lambda_j^{k+l-v})  \\
&+\frac12 \left(\sum_{v=1,\ldots,k}^\triangle -\sum_{v=1,\ldots,l}^\triangle \right)  \sum_{i,j=1}^n
(\aaa_j^{s \alpha} \ccc_i^{r \beta} \lambda_i^{k+l-v} \aaa_i^{p \gamma}\ccc_j^{q \epsilon} \lambda_j^{v}) \\
&+\frac12 [o(p,r)-o(p,s)+o(q,s)-o(q,r)] \sum_{i,j=1}^n
(\aaa_i^{s \alpha} \ccc_j^{q \epsilon} \lambda_j^{k} \aaa_j^{p \gamma}\ccc_i^{r \beta} \lambda_i^{l})\\
&+\frac12 \delta_{ps} o(\alpha,\gamma) \sum_{i,j=1}^n
[(\aaa_i^{p \gamma} \ccc_j^{q \epsilon} \lambda_j^{k} \aaa_j^{s \alpha}\ccc_i^{r \beta} \lambda_i^{l})
+(\aaa_i^{s \alpha} \ccc_j^{q \epsilon} \lambda_j^{k} \aaa_j^{p \gamma}\ccc_i^{r \beta} \lambda_i^{l})] \\
&+\frac12 \delta_{qr}o(\epsilon,\beta) \sum_{i,j=1}^n
[(\aaa_i^{s \alpha} \ccc_j^{r \beta} \lambda_j^{k} \aaa_j^{p \gamma}\ccc_i^{q \epsilon} \lambda_i^{l})
-(\aaa_i^{s \alpha} \ccc_j^{q \epsilon} \lambda_j^{k} \aaa_j^{p \gamma}\ccc_i^{r \beta} \lambda_i^{l})]\\
&+\frac12 \Big( \delta_{qs}[o(\epsilon,\alpha)+\delta_{\alpha \epsilon}] - \delta_{pr} [o(\beta,\gamma)+\delta_{\beta \gamma}] \Big) \sum_{i,j=1}^n
(\aaa_i^{s \alpha} \ccc_j^{q \epsilon} \lambda_j^{k} \aaa_j^{p \gamma}\ccc_i^{r \beta} \lambda_i^{l})  \\
&+\delta_{qs}\delta_{\alpha \epsilon} \sum_{i,j=1}^n
\left( (Z_{s-1})_{ij} + \sum_{\lambda=1}^{\epsilon-1}\aaa_i^{s \lambda}\ccc_j^{s \lambda}\right)\lambda_j^k \aaa_j^{p \gamma} \ccc_i^{r \beta} \lambda_i^l   \\ 
&-\delta_{pr}\delta_{\beta \gamma} \sum_{i,j=1}^n
\left( (Z_{p-1})_{ji} +\sum_{\mu=1}^{\beta-1} \aaa_j^{p \mu}\ccc_i^{p \mu}  \right)
\lambda_i^l \aaa_i^{s\alpha}\ccc_j^{q \epsilon} \lambda_j^k\,,
\end{aligned}
\end{equation}
} 
where  in the sum $\sum \limits^\bullet$ we require $v\underset{m}{\equiv}p-(r-1)$, while for $\sum \limits^\triangle$ we require $v\underset{m}{\equiv}s-(q-1)$. To understand how we get the factor $(Z_{s-1})_{ij}$ when we write $\xi^\ast \tr(Z X^k \As^p E^{pr}_{\gamma\beta} \Cs^r X^{l})$, remark that 
 $\Cs^r \in\Hom(\VV_{r-1}, \VV_{\infty})$ so that the element $X^lZ$ in $\Cs^r X^lZ$ acts as $X_{r-1} \ldots X_{l+r-2}Z_{l+r-2}$. By assumption $(r-1)+l=s$ modulo $m$, so that the element $Z$ in this expression can be replaced by $Z_{s-1}$.  
The same observation explains how we get $(Z_{p-1})_{ji}$ in $\xi^\ast \tr(Z X^l \As^s E^{sq}_{\alpha \epsilon} \Cs^q X^k)$. 

\subsubsection{Derivation of \eqref{EqCyFirst}}  \label{sss:Loc-Lambda}

The first identity in \eqref{EqCyFirst} implies $\brloc{\xi^\ast f^k, \xi^\ast f^l}=0$, as desired from \eqref{brapCyff}. Similarly, we get 
\begin{equation*}
  \brloc{\xi^\ast f^k, \sum_{1\leq\alpha\leq d_0} \xi^\ast g^{l}_{0\alpha, r \beta}}
=m k\sum_{i,j=1}^n \lambda_i^{k-1}\lambda_j^l \brloc{\lambda_i,\ccc_j^{r\beta}}
=k\sum_{i=1}^n \lambda_i^{k+l}\ccc_i^{r\beta}\,.
\end{equation*}
This is nothing else than $k\sum_{\alpha=1}^{d_0} g^{k+l}_{0\alpha, r \beta}$, 
which is  $\xi^\ast\br{ f^k, \sum_{\alpha=1}^{d_0} g^{l}_{0\alpha, r \beta}}$ by Lemma \ref{LemPoiCy} as we noticed in \eqref{brapCyff}. The last identity is checked in the same way.

\subsubsection{Derivation of \eqref{Eq:CC00}--\eqref{Eq:CCq<r}} \label{sss:Loc-CC}

We check that 
\begin{equation} \label{brapCyCC}
  \brloc{\sum_{1\leq\gamma\leq d_0} \xi^\ast g^k_{0\gamma, q \epsilon}, \sum_{1\leq\alpha\leq d_0} \xi^\ast g^{l}_{0\alpha, r \beta}}
=\sum_{1\leq\alpha,\gamma\leq d_0} \xi^\ast\br{ g^k_{0\gamma, q \epsilon}, g^{l}_{0\alpha, r \beta}}
\end{equation}
for all possible indices. We can first rewrite the left-hand side using \eqref{EqCyFirst}  as 
\begin{equation} 
\eqref{brapCyCC}_{LHS}=
\frac{k-l}{m}\sum_{i=1}^n \lambda_i^{k+l}\ccc_i^{q \epsilon}\ccc_i^{r \beta} 
+ \sum_{i,j=1}^n \lambda_j^k \lambda_i^l 
\brloc{\ccc_j^{q \epsilon},\ccc_i^{r \beta}}\,,
\end{equation}
while we get for the right-hand side using \eqref{brapCygg} that 
{ \allowdisplaybreaks  
\begin{equation*} 
\begin{aligned}
 \eqref{brapCyCC}_{RHS}
=\,&\frac12 
\left(\sum_{v=1,\ldots,k}^\bullet -\sum_{v=1,\ldots,l}^\bullet \right) \sum_{i,j=1}^n
( \ccc_i^{r \beta} \lambda_i^v \ccc_j^{q \epsilon} \lambda_j^{k+l-v})  \\
&+\frac12 \left(\sum_{v=1,\ldots,k}^\triangle -\sum_{v=1,\ldots,l}^\triangle \right)  \sum_{i,j=1}^n
( \ccc_i^{r \beta} \lambda_i^{k+l-v} \ccc_j^{q \epsilon} \lambda_j^{v}) \\
&+\frac12 [o(0,r)+o(q,0)-o(q,r)] \sum_{i,j=1}^n
( \ccc_j^{q \epsilon} \lambda_j^{k} \ccc_i^{r \beta} \lambda_i^{l})\\
&+\frac12 \delta_{qr}o(\epsilon,\beta) \sum_{i,j=1}^n
[( \ccc_j^{r \beta} \lambda_j^{k} \ccc_i^{q \epsilon} \lambda_i^{l})
-( \ccc_j^{q \epsilon} \lambda_j^{k} \ccc_i^{r \beta} \lambda_i^{l})]\\
&+\frac12  \delta_{q0}\sum_{\alpha=1}^{d_0} [o(\epsilon,\alpha)+\delta_{\alpha \epsilon}]  \sum_{i,j=1}^n
(\aaa_i^{0 \alpha} \ccc_j^{q \epsilon} \lambda_j^{k} \ccc_i^{r \beta} \lambda_i^{l})  \\
& -\frac12  \delta_{0r} \sum_{\gamma=1}^{d_0} [o(\beta,\gamma)+\delta_{\beta \gamma}]  \sum_{i,j=1}^n
( \ccc_j^{q \epsilon} \lambda_j^{k} \aaa_j^{0 \gamma}\ccc_i^{r \beta} \lambda_i^{l})  \\
&+\delta_{q0} \sum_{i,j=1}^n
\left( (Z_{m-1})_{ij} + \sum_{\lambda=1}^{\epsilon-1}\aaa_i^{0 \lambda}\ccc_j^{0 \lambda}\right)\lambda_j^k  \ccc_i^{r \beta} \lambda_i^l   \\ 
&-\delta_{0r} \sum_{i,j=1}^n \left( (Z_{m-1})_{ji} +\sum_{\mu=1}^{\beta-1} \aaa_j^{0 \mu}\ccc_i^{0 \mu}  \right)
\lambda_i^l \ccc_j^{q \epsilon} \lambda_j^k\,,
\end{aligned}
\end{equation*}
} 
where  in the sum $\sum \limits^\bullet$ we require $v\underset{m}{\equiv}1-r$, while for $\sum \limits^\triangle$ we require $v\underset{m}{\equiv}1-q$.  
Here, we used that the fourth line of \eqref{brapCygg} vanishes when we sum over all $\alpha,\gamma=1,\ldots,d_0$. 
We clearly see that we have to discuss the possible choices of $r,q=0,\ldots,m-1$ separately. 

\medskip

\textbf{Checking \eqref{Eq:CC00}.} We have to see that for $q=r=0$ and for any $1\leq \beta, \epsilon\leq d_0$, \eqref{brapCyCC} holds. It suffices to reproduce the derivation of \cite[(A.9)]{CF2}, with $x_j:=\lambda_j^m$ and $b_j^\alpha:=\ccc_j^{0\alpha}$. 

\textbf{Checking \eqref{Eq:CCq0}.} We have to see that for $r=0$ and for  any $q\in \Z_m\setminus\{0\}$, $1\leq \epsilon\leq d_q$,  $1\leq \beta\leq d_0$,  \eqref{brapCyCC} holds. We do the full derivation of this particular case, so that the reader can see the tricks that are needed to check all the other cases whose proofs are omitted. 

Note that for $g^k_{0\gamma, q \epsilon}$ and $g^{l}_{0\alpha, 0 \beta}$ to be nonzero, we need $k=k_0m+1-q$ and $l=l_0m+1$ for some $k_0,l_0\in \N^\times$. In particular, remark that $k<k_0m+1$ and we can write 
\begin{equation*} 
\begin{aligned}
 \eqref{brapCyCC}_{RHS}
=\,&\frac12 
\left(\sum_{v_0=0}^{k_0-1} -\sum_{v_0=0}^{l_0} \right) \sum_{i,j=1}^n
 \ccc_i^{0 \beta} \lambda_i^{v_0m+1} \ccc_j^{q \epsilon} \lambda_j^{(k_0+l_0-v_0)m+1-q}\\
&+\frac12 \left( \sum_{v_0=1}^{k_0} -\sum_{v_0=1}^{l_0} \right)  \sum_{i,j=1}^n
\ccc_i^{0 \beta} \lambda_i^{(k_0+l_0-v_0)m+1} \ccc_j^{q \epsilon} \lambda_j^{v_0m+1-q} \\
& -\frac12   \sum_{\gamma=1}^{d_0} [o(\beta,\gamma)+\delta_{\beta \gamma}]  \sum_{i,j=1}^n
 \ccc_j^{q \epsilon} \lambda_j^{k} \aaa_j^{0 \gamma}\ccc_i^{0 \beta} \lambda_i^{l}  \\
&-\sum_{i,j=1}^n \left( (Z_{m-1})_{ji} +\sum_{\mu=1}^{\beta-1} \aaa_j^{0 \mu}\ccc_i^{0 \mu}  \right)
\lambda_i^l \ccc_j^{q \epsilon} \lambda_j^k\,.
\end{aligned}
\end{equation*}
Indeed, in the first sum, we need $v\underset{m}{\equiv}1$ so we sum over $v=v_0m+1$ but we can not consider $v_0=k_0$, while for the second $v\underset{m}{\equiv}1-q$ and $v=v_0m+1-q<v_0m+1$.  Using that 
$ [o(\beta,\gamma)+\delta_{\beta \gamma}]=1-2\delta_{(\beta>\gamma)}$, we can write 
\begin{equation*} 
\begin{aligned}
 \eqref{brapCyCC}_{RHS}
=\,&\frac12  \sum_{i,j=1}^n\ccc_i^{0 \beta}\ccc_j^{q \epsilon} \lambda_i \lambda_j^{1-q}  \,\Sigma^{(k,l)}_{(i,j)} 
-\frac12  \sum_{i,j=1}^n \ccc_i^{0 \beta}\ccc_j^{q \epsilon} \lambda_i^{k_0m+1}  \lambda_j^{l_0m+1-q}
\\
& -\frac12    \sum_{i,j=1}^n \lambda_j^{k}  \lambda_i^{l}  \ccc_j^{q \epsilon}\ccc_i^{0 \beta} 
+      \sum_{i,j=1}^n \lambda_j^{k} \lambda_i^{l}
\ccc_j^{q \epsilon} \sum_{\gamma=1}^{\beta-1} \aaa_j^{0 \gamma}\ccc_i^{0 \beta}  \\
&-\sum_{i,j=1}^n  \lambda_j^k\lambda_i^l \ccc_j^{q \epsilon}
\left( (Z_{m-1})_{ji} +\sum_{\mu=1}^{\beta-1} \aaa_j^{0 \mu}\ccc_i^{0 \mu}  \right) \,,
\end{aligned}
\end{equation*}
where, for $1\leq i,j,k,l\leq n$, we introduced 
\begin{equation} \label{Sumr}
\Sigma^{(k,l)}_{(i,j)} :=
  \left(\sum_{v_0=1}^{k_0} -\sum_{v_0=1}^{l_0} \right)
\left(\lambda_i^{v_0m}  \lambda_j^{(k_0+l_0-v_0)m}+ \lambda_i^{(k_0+l_0-v_0)m}  \lambda_j^{v_0m}  \right)\,.
\end{equation}
To reduce $\Sigma^{(k,l)}_{(i,j)}$, we have the following result which is easily checked. 
\begin{lem}\label{L:Cyrkl}
  If $i=j$, $\Sigma^{(k,l)}_{(i,j)} =(k_0-l_0)\lambda_i^{(k_0+l_0)m}$, while if $i \neq j$, 
\begin{equation}
  \Sigma^{(k,l)}_{(i,j)}=\frac{\lambda_i^m+\lambda_j^m}{\lambda_i^m-\lambda_j^m}
\left(\lambda_i^{k_0m}  \lambda_j^{l_0m}- \lambda_i^{l_0m}  \lambda_j^{k_0m}  \right)\,.
\end{equation}
\end{lem}
Using this lemma, we can write\footnote{For the remainder of this section, we write $\sum\limits_{i\neq j}$ when we sum over all $1\leq i,j\leq n$ with $i\neq j$.} 
\begin{equation*} 
\begin{aligned}
 \eqref{brapCyCC}_{RHS}
=\,&
\frac12  \sum_{i\neq j}\ccc_i^{0 \beta}\ccc_j^{q \epsilon} \lambda_i \lambda_j^{1-q}  
\frac{\lambda_i^m+\lambda_j^m}{\lambda_i^m-\lambda_j^m}
\left(\lambda_i^{k_0m}  \lambda_j^{l_0m}- \lambda_i^{l_0m}  \lambda_j^{k_0m}  \right) \\
&+(k_0-l_0-1) \sum_{i=1}^n \lambda_i^{k+l}\ccc_i^{0 \beta}\ccc_i^{q \epsilon}
-\frac12  \sum_{i\neq j} \ccc_i^{0 \beta}\ccc_j^{q \epsilon} 
(\lambda_i^{k_0m+1}  \lambda_j^{l_0m+1-q}+ \lambda_j^{k}  \lambda_i^{l})   \\
&-\sum_{i,j=1}^n  \lambda_j^k\lambda_i^l \ccc_j^{q \epsilon} (Z_{m-1})_{ji}
-\sum_{i,j=1}^n  \lambda_j^k\lambda_i^l \ccc_j^{q \epsilon} 
\sum_{\mu=1}^{\beta-1} \aaa_j^{0 \mu}(\ccc_i^{0 \mu}-\ccc_i^{0 \beta} )   \,.
\end{aligned}
\end{equation*}
Recalling that $k=k_0m+1-q$ and $l=l_0m+1$, we find after relabelling some of the indices $i,j$ in the first and third sums that 
\begin{equation*} 
\begin{aligned}
 \eqref{brapCyCC}_{RHS}
=\,&\frac12  \sum_{i\neq j} \lambda_j^k \lambda_i^{l} \frac{\lambda_j^m+\lambda_i^m}{\lambda_j^m-\lambda_i^m}
\left(\ccc_j^{q \epsilon}\ccc_i^{0 \beta} + \frac{\lambda_j^q}{\lambda_i^q} \ccc_i^{q\epsilon} \ccc_j^{0\beta} \right) \\
&+\left(\frac{k-l}{m}+\frac{q-m}{m}\right)\sum_{i=1}^n \lambda_i^{k+l}\ccc_i^{0 \beta}\ccc_i^{q \epsilon} 
-\frac12  \sum_{i\neq j} \lambda_j^k \lambda_i^l \left(\ccc_j^{q \epsilon}\ccc_i^{0 \beta} + \frac{\lambda_j^q}{\lambda_i^q} \ccc_i^{q\epsilon} \ccc_j^{0\beta} \right)
 \\
&-\sum_{i,j=1}^n  \lambda_j^k\lambda_i^l \ccc_j^{q \epsilon} (Z_{m-1})_{ji}
-\sum_{i,j=1}^n  \lambda_j^k\lambda_i^l \ccc_j^{q \epsilon} 
\sum_{\mu=1}^{\beta-1} \aaa_j^{0 \mu}(\ccc_i^{0 \mu}-\ccc_i^{0 \beta} )   \,.
\end{aligned}
\end{equation*}
If we sum together the first and third terms of $\eqref{brapCyCC}_{RHS}$ just obtained, it is not hard to see that it matches  $\eqref{brapCyCC}_{LHS}$ after introducing \eqref{Eq:CCq0} in it.

\textbf{Checking \eqref{Eq:CC0r}.} This is equivalent to \eqref{Eq:CCq0} by antisymmetry. 

\textbf{Checking \eqref{Eq:CCqq}.} It suffices to check that for any $q\in \Z_m\setminus\{0\}$ with $r=q$ and $1\leq \beta,\epsilon\leq d_q$, \eqref{brapCyCC} holds. 

\textbf{Checking \eqref{Eq:CCq>r}.} It suffices to check that for any $0<r<q\leq m-1$, $1\leq \beta\leq d_r$ and $1\leq \epsilon\leq d_q$, \eqref{brapCyCC} holds.

\textbf{Checking \eqref{Eq:CCq<r}.} This is equivalent to \eqref{Eq:CCq>r} by antisymmetry.

\subsubsection{Derivation of \eqref{Eq:CA00}--\eqref{Eq:CAqNs}} \label{sss:Loc-CA}

We need to check that 
\begin{equation} \label{brapCyCA}
  \brloc{\sum_{1\leq \gamma\leq d_0} \xi^\ast g^k_{0\gamma, q \epsilon}, \xi^\ast g^{l}_{s\alpha, 01}}
=\sum_{1\leq \gamma\leq d_0} \xi^\ast\br{ g^k_{0\gamma, q \epsilon}, g^{l}_{s\alpha, 01}}
\end{equation}
for all possible pairs of indices $(q,\epsilon)$ and $(s,\alpha)$. 
First, we have by \eqref{EqCyFirst} that the left-hand side is given by 
\begin{equation*} 
\eqref{brapCyCA}_{LHS}=
\frac{k-l}{m}\sum_{i=1}^n  \lambda_i^{k+l}\ccc_i^{q \epsilon}\ccc_i^{01}\aaa_i^{s\alpha} 
+ \sum_{i,j=1}^n  \lambda_j^k \lambda_i^l \aaa_i^{s\alpha}  \brloc{\ccc_j^{q \epsilon},\ccc_i^{01}} +
\sum_{i,j=1}^n   \lambda_j^k \lambda_i^l \ccc_i^{01} \brloc{\ccc_j^{q \epsilon},\aaa_i^{s\alpha}}\,.
\end{equation*}
For the middle term, we can use the established expression for $\br{\ccc_j^{q \epsilon},\ccc_i^{01}}$ given by  \eqref{Eq:CC00}  if $q=0$, or \eqref{Eq:CCq0} otherwise.  Hence, one needs to consider these two cases separately. Next, note that \eqref{brapCygg} gives 
{ \allowdisplaybreaks  
\begin{equation*} 
\begin{aligned}
\eqref{brapCyCA}_{RHS} 
=\,&\frac12 
\left(\sum_{v=1,\ldots,k}^\bullet -\sum_{v=1,\ldots,l}^\bullet \right) \sum_{i,j=1}^n 
\,\aaa_j^{s \alpha} \ccc_i^{01} \lambda_i^v \ccc_j^{q \epsilon} \lambda_j^{k+l-v}  \\
&+\frac12 \left(\sum_{v=1,\ldots,k}^\triangle -\sum_{v=1,\ldots,l}^\triangle \right)  \sum_{i,j=1}^n 
\,\aaa_j^{s \alpha} \ccc_i^{01} \lambda_i^{k+l-v} \ccc_j^{q \epsilon} \lambda_j^{v} \\
&+\frac12 [-o(0,s)+o(q,s)-o(q,0)] \sum_{i,j=1}^n  \lambda_j^{k}\lambda_i^{l}
\aaa_i^{s \alpha} \ccc_j^{q \epsilon} \ccc_i^{01} \\
&+\frac12 \delta_{0s} \sum_{\gamma=1}^{d_0} o(\alpha,\gamma) \sum_{i,j=1}^n  \lambda_j^{k}\lambda_i^{l} \ccc_j^{q \epsilon} \ccc_i^{01} [\aaa_i^{0 \gamma}   \aaa_j^{s \alpha} +\aaa_i^{s \alpha}  \aaa_j^{0 \gamma} ] \\
&-\frac12 \delta_{q0}\delta_{(\epsilon\neq 1)} \sum_{i,j=1}^n  \lambda_j^{k}\lambda_i^{l} \aaa_i^{s \alpha} 
[\ccc_j^{01} \ccc_i^{q \epsilon} -\ccc_j^{q \epsilon} \ccc_i^{01} ]\\
&+\frac12  \delta_{qs}[o(\epsilon,\alpha)+\delta_{\alpha \epsilon}] 
\sum_{i,j=1}^n   \lambda_j^{k}\lambda_i^{l}\ccc_j^{q \epsilon} \aaa_i^{s \alpha}  \ccc_i^{01} 
-\frac12 \sum_{i,j=1}^n   \lambda_j^{k}\lambda_i^{l}\ccc_j^{q \epsilon} \aaa_i^{s \alpha}  \ccc_i^{01}  \\
&+\delta_{qs}\delta_{\alpha \epsilon} \sum_{i,j=1}^n  \lambda_j^k \lambda_i^l \ccc_i^{01}
\left( (Z_{s-1})_{ij} + \sum_{\lambda=1}^{\epsilon-1}\aaa_i^{s \lambda}\ccc_j^{s \lambda}\right)    \\ 
&- \sum_{i,j=1}^n  \lambda_j^k\lambda_i^l (Z_{m-1})_{ji}  \aaa_i^{s\alpha}\ccc_j^{q \epsilon} \,,
\end{aligned}
\end{equation*}
}  
where  in the sum $\sum \limits^\bullet$ we require $v\underset{m}{\equiv}+1$, while for $\sum \limits^\triangle$ we require $v\underset{m}{\equiv}s-q+1$.

\textbf{Checking \eqref{Eq:CA00}.}  We have to see that for $q=r=0$ and for any $1\leq \beta, \epsilon\leq d_0$, \eqref{brapCyCA} holds. 
It is useful to note that for $q=0$, we can use \eqref{Eq:CC00} to write $\eqref{brapCyCA}_{LHS}$  as 
\begin{equation} 
 \begin{aligned}   
\eqref{brapCyCA}_{LHS}\stackrel{(q=0)}{=}&
\frac{k-l}{m}\sum_{i=1}^n  \lambda_i^{k+l}\ccc_i^{0 \epsilon}\ccc_i^{01}\aaa_i^{s\alpha} +
\sum_{i,j=1}^n   \lambda_j^k \lambda_i^l \ccc_i^{01} \brloc{\ccc_j^{0 \epsilon},\aaa_i^{s\alpha}} \\
&+ \frac12 \sum_{i\neq j} \lambda_j^k \lambda_i^l \frac{\lambda_j^m+\lambda_i^m}{\lambda_j^m-\lambda_i^m}  \aaa_i^{s\alpha}  (\ccc_j^{0\epsilon}\ccc_i^{01} + \ccc_i^{0\epsilon} \ccc_j^{01})     \\
&-\frac12 \delta_{(\epsilon \neq 1)} \sum_{i,j=1}^n  \lambda_j^k \lambda_i^l \aaa_i^{s\alpha}  
(\ccc_i^{0\epsilon} \ccc_j^{01} -\ccc_j^{0\epsilon}\ccc_i^{01}) \\
&+ \sum_{i,j=1}^n  \lambda_j^k \lambda_i^l \aaa_i^{s\alpha}  
\left( (Z_{m-1})_{ij} \ccc_i^{0 1} - (Z_{m-1})_{ji} \ccc_j^{0 \epsilon}  \right) \\
&+ \sum_{i,j=1}^n  \lambda_j^k \lambda_i^l \aaa_i^{s\alpha}  \ccc_i^{01} \sum_{\lambda=1}^{\epsilon-1} \aaa_i^{0\lambda} (\ccc_j^{0\lambda} - \ccc_j^{0\epsilon}) \,.  \label{LHSq0}
 \end{aligned}
\end{equation}
It then suffices to reproduce the derivation of \cite[(A.12)]{CF2}, with $x_j:=\lambda_j^m$, $a_j^\alpha:=\aaa_j^{0\alpha}$, and $b_j^\alpha:=\ccc_j^{0\alpha}$. 

\textbf{Checking \eqref{Eq:CA0s}.} It suffices to check that for $q=0$, $s\in \Z_m \setminus\{0\}$, $1\leq \alpha\leq d_s$ and $1\leq \epsilon\leq d_0$,  \eqref{brapCyCA} holds. We can use \eqref{LHSq0} to evaluate the left-hand side.  

\textbf{Checking \eqref{Eq:CAqq}.} 
It suffices to check that for any $q\in \Z_m\setminus \{0\}$ with $s=q$ and $1\leq \epsilon, \alpha\leq d_q$,  \eqref{brapCyCA} holds. 
It is useful to note that for $q\neq 0$, we can use \eqref{Eq:CCq0} to rewrite $\eqref{brapCyCA}_{LHS}$ as 
\begin{equation}  
\begin{aligned}
   \eqref{brapCyCA}_{LHS}\stackrel{(q\neq0)}{=}&
\frac{k-l+q-m}{m}\sum_{i=1}^n \lambda_i^{k+l}\ccc_i^{q \epsilon}\ccc_i^{01}\aaa_i^{s\alpha} 
+\sum_{i,j=1}^n   \lambda_j^k \lambda_i^l \ccc_i^{01} \brloc{\ccc_j^{q \epsilon},\aaa_i^{s\alpha}} \\
&+ \sum_{i\neq j} \lambda_j^k \lambda_i^l \frac{\lambda_i^m}{\lambda_j^m-\lambda_i^m}  \aaa_i^{s\alpha}
\left(\ccc_j^{q\epsilon}\ccc_i^{01} +\frac{\lambda_j^q}{\lambda_i^q} \ccc_i^{q\epsilon} \ccc_j^{01}\right) 
- \sum_{i,j=1}^n  \lambda_j^k \lambda_i^l \aaa_i^{s\alpha} (Z_{m-1})_{ji} \ccc_j^{q \epsilon}\,.   \label{LHSqq}
\end{aligned}
\end{equation}

\textbf{Checking \eqref{Eq:CAq0}.} It suffices to check that for $s=0$, $q\in \Z_m\setminus \{0\}$, $1\leq \epsilon\leq d_q$ and $1\leq \alpha\leq d_0$, \eqref{brapCyCA} holds. We can use \eqref{LHSqq} to evaluate the left-hand side. 

\textbf{Checking \eqref{Eq:CAqNs}.} It suffices to check that for any $q,s\in \Z_m \setminus \{0\}$ with $q \neq s$, and for any $1\leq \epsilon\leq d_q$,  $1\leq \alpha\leq d_s$, \eqref{brapCyCA} holds. 
 We can use \eqref{LHSqq} to evaluate the left-hand side. 

\subsubsection{Derivation of \eqref{Eq:AA00}--\eqref{Eq:AAp>s}} \label{sss:Loc-AA}

We need to check that \begin{equation} \label{brapCyAA}
  \brloc{\xi^\ast g^k_{p\gamma, 01}, \xi^\ast g^{l}_{s\alpha, 01}}
= \xi^\ast\br{ g^k_{p\gamma, 01}, g^{l}_{s\alpha, 01}}
\end{equation}
for all possible pairs of indices $(p,\gamma)$ and $(s,\alpha)$. Using \eqref{EqCyFirst} and \eqref{Eq:CC00}, the left-hand side is given by 
\begin{equation} 
\begin{aligned}
\eqref{brapCyAA}_{LHS}=&
\sum_{i,j=1}^n   \lambda_j^k \lambda_i^l 
\left(\aaa_j^{p\gamma} \ccc_i^{01} \brloc{\ccc_j^{01},\aaa_i^{s\alpha}} 
+  \ccc_j^{01} \aaa_i^{s\alpha}  \brloc{\aaa_j^{p\gamma},\ccc_i^{01}} 
+  \ccc_j^{01} \ccc_i^{01}  \brloc{\aaa_j^{p\gamma},\aaa_i^{s\alpha}}  \right)  \\
&+\frac{k-l}{m} \sum_{i=1}^n  \lambda_i^{k+l}(\ccc_i^{01})^2\aaa_i^{p\gamma}\aaa_i^{s\alpha} 
+  \sum_{i\neq j} \lambda_j^k \lambda_i^l \frac{\lambda_j^m+\lambda_i^m}{\lambda_j^m-\lambda_i^m} 
\aaa_j^{p\gamma} \aaa_i^{s\alpha} \ccc_j^{01}\ccc_i^{01} \\
& + \sum_{i,j=1}^n  \lambda_j^k \lambda_i^l (Z_{m-1})_{ij}\aaa_j^{p\gamma} \aaa_i^{s\alpha} \ccc_i^{01}
- \sum_{i,j=1}^n  \lambda_j^k \lambda_i^l (Z_{m-1})_{ji}\aaa_j^{p\gamma} \aaa_i^{s\alpha} \ccc_j^{01} \,,
\end{aligned}
\end{equation}
with the brackets in the first line depending on the combination of indices.  
For the right-hand side,  \eqref{brapCygg} gives 
{ \allowdisplaybreaks  
\begin{equation*} 
\begin{aligned}    
\eqref{brapCyAA}_{RHS}=\,&
\frac12 \left(\sum_{v=1,\ldots,k}^\bullet -\sum_{v=1,\ldots,l}^\bullet \right) \sum_{i,j=1}^n 
\aaa_j^{s \alpha} \ccc_i^{01} \aaa_i^{p \gamma}\ccc_j^{01} \, \lambda_i^v \lambda_j^{k+l-v}  \\
&+\frac12 \left(\sum_{v=1,\ldots,k}^\triangle -\sum_{v=1,\ldots,l}^\triangle \right)  \sum_{i,j=1}^n 
\aaa_j^{s \alpha} \ccc_i^{01}  \aaa_i^{p \gamma}\ccc_j^{01} \, \lambda_i^{k+l-v}\lambda_j^{v} \\
&+\frac12 [o(p,0)-o(p,s)+o(0,s)]\sum_{i,j=1}^n  \lambda_j^{k}\lambda_i^{l} \aaa_j^{p \gamma}\ccc_j^{01} \aaa_i^{s \alpha} \ccc_i^{01} \\
&+\frac12 \delta_{ps} o(\alpha,\gamma) \sum_{i,j=1}^n  \lambda_j^{k} \lambda_i^{l} \ccc_j^{01} \ccc_i^{01} 
(\aaa_i^{p \gamma}   \aaa_j^{s \alpha} +\aaa_i^{s \alpha}  \aaa_j^{p \gamma}) \\
&+\frac12 \Big( \delta_{0s} - \delta_{p0}  \Big) 
\sum_{i,j=1}^n  \lambda_j^{k}\lambda_i^{l}  \ccc_j^{01}  \aaa_j^{p \gamma}\ccc_i^{01}\aaa_i^{s \alpha}    \\
&+\delta_{0s}\delta_{\alpha 1} \sum_{i,j=1}^n  \lambda_j^k\lambda_i^l  (Z_{m-1})_{ij}  \aaa_j^{p \gamma} \ccc_i^{01}   
-\delta_{p0}\delta_{1 \gamma} \sum_{i,j=1}^n  \lambda_j^k\lambda_i^l  (Z_{m-1})_{ji}  \aaa_i^{s\alpha}\ccc_j^{01} \,,
\end{aligned}
\end{equation*}
}  
where  in the sum $\sum \limits^\bullet$ we require $v\underset{m}{\equiv}p+1$, while for $\sum \limits^\triangle$ we require $v\underset{m}{\equiv}s+1$. 
We have to analyse the different choices of $(p,\gamma)$ and $(s,\alpha)$.

\textbf{Checking \eqref{Eq:AA00}.} 
We have to see that for $p=s=0$ and for any $1\leq \alpha, \gamma\leq d_0$, \eqref{brapCyAA} holds. 
We note that for $q=0$, we can use \eqref{Eq:CC00}, \eqref{Eq:CA00} and \eqref{Eq:AA00} to simplify $\eqref{brapCyAA}_{LHS}$. 
It then suffices to reproduce the derivation of \cite[(A.14)]{CF2}, with $x_j:=\lambda_j^m$, $a_j^\alpha:=\aaa_j^{0\alpha}$, and $b_j^\alpha:=\ccc_j^{0\alpha}$. 

\textbf{Checking \eqref{Eq:AApp}.} It suffices to check that for any $p\in \Z_m\setminus \{0\}$ with $s=p$ and $1\leq \alpha, \gamma\leq d_p$, we have that \eqref{brapCyAA} holds. 
It is useful to note that in the case $p=s$ nonzero, we can use \eqref{Eq:CA0s} to obtain after simplifications 
\begin{equation} 
\begin{aligned}
\eqref{brapCyAA}_{LHS}\stackrel{(p=s)}{=}&
\frac{k-l}{m} \sum_{i=1}^n \lambda_i^{k+l}(\ccc_i^{01})^2\aaa_i^{p\gamma}\aaa_i^{p\alpha} 
+ \sum_{i,j=1}^n \lambda_j^k \lambda_i^l \ccc_j^{01} \ccc_i^{01}  \brloc{\aaa_j^{p\gamma},\aaa_i^{p\alpha}}  \\
&+ \sum_{i\neq j} \lambda_j^k \lambda_i^l \ccc_j^{01} \ccc_i^{01} \aaa_j^{p\gamma}
\frac{\lambda_i^m}{\lambda_j^m-\lambda_i^m}\frac{\lambda_j^p}{\lambda_i^p}  \aaa_j^{p\alpha}  
+ \sum_{i\neq j} \lambda_j^k \lambda_i^l \ccc_j^{01} \ccc_i^{01} \aaa_i^{p\alpha}
 \frac{\lambda_j^m}{\lambda_j^m-\lambda_i^m}\frac{\lambda_i^p}{\lambda_j^p}  \aaa_i^{p\gamma}  \,. \label{LHSAApp}
\end{aligned}
\end{equation}

\textbf{Checking \eqref{Eq:AA0s}.} It suffices to check that for $p=0$ and for any $s\in \Z_m\setminus \{0\}$, $1\leq \gamma\leq d_0$ and $1\leq \alpha\leq d_s$, \eqref{brapCyAA} holds. It is useful to note that from \eqref{Eq:CA00} and \eqref{Eq:CA0s}, we can obtain 
{ \allowdisplaybreaks  
\begin{align*}
\eqref{brapCyAA}_{LHS}\stackrel{(0=p\neq s)}{=}&
\frac{k-l+s-m}{m} \sum_{i=1}^n \lambda_i^{k+l}(\ccc_i^{01})^2\aaa_i^{0\gamma}\aaa_i^{s\alpha} 
+ \sum_{i,j=1}^n \lambda_j^k \lambda_i^l \ccc_j^{01} \ccc_i^{01}  \brloc{\aaa_j^{0\gamma},\aaa_i^{s\alpha}}  \\
&+ \sum_{i\neq j} \lambda_j^k \lambda_i^l \ccc_j^{01} \ccc_i^{01} \aaa_j^{s \alpha}\aaa_j^{o\gamma}
\frac{\lambda_i^m}{\lambda_j^m-\lambda_i^m}\frac{\lambda_j^s}{\lambda_i^s}    
+\frac12  \sum_{i\neq j} \lambda_j^k \lambda_i^l \ccc_j^{01} \ccc_i^{01} \aaa_i^{s\alpha}\aaa_i^{0\gamma}
 \frac{\lambda_j^m+\lambda_i^m}{\lambda_j^m-\lambda_i^m} \\
&-\frac12 \sum_{i\neq j} \lambda_j^k \lambda_i^l \ccc_j^{01}\aaa_j^{o\gamma} \ccc_i^{01} \aaa_i^{s \alpha}
-\delta_{\gamma1} \sum_{i,j=1}^n \lambda_j^k \lambda_i^l \ccc_j^{01}  \aaa_i^{s\alpha} (Z_{m-1})_{ji} \\
&-\frac12 \sum_{i,j=1}^n \lambda_j^k \lambda_i^l \ccc_j^{01} \ccc_i^{01} \aaa_i^{s\alpha} 
\sum_{\sigma=1}^{d_0} o(\gamma,\sigma)  (\aaa_j^{0\gamma} \aaa_i^{0\sigma} +\aaa_i^{0\gamma}\aaa_j^{0\sigma}) \,. 
\end{align*}
}  

\textbf{Checking \eqref{Eq:AAp0}.} This is equivalent to \eqref{Eq:AA0s} by antisymmetry. 

\textbf{Checking \eqref{Eq:AAp<s}.} It suffices to check that for any $p,s\in \Z_m \setminus \{0\}$ with $p < s$, and for any $1\leq \gamma\leq d_p$,  $1\leq \alpha\leq d_s$, \eqref{brapCyAA} holds.
It is useful to note that from \eqref{Eq:CA0s}, we can obtain 
{ \allowdisplaybreaks  
\begin{align*}
\eqref{brapCyAA}_{LHS}\stackrel{(0<p< s)}{=}&
\frac{k-p-l+s}{m}\sum_{i=1}^n \lambda_i^{k+l}(\ccc_i^{01})^2\aaa_i^{p\gamma}\aaa_i^{s\alpha} 
+ \sum_{i,j=1}^n \lambda_j^k \lambda_i^l \ccc_j^{01} \ccc_i^{01}  \brloc{\aaa_j^{p\gamma},\aaa_i^{s\alpha}}  \\
&+ \sum_{i\neq j} \lambda_j^k \lambda_i^l \ccc_j^{01} \ccc_i^{01} \aaa_j^{p\gamma}  \aaa_j^{s\alpha} 
\frac{\lambda_i^m}{\lambda_j^m-\lambda_i^m}\frac{\lambda_j^s}{\lambda_i^s} \\
&+ \sum_{i\neq j} \lambda_j^k \lambda_i^l \ccc_j^{01} \ccc_i^{01} \aaa_i^{p\gamma} \aaa_i^{s\alpha}
 \frac{\lambda_j^m}{\lambda_j^m-\lambda_i^m}\frac{\lambda_i^p}{\lambda_j^p}    \,. 
\end{align*} 
}  

\textbf{Checking \eqref{Eq:AAp>s}.} This is equivalent to \eqref{Eq:AAp<s} by antisymmetry.


\subsection{Local expression of distinguished functions} \label{ss:Loc-express}

In Section \ref{S:Subalg}, we integrated explicitly the vector fields associated with several distinguished regular functions \emph{before} reduction.  We will turn to discussing their integrability \emph{after} reduction in Section \ref{S:Int}. Using the coordinates from the space $\hloc$ given in \ref{ss:Loc-coord}, we now write some of these functions locally. Let us note that such expressions can be found in \cite{CF1} for $\dfat=(1,0,\ldots,0)$ and in \cite{CF2,F1} for $\dfat=(d,0,\ldots,0)$ with $d\geq 2$.

\subsubsection{Spin RS system with different types of spin variables} \label{sss:SpinRS}

After reduction from $\MM_{Q_{\dfat},\nfat}^{\bullet}$ to $\Cnm$, the invariant functions forming $\HH_+$ \eqref{Eq:Palg-H+} are pairwise Poisson commuting by Proposition \ref{Pr:DgIScy-1XY}. 
For $r\in\Z_m$, put $T_r:= (\Id_{\VV_r}+X_rY_r)$. Hence $\HH_+$ is generated by the regular functions $\tr T_r^k$, $k\in \N$, $r\in \Z_m$.  
If we restrict these functions further to $\Cnm'$, we note that we can write $T_r=X_rZ_r$. Then,  Proposition \ref{Pr:diffeo} allows us to write that for $r\in\Z_m$, $\hat{T}_r:=\xi^\ast T_r$ is given in terms of the coordinates of $\hloc$ by  
\begin{equation} \label{hatTij}
  (\hat{T}_r)_{ij}= \sum_{s=0}^r \frac{t_r}{t_{s-1}}\frac{\lambda_i^{m+(s-r)}\lambda_j^{-(s-r-1)}}{\lambda_i^m - t \lambda_j^m} \gfat_{ij}^s 
+\sum_{s=r+1}^{m-1}\frac{tt_r}{t_{s-1}}\frac{\lambda_i^{s-r}\lambda_j^{m-(s-r-1)}}{\lambda_i^m - t \lambda_j^m} \gfat_{ij}^s \,.
\end{equation}
In particular, we can write down locally the Hamiltonians  $F^{m,\dfat}_{r,k}:= \tr T_r^k$, $k\in \N$, from \eqref{hatTij}. E.g.   
\begin{equation} \label{Eq:LocCyF}
F^{m,\dfat}_{r,1}=\frac{t_r}{1-t}  \sum_{s=0}^r t_{s-1}^{-1} \sum_{i=1}^n \ffat_{ii}^s
+\frac{t_r}{1-t} \sum_{s=r+1}^{m-1}t t_{s-1}^{-1} \sum_{i=1}^n \ffat_{ii}^s\,, \quad \text{ for } \ffat_{ii}^s:=\gfat_{ii}^s \lambda_i\,.
\end{equation}
Here, we introduced the elements $\ffat_{ii}^s$ with $s\in \Z_m$ and $1\leq i \leq n$ because they are $W$-invariant, see the discussion before Proposition \ref{Pr:diffeo}. 
Let us note that $t_r F^{m,\dfat}_{r+1,1}-t_{r+1} F^{m,\dfat}_{r,1}=tt_{r+1} \sum_i \ffat_{ii}^{r+1}$ for $r=0,\ldots,m-2$. Hence, the following regular functions belong to $\HH_+$:
\begin{equation*}
 F_0= \sum_{i=1}^n \ffat_{ii}^0\,,\,\, F_1= \sum_{i=1}^n \ffat_{ii}^1\,, \, \ldots \,,\,\,
 F_{m-1}= \sum_{i=1}^n \ffat_{ii}^{m-1}\,.
\end{equation*}
Using Proposition \ref{Pr:CyPoi}, the corresponding Hamiltonian vector fields can be explicitly written in terms of the local coordinates $(\lambda_i, \aaa_i^{s,\alpha},\ccc_i^{s,\alpha})$. We will only deal with the case of the vector field $\frac{d}{dt}:=\br{-,F_0}_{\loc}$, and the reader can find the remaining expressions in \cite[\S5.3.1]{Fth}. For any distinct $1\leq i,j\leq n$, we set 
\begin{equation} \label{EqPotVr}
  \tilde{V}^0_{ij}=\frac{\lambda_i^m+\lambda_j^m}{\lambda_i^m-\lambda_j^m}\gfat_{ij}^0 -2 (Z_{m-1})_{ij}-\gfat_{ij}^0\,.
\end{equation}

\begin{lem} \label{BrF0coord}
  For any $1\leq \gamma,\epsilon \leq d_0$ and $i,j=1,\ldots,n$, 
    \begin{equation*}
       \begin{aligned}
\frac{d \lambda_i}{dt}=&\frac1m \lambda_i \ffat_{ii}^0\,,\quad \\
\frac{d\aaa_i^{0\gamma}}{dt}=&-\frac12 \sum_{k\neq i} 
(\aaa_i^{0\gamma}-\aaa_k^{0\gamma}) \tilde V^0_{ik} \lambda_k \,, \\ 
\frac{d \ccc_j^{0\epsilon}}{dt}=&-\frac1m \ccc_j^{0\epsilon} \ffat_{jj}^0 + \frac12 \sum_{k\neq j} 
(\ccc_j^{0\epsilon} \tilde V_{jk}^0-\ccc_k^{0\epsilon} \tilde V_{kj}^0) \lambda_k \,.
       \end{aligned}
  \end{equation*}
  For any $q\in \Z_m \setminus \{0\}$, $1\leq \gamma,\epsilon \leq d_q$ and $i,j=1,\ldots,n$, 
    \begin{equation*}
       \begin{aligned}
\frac{d\aaa_i^{q\gamma}}{dt}=&\frac{m-q}{m}\aaa_i^{q\gamma} \ffat_{ii}^0 + \aaa_i^{q\gamma} (Z_{m-1})_{ii}\lambda_i 
-\frac12 \sum_{k\neq i} \left(\aaa_i^{q\gamma} \tilde V^0_{ik} \lambda_k -\aaa_k^{q\gamma} 
\frac{2\lambda_i^{m-q} \lambda_k^{q+1}}{\lambda_i^m-\lambda_k^m} \gfat_{ik}^0 \right )  \,, \\ 
\frac{d \ccc_j^{q\epsilon}}{dt}=&\frac{q-1-m}{m} \ccc_j^{q\epsilon} \ffat_{jj}^0  
- \ccc_j^{q\epsilon} (Z_{m-1})_{jj}\lambda_j
 + \frac12 \sum_{k\neq j} \left( \ccc_j^{q\epsilon} \tilde V_{jk}^0\lambda_k 
-\ccc_k^{q\epsilon} \frac{2\lambda_k^{m-q+1}\lambda_j^{q}}{\lambda_k^m-\lambda_j^m} \gfat_{kj}^0 \right)  \,.
       \end{aligned}
  \end{equation*}
\end{lem}
\begin{proof}
The first bracket is easy to obtain using \eqref{Eq:PBlambg} and $\ffat^0_{ii}=\gfat^0_{ii}\lambda_i$. For the next two we need 
\begin{equation*}
   \begin{aligned}
\br{\aaa_i^{0\gamma},\gfat_{kk}^0}_{\loc}=& \frac12 \delta_{(i\neq k)} (\aaa_k^{0\gamma}-\aaa_i^{0\gamma}) 
\left[\frac{\lambda_i^m+\lambda_k^m}{\lambda_i^m-\lambda_k^m} \gfat_{ik}^0 - \gfat_{ik}^0 - 2 (Z_{m-1})_{ik} \right] \,, \\
   \br{\ccc_j^{0\epsilon},\gfat_{kk}^0}_{\loc}=& \frac12 \delta_{(j\neq k)} \ccc_j^{0\epsilon} 
 \left[\frac{\lambda_j^m+\lambda_k^m}{\lambda_j^m-\lambda_k^m} \gfat_{jk}^0 - \gfat_{jk}^0 - 2 (Z_{m-1})_{jk} \right] \\
 &- \frac12 \delta_{(j\neq k)} \ccc_k^{0\epsilon} 
 \left[\frac{\lambda_k^m+\lambda_j^m}{\lambda_k^m-\lambda_j^m} \gfat_{kj}^0 - \gfat_{kj}^0 - 2 (Z_{m-1})_{kj} \right]\,. 
   \end{aligned}
\end{equation*}
These identities are direct consequences of  Lemma \ref{ggMed00}. Similarly, the last two brackets are obtained using 
\begin{equation*}
   \begin{aligned}
\br{\aaa_i^{q\gamma},\gfat_{kk}^0}_{\loc}=& \delta_{ik} \frac{m-q}{m} \aaa_i^{q\gamma} \gfat_{kk}^0 + 
\delta_{ik} \aaa_i^{q\gamma}(Z_{m-1})_{ik}  
+\delta_{(i\neq k)} \aaa_k^{q\gamma} \gfat_{ik}^0 \frac{\lambda_i^{m-q}\lambda_k^q}{\lambda_i^m-\lambda_k^m} \\
&- \frac12 \delta_{(i\neq k)} \aaa_i^{q\gamma} 
\left[\frac{\lambda_i^m+\lambda_k^m}{\lambda_i^m-\lambda_k^m} \gfat_{ik}^0 - \gfat_{ik}^0 - 2 (Z_{m-1})_{ik} \right] \,,\\
\br{\ccc_j^{q\epsilon},\gfat_{kk}^0}_{\loc}=& \delta_{jk} \frac{q-m}{m} \ccc_j^{q\epsilon} \gfat_{kk}^0 
-\delta_{jk}\ccc_j^{q\epsilon} (Z_{m-1})_{kk}  
+ \delta_{(j\neq k)} \ccc_k^{q\epsilon} \frac{\lambda_k^{m-q} \lambda_j^q}{\lambda_k^m-\lambda_j^m} \gfat_{kj}^0\\
&+\frac12 \delta_{(j\neq k)} \ccc_j^{q\epsilon} 
 \left[\frac{\lambda_j^m+\lambda_k^m}{\lambda_j^m-\lambda_k^m} \gfat_{jk}^0 - \gfat_{jk}^0 - 2 (Z_{m-1})_{jk} \right]\,, 
   \end{aligned}
\end{equation*}
which can be derived in the same way as Lemma \ref{ggMed00}.
\end{proof}
To see how Lemma \ref{BrF0coord} generalises the equations of motion of the spin trigonometric Ruijsenaars-Schneider system \cite{KrZ}, we let for distinct $1\leq i,j\leq n$, 
\begin{equation}
  V^0_{ij}=\frac{\lambda_i^m+\lambda_j^m}{\lambda_i^m-\lambda_j^m}- 
\frac{\lambda_i^m+t \lambda_j^m}{\lambda_i^m-t \lambda_j^m}\,, \quad 
U^{0,s}_{ij}=\frac{-2t}{t_{s-1}} \frac{\lambda_i^s \lambda_j^{m-s}}{\lambda_i^m-t\lambda_j^m}\,,\,\, s\in\Z_m\setminus\{0\}\,.
\end{equation}
Using \eqref{hatZij}, we can then express \eqref{EqPotVr} as  $\tilde V_{ij}^0=V_{ij}^0 \gfat_{ij}^0 + \sum_{s\neq 0} U_{ij}^{0,s} \gfat_{ij}^s$.
If we also set  
\begin{equation*}
\overline{\ccc}_j^{0\epsilon}:=\ccc_j^{0\epsilon}\lambda_j\,, \quad \ffat_{ij}^0:=\gfat_{ij}^0\lambda_j= \sum_{\alpha=1}^{d_0} \aaa_i^{0\alpha}\overline{\ccc}_j^{0\alpha}\,, \quad 
1\leq i,j\leq n,\,\, 1\leq \epsilon \leq d_0,
\end{equation*}
we can write the first group of equations from Lemma \ref{BrF0coord} as $\frac{d \lambda_i^m}{dt}=\lambda_i^m \ffat_{ii}^0$ together with 
    \begin{equation*}
       \begin{aligned}
\frac{d\aaa_i^{0\gamma}}{dt}=&-\frac12 \sum_{k\neq i} (\aaa_i^{0\gamma}-\aaa_k^{0\gamma}) V^0_{ik} \ffat_{ik}^0 
-\frac12 \sum_{k\neq i} \sum_{s\neq 0}(\aaa_i^{0\gamma}-\aaa_k^{0\gamma}) U_{ik}^{0,s} \ffat_{ik}^s \,, \\ 
\frac{d \overline{\ccc}_j^{0\epsilon}}{dt}=&
 \frac12 \sum_{k\neq j} (\overline{\ccc}_j^{0\epsilon} V_{jk}^0 \ffat_{jk}^0-\overline{\ccc}_k^{0\epsilon} V_{kj}^0 \ffat_{kj}^0) 
+ \frac12 \sum_{k\neq j} \sum_{s\neq 0} (\overline{\ccc}_j^{0\epsilon} U_{jk}^{0,s} \ffat_{jk}^s 
-\overline{\ccc}_k^{0\epsilon} U_{kj}^{0,s} \ffat_{kj}^s)  \,.
       \end{aligned}
  \end{equation*}
Note from Lemma \ref{BrF0coord} that the vector field $d/dt$ restricts to the closed subvariety of $\Cnm'$ where 
$$\aaa_i^{s\alpha}=0,\,\, \ccc_i^{s\alpha}=0, 
\quad 1\leq i\leq n,\quad \text{spin indices }(s,\alpha)\text{ with }s\neq 0\,.$$ 
This closed subvariety is the intersection $\Cnm^{\res,d_0}\cap \Cnm'$ with $\Cnm^{\res,d_0}$ defined in Remark \ref{Rem:OpenChain} (which is the full variety $\Cnm'$ in the case $\dfat=(d_0,0,\ldots,0)$). 
There, we have $\ffat^s_{jk}=0$ whenever $s\neq0$. Furthermore, if we set $t=e^{-2\gamma}$ and introduce $(q_i)_{i=1}^n$ through $\lambda_i^m=e^{2 q_i}$ locally, we can write that 
\begin{equation}
  V^0_{ij}=\coth(q_i-q_j)-\coth(q_i-q_j+\gamma)\,, \quad i\neq j\,,
\end{equation}
while the vector field becomes completely determined on this subvariety by 
    \begin{equation}
       \begin{aligned}
\frac{dq_i}{dt}=&\frac12 \ffat_{ii}^0\,, \quad
\frac{d\aaa_i^{0\gamma}}{dt}=-\frac12 \sum_{k\neq i} (\aaa_i^{0\gamma}-\aaa_k^{0\gamma}) V^0_{ik} \ffat_{ik}^0 \,, \quad
\frac{d \overline{\ccc}_j^{0\epsilon}}{dt}=
 \frac12 \sum_{k\neq j} (\overline{\ccc}_j^{0\epsilon} V_{jk}^0 \ffat_{jk}^0-\overline{\ccc}_k^{0\epsilon} V_{kj}^0 \ffat_{kj}^0) \,.
       \end{aligned}
  \end{equation}
As observed in \cite{CF2,F1}, these are the equations of motion for the spin trigonometric RS system from \cite{KrZ} written in a fixed gauge by Arutyunov and Frolov \cite{AF}.   
It is also interesting to note that we can get in this way the completed phase space of the  spin trigonometric RS system as the closed subvariety  $\Cnm^{\res,d_0}\subset \Cnm^\circ$ defined by the condition that $V_{s,\alpha}=0_{1\times n}$, $W_{s,\alpha}=0_{n\times 1}$ for all $s\neq 0$ and $1\leq \alpha \leq d_s$.  
Indeed, we have on $\Cnm^{\res,d_0}\cap \Cnm'$ that the vector field $d/dt=\sum_i\br{-,\ffat_{ii}^0}_{\loc}$ coincides with the Hamiltonian vector field associated with $q_0(1-t)^{-1} \tr(T_0)$ in view of \eqref{hatTij}. At the same time, using quasi-Hamiltonian reduction, $\Cnm^{\res,d_0}\subset \Cnm^\circ$ is the subset corresponding to $\MM_{Q_{d_0}^{\res},\nfat}^{\circ}\subset \MM_{Q_{\dfat},\nfat}^{\circ}$ as defined in Remark \ref{Rem:OpenChain}. 
We can directly see from Proposition \ref{Pr:floTcy} that the flow of the vector field $\br{-,\tr(T_0)}$ can be restricted to $\MM_{Q_{d_0}^{\res},\nfat}^{\circ}$ where it is complete, and this last property is inherited by the projected flow on  $\Cnm^{\res,d_0}$. 

\begin{rem}
As a consequence of Propositions \ref{Pr:floYcy} and \ref{Pr:floZcy}, 
the flows associated with the Hamiltonian vector fields defined by the functions $(\tr (Y^k))$ or $(\tr (Z^k))$, $k\in \N$, can be restricted in the exact same way to complete flows on any closed subvariety $\Cnm^{\res,b}$ for $b\in \{1,\ldots,|\dfat|\}$. When restricted to $\Cnm^{\res,1}$, we recover the results from \cite{CF1}. Local expressions for such Hamiltonian functions can be found in \cite{CF1} written in terms of $2n$ log-canonical coordinates. 
\end{rem}

\subsubsection{The $m=2$ case}

Consider the special case $m=2$ for the multiplicative quiver variety $\mathcal{C}_{n,2,\qfat,\dfat}$. We assume that $\dfat=(d_0,d_1)$ for $d_0,d_1\geq 1$, and we note that the regularity condition from Proposition \ref{Pr:CyMQVter} amounts to require 
\begin{equation*}
 t^k\neq 1\,\, \text{ for all } k\in \Z^\times,\qquad \qquad t^l\neq q_0\,\, \text{ for all }  l\in \Z,
\end{equation*}
for the parameter $\qfat=(q_0,q_1)\in \CC^\times \times \CC^\times$, with $t=q_0q_1$. 

We obtain from Proposition \ref{Pr:DgIScy} that the invariant functions forming $\HHl$ \eqref{Eq:Palg-HL} (with $\rL=X,Y$ or $Z$) are pairwise Poisson commuting on  $\mathcal{C}_{n,2,\qfat,\dfat}$, and we derived their explicit flows in \ref{sss:PoiDyn-flow} before reduction. Let us write some of these functions on $\mathcal{C}_{n,2,\qfat,\dfat}'$ thanks to the $2n(d_0+d_1)+n$ local coordinates  
$$(\lambda_i, \aaa_i^{s,\alpha},\ccc_i^{s,\alpha})\,, \qquad 1\leq i\leq n,\quad s=0,1,\quad 1\leq \alpha \leq d_s\,,$$
under the $n$ constraints $\sum_{\alpha=1}^{d_0} \aaa_i^{0,\alpha}=1$. 
The eight matrices $(X_s,Z_s,\As^s,\Cs^s)_{s=0,1}$,  which determine a point of $\Cnqm'$  can be written down in terms of the local coordinates using Proposition \ref{Pr:diffeo}.  We have in particular that $X_0=X_1=\diag(\lambda_1,\ldots,\lambda_n)$ while by \eqref{hatZij} we get for any $1\leq i,j\leq n$,  
\begin{equation*}
\begin{aligned}
  (Z_0)_{ij}=&q_0 \frac{\lambda_i}{\lambda_i^2-t \lambda_j^2}\ffat_{ij}^0 
+ t \frac{\lambda_j}{\lambda_i^2-t \lambda_j^2} \ffat_{ij}^1 \,, \quad
(Z_1)_{ij}=t \frac{\lambda_j}{\lambda_i^2-t \lambda_j^2} \ffat_{ij}^0 
+ \frac{t}{q_0} \frac{\lambda_i}{\lambda_i^2-t \lambda_j^2}\ffat_{ij}^1\,.
\end{aligned}
\end{equation*}
Here, we have set $\ffat_{ij}^s=\gfat_{ij}^s \lambda_j= \sum_{\alpha=1}^{d_s} \aaa_i^{s,\alpha}\ccc_j^{s,\alpha}\lambda_j$ for $s=0,1$. 
It is obvious to see that $\tr (X^{2k})=2\tr((X_0X_1)^k)=2\sum_{i=1}^n \lambda_i^{2k}$, and similar easy expressions in fact hold for any $m\geq 2$.  

To write the first function of each remaining family $(\tr(Z^k))$ or $(\tr(Y^k))$, we fix a square root $\sqrt{t}$ of $t\neq0$ so that   
\begin{equation*}
  \frac{\lambda_i}{\lambda_i^2-t \lambda_j^2}=\frac12 \frac{1}{\lambda_i-\sqrt{t} \lambda_j} + \frac12 \frac{1}{\lambda_i+\sqrt{t} \lambda_j}\,, \quad 
\frac{\lambda_j}{\lambda_i^2-t \lambda_j^2}=\frac{1}{2\sqrt{t}} \frac{1}{\lambda_i-\sqrt{t} \lambda_j} - \frac{1}{2\sqrt{t}} \frac{1}{\lambda_i+\sqrt{t} \lambda_j}\,.
\end{equation*}
We will use these expressions for the entries of $Z_0,Z_1$. The nonzero regular function $\tr(Z^{2k})=2\tr((Z_1Z_0)^k)$ can be written when $k=1$ as 
\begin{equation} \label{Eq:m2-trZ2}
  \tr(Z^2)=\frac{\sqrt{t}}{2 q_0} \sum_{i,j=1}^n 
\left[\frac{q_0 \ffat_{ij}^0 + \sqrt{t}\ffat_{ij}^1}{\lambda_i- \sqrt{t}\lambda_j} 
+ \frac{q_0 \ffat_{ij}^0 - \sqrt{t}\ffat_{ij}^1}{\lambda_i+ \sqrt{t}\lambda_j} \right] 
\left[\frac{q_0 \ffat_{ji}^0 + \sqrt{t}\ffat_{ji}^1}{\lambda_j- \sqrt{t}\lambda_i} 
- \frac{q_0 \ffat_{ji}^0 - \sqrt{t}\ffat_{ji}^1}{\lambda_j+ \sqrt{t}\lambda_i} \right]\,.
\end{equation} 
(For each summand, the first factor is $2(Z_0)_{ij}$ and the second is $\frac{2q_0}{\sqrt{t}}(Z_1)_{ji}$.)  
To write the function $\tr(Y^k)=2\tr((Y_1Y_0)^k)$ with $k=1$, we recall that $Y_s=Z_s-X_s^{-1}$ and obtain  
\begin{equation*}
  \tr(Y^2)=\tr(Z^2)-\frac{2}{1-t}\sum_{i=1}^n \left[ (q_0+t)\frac{\ffat_{ii}^0}{\lambda_i^2} 
+ (q_1+t) \frac{\ffat_{ii}^1}{\lambda_i^2} \right] + \sum_{i=1}^n \frac{2}{\lambda_i^2}\,.
\end{equation*}
We can see the $W$-invariance of both functions from these expressions. Indeed, the $S_n$-invariance is direct. As discussed before Proposition \ref{Pr:diffeo}, an element $(k_1,\ldots,k_n)\in \Z^n_2$ acts by 
$\lambda_i \mapsto (-1)^{k_i}\lambda_i$, $\ffat_{ij}^0 \mapsto \ffat_{ij}^0$,  and $\ffat_{ij}^1 \mapsto (-1)^{k_j-k_i} \ffat_{ij}^1$. Thus, each factor in the right-hand side of \eqref{Eq:m2-trZ2} is multiplied by $(-1)^{k_i}$ for such a transformation, so $\tr(Z^2)$ is $\Z^n_2$-invariant. 
The invariance of $\tr(Y^2)$ directly follows. 

The above results also hold in the case $d_1=0$, where we omit the terms containing $\ffat_{ij}^1$ from the expressions. We recover in that case the Hamiltonians written down in \cite[\S5.6.3]{F1} up to a multiplicative factor. Furthermore, the case $(d_0,d_1)=(1,0)$ can be related to the functions $G_{2,1}$ and $H_{2,1}$ written in \cite[\S4.2]{CF1}, where a local set of Darboux coordinates can be explicitly constructed.

\section{Integrability}  \label{S:Int} 

As in Section \ref{S:Loc}, we work with the dimension vector $(1,n\delta)$ and the regularity condition from Propositions \ref{Pr:CyMQVbis}-\ref{Pr:CyMQVter}. We denote  the multiplicative quiver variety $\Cnqm$ and its open subset $\Cnqm^\circ$ by $\Cnm$ and $\Cnm^\circ$, respectively.  
Recall that these are smooth affine varieties of dimension $2n|\dfat|$ endowed with a non-degenerate Poisson bracket, which are constructed by reduction from the quasi-Hamiltonian varieties $\MM_{Q_{\dfat},\nfat}^{\bullet}$ and  $\MM_{Q_{\dfat},\nfat}^{\circ}$. 
Hereafter, we consider these smooth varieties as complex manifolds, and we work with the corresponding topology.

Let us fix some terminology used in this section. Consider $M$ to be a complex Poisson manifold of dimension $2k$ with non-degenerate Poisson bracket. Fix a Poisson subalgebra $\mathcal{A}$ of the algebra of holomorphic functions on $M$. 
We say that $\mathcal{A}$ is a degenerately integrable system of rank $1\leq r\leq k$ if $\mathcal{A}$ has functional dimension $2k-r$ and its Poisson centre has functional dimension $r$. (Here, the functional dimension is the number of functionally independent elements at a generic point of $M$.) We say that $\mathcal{A}$ is a Liouville integrable system when $r=k$. 
In this section, we will show that the Poisson algebras constructed in Section \ref{S:Subalg} using invariant functions on the quasi-Hamiltonian varieties $\MM_{Q_{\dfat},\nfat}^{\bullet}$ and  $\MM_{Q_{\dfat},\nfat}^{\circ}$ descend to degenerately/Liouville integrable systems on the corresponding multiplicative quiver varieties. We assume from now on that $|\dfat|\geq 2$, as the results for $|\dfat|=1$ can be found in \cite{CF1}.


\subsection{Degenerately integrable systems}  \label{ss:DIS}

Consider the Poisson algebras $\IIl$ and $\II_+$ introduced in \ref{ss:PoiDyn}. These are generated by invariant functions, hence these Poisson algebras descend to \emph{reduced} Poisson subalgebras of the holomorphic functions on $\Cnm$ (or $\Cnm^\circ$ for $\rL=Z$) which we denote in the same way. Furthermore, the abelian Poisson algebras $\HHl$ and $\HH_+$ that are contained in their Poisson centre are also made of invariant functions, and so they descend to $\Cnm$ (or $\Cnm^\circ$) as well, where they are still included inside the Poisson centre of the reduced Poisson algebras $\IIl$ and $\II_+$. Our goal is to build on these observations to prove that the reduced Poisson algebras $\IIl$ and $\II_+$ are degenerately integrable systems.  
Note that many instances of the spin RS systems are known to be degenerately integrable \cite{AR,CF2,FFM,Re,Re2}, so the degenerate integrability of $\II_+$ is natural since an element from its centre induces (locally) equations of motion that generalise those for the spin trigonometric RS system, see \ref{sss:SpinRS}.

\subsubsection{Local coordinates in the case of $\IIl$ with $\rL=Z$} \label{sss:loc-Z}

We start with the construction of a suitable set of local coordinates in which the elements of $\II_Z$ take a simple form, following the method of \cite[\S6.1]{FFM}. We do so on the irreducible component containing $\Cnm'$ that we defined in Remark \ref{rem:Cnmp}. (Note that this irreducible component is closed in the Zariski topology, therefore it is also closed in the complex topology with which we work.) 
This irreducible component contains $\Cnm^{\circ,\res,1}\subset \Cnm^\circ$ which is the space corresponding to $\dfat=(1,0,\ldots,0)$ studied in \cite{CF1} and which is irreducible by \cite{CF1,Ob}. 
Let us remark the following characterisation of $\Cnm^{\circ,\res,1}$ around a point where $Z_{m-1}$ has generic eigenvalues.

\begin{lem}[\cite{CF1}] \label{Lem:Zcoord-Prel}
Fix $(z_i)_{i=1}^n\in \hreg$ as defined in \eqref{Eqhreg}. There exists an $n$-dimensional submanifold of $\Cnm^\circ$ consisting of points 
$$(X,Z,V_{s,\alpha},W_{s,\alpha})\in \Cnm^{\circ,\res,1}\subset \Cnm^\circ$$
such that 
$Z_{m-1}=\diag(z_1,\ldots,z_n)$, $Z_{s}=\Id_n$ if $s\neq m-1$, while $W_{0,1}=(1,\ldots,1)^T$ and furthermore $V_{s,\alpha}=0_{1\times n}$, $W_{s,\alpha}=0_{n\times 1}$ whenever $(s,\alpha)\neq (0,1)$.  
\end{lem}
\begin{proof}
We need to find $V_{0,1}$ and $X$ such that the $m$ moment map constraints \eqref{GeomCys-v2} are satisfied. If the elements $V_{s,\alpha},W_{s,\alpha}$ vanish whenever $(s,\alpha)\neq (0,1)$, the constraints \eqref{GeomCys-v2} can be written as 
\begin{equation*}
X_s Z_s X_{s-1}^{-1}=q_s Z_{s-1},\quad \text{ for }s\neq 0\,, \qquad X_0 Z_0 X_{m-1}^{-1}=q_0(\Id_n+W_{0,1}V_{0,1})Z_{m-1}\,.
\end{equation*}
With the stated matrices $(Z_s)_{s\in\Z_m}$, the identities with $s\neq 0$ allow to define $X_0,\ldots,X_{m-2}$ from $X_{m-1}$ and $Z_{m-1}$. Moreover, taking the product of all these identities yields 
\begin{equation} \label{Eq:Lem:Zcoord-1}
 X_{m-1}Z_{m-1}X_{m-1}^{-1}=t (\Id_n+W_{0,1}V_{0,1})Z_{m-1}\,.
\end{equation}
In other words, we seek a solution $V\in \Mat(1\times n,\CC)$ to the spectral problem 
\begin{equation}
 \det(Z_{m-1}-\lambda \Id_n)=\det(tZ_{m-1}+t WVZ_{m-1}-\lambda \Id_n)\,, \quad W=(1,\ldots,1)^T\,.
\end{equation}
For the given $Z_{m-1}$, the left-hand side is $\prod_{k=1}^n (z_k-\lambda)$, while for the right-hand side we can use the Sylvester determinant formula to get 
\begin{equation}
\begin{aligned}
  \det(tZ_{m-1}+t WVZ_{m-1}-\lambda \Id_n)=&\det(tZ_{m-1}-\lambda \Id_n)\, (1+t VZ_{m-1}(tZ_{m-1}-\lambda \Id_n)^{-1}W) \\
  =&\prod_{1\leq k \leq n} (tz_k-\lambda) +\sum_{j=1}^n tV_jz_j \prod_{k\neq j} (tz_k-\lambda)\,.
\end{aligned}
\end{equation}
Evaluating both sides at $\lambda=tz_j$, we find that $V_j=(1-t)t^{-n}\prod_{k\neq j}\frac{z_k-tz_j}{z_k-z_j}$ and $V$ has no zero entry.  
The matrix $X_{m-1}$ must solve the diagonalisation problem \eqref{Eq:Lem:Zcoord-1}. Thus given a fixed solution $X_{m-1}^{\ast}$ of this problem, we can take any $X_{m-1}=X_{m-1}^\ast \diag(a_1,\ldots,a_n)$ with $a_i\in \CC^\times$ for $1\leq i \leq n$.
\end{proof}

Note that the space of points constructed in Lemma \ref{Lem:Zcoord-Prel} is unchanged if we permute the entries of the fixed $(z_1,\ldots,z_n)\in\hreg$. Indeed, this permutation corresponds to the action of $S_n$ by simultaneous permutation of the entries of an element $(X,Z,V_{s,\alpha},W_{s,\alpha})\in \Cnm^\circ$. 

We build on the idea of the proof of Lemma \ref{Lem:Zcoord-Prel} to get local coordinates. 
We start by introducing an open submanifold $\mathfrak{C}_0\subset \Cnm^\circ$ as follows:  inside the irreducible component containing $\Cnm'$, $\mathfrak{C}_0$ is the open submanifold defined by requiring that each point $(X,Z,V_{s,\alpha},W_{s,\alpha})$ admits a representative such that 
\begin{equation} \label{Eq:Rep-Z}
Z_{m-1}=\diag(z_1,\ldots,z_n),\,\, (z_i)_{i=1}^n\in\hreg;\quad Z_s=\Id_n \text{ for }s\neq m-1\,; \quad W_{0,1}=(1,\ldots,1)^T\,.
\end{equation}
By Lemma \ref{Lem:Zcoord-Prel}, $\mathfrak{C}_0$ is not empty. Given a point of $\mathfrak{C}_0$ (or in any subset introduced below), we will assume that we work with a representative in the form \eqref{Eq:Rep-Z}. This fixes the gauge up to a residual $S_n$ action by simultaneous permutation of all the matrix entries.  
For our next step, we introduce 
\begin{equation} \label{Eq:tilF}
 \tilde{F}=t 
\left( \prod^{\longleftarrow}_{s>0} \left( \prod^{\longleftarrow}_{1\leq \alpha \leq d_{s}}(\Id_n+W_{s,\alpha}V_{s,\alpha}) \right) Z_{s-1}\right)
\left( \prod^{\longleftarrow}_{1< \alpha \leq d_{0}}(\Id_n+W_{0,\alpha}V_{0,\alpha}) \right) Z_{m-1}
\,.
\end{equation}
It is important to note that the factor $(\Id_n+W_{0,1}V_{0,1})$ is omitted.  
We then define the submanifold 
\begin{equation}
 \mathfrak{C}_1:=\{(X,Z,V_{s,\alpha},W_{s,\alpha})\in\mathfrak{C}_0 \mid  \text{the spectrum of }\tilde{F} \text{ takes value in }\hreg\}\,.
\end{equation}
Note that $\mathfrak{C}_1$ is not empty: in the case considered in Lemma \ref{Lem:Zcoord-Prel}, $\tilde{F}=tZ_{m-1}$ which has its spectrum inside $\hreg$ since $Z_{m-1}$ does. 
Finally, if $\tilde{g}\in \Gl(n)$ is a matrix diagonalising $\tilde{F}$, i.e. 
\begin{equation} \label{Eq:DiagF}
 \tilde{g}\tilde{F}\tilde{g}^{-1}=\diag(\mu_1,\ldots,\mu_n)\,, \quad (\mu_i)_{i=1}^n\in \hreg\,,
\end{equation}
we let $\mathfrak{C}_2\subset \mathfrak{C}_1$ be the submanifold defined by the condition that $\sum_{j=1}^n\tilde{g}_{ij}z_j^{-1}\neq 0$ for all $1\leq i \leq n$. Note that this condition is independent of the choice of diagonalising matrix $\tilde{g}$ as in \eqref{Eq:DiagF}.
We have that $\mathfrak{C}_2$ is not empty since in the case of Lemma \ref{Lem:Zcoord-Prel} we see that $\tilde{F}=tZ_{m-1}$ is already diagonal, so we can pick $\tilde{g}=\Id_n$ and the condition reads $z_i^{-1}\neq 0$ for all $1\leq i\leq n$, which holds since $(z_i)_{i=1}^n\in \hreg$.  

\begin{lem} \label{Lem:Zcoord-Main}
Given $(X,Z,V_{s,\alpha},W_{s,\alpha})\in \mathfrak{C}_2$, pick a representative such that \eqref{Eq:Rep-Z} holds. 
Let $\tilde{g}\in  \Gl(n)$ be a diagonalising matrix for $\tilde{F}$ as in \eqref{Eq:DiagF}. Then 
\begin{equation} \label{Eq:ZMain-1}
 V_{0,1}=\tilde{V}\tilde{g}Z_{m-1}^{-1}\,, \quad \text{ for } 
\tilde{V}_j=(z_j-\mu_j)\,\mu_j^{-1}\, \Big(\sum_{1\leq \ell\leq n} \tilde{g}_{j\ell}z_\ell^{-1}\Big)^{-1} \prod_{k\neq j} \frac{z_k-\mu_j}{\mu_k-\mu_j}\,,\quad  1\leq j \leq n,
\end{equation}
while the matrices constituting $X$ satisfy 
\begin{align}
X_{m-1}=& X_{m-1}^\ast \diag(\zeta_1,\ldots,\zeta_n)\,, \quad \text{ for }(\zeta_j)_{j=1}^n\in (\CC^\times)^n \,,    \label{Eq:ZMain-2}\\
X_s=&\frac{t_s}{t} \left(\prod_{s<r\leq m-1}^{\longrightarrow} \,\prod_{1\leq \alpha \leq d_{r}}^{\longrightarrow} 
(\Id_n+W_{r,\alpha}V_{r,\alpha})^{-1}\right)\, X_{m-1}Z_{m-1}\,, \quad 0\leq s < m-1\,,   \label{Eq:ZMain-3}
\end{align}
where $X_{m-1}^\ast$ is a fixed solution of the diagonalisation problem 
\begin{equation} \label{Eq:ZMain-4}
 X_{m-1}^\ast Z_{m-1} (X_{m-1}^\ast)^{-1} = \tilde{F} + \tilde{F} Z_{m-1}^{-1} W_{0,1}V_{0,1} Z_{m-1}\,.
\end{equation}
\end{lem}
\begin{proof}
In order to establish \eqref{Eq:ZMain-1}--\eqref{Eq:ZMain-4}, we study the moment map constraints \eqref{GeomCys-v2} which we rewrite as  
\begin{equation} \label{Eq:Zmain-pf1}
 X_s Z_s X_{s-1}^{-1}=q_s \prod_{1\leq \alpha \leq d_s}^{\longleftarrow} (\Id_n+W_{s,\alpha}V_{s,\alpha}) Z_{s-1}\,, \quad 0\leq s \leq m-1\,.
\end{equation}
These equations allow us to write for $s\neq m-1$ that 
\begin{equation}
X_s= q_{s+1}^{-1}Z_s^{-1} \prod_{1\leq \alpha \leq d_{s+1}}^{\longrightarrow} (\Id_n+W_{s+1,\alpha}V_{s+1,\alpha})^{-1}\,  X_{s+1}Z_{s+1}\,.
\end{equation}
Thus, \eqref{Eq:ZMain-3} can be checked by induction for $Z$ given as in \eqref{Eq:Rep-Z} if we also use that $q_{s+1}^{-1}\ldots q_{m-1}^{-1}=t_s t^{-1}$ by \eqref{Eq:tparam}. Taking the product of the identities in \eqref{Eq:Zmain-pf1}, we can write using \eqref{Eq:tilF} that 
\begin{align*}
 X_{m-1}Z_{m-1}\ldots Z_0 X_{m-1}^{-1}&= t \prod_{0\leq s \leq m-1}^{\longleftarrow}\left( \prod_{1\leq \alpha \leq d_s}^{\longleftarrow} (\Id_n+W_{s,\alpha}V_{s,\alpha}) \right) Z_{s-1} \\
 &= \tilde{F} Z_{m-1}^{-1}(\Id_n+W_{0,1}V_{0,1})Z_{m-1}\,.
\end{align*}
For $Z$ given as in \eqref{Eq:Rep-Z}, this can be put in the form 
\begin{equation} \label{Eq:ZMain-pf2}
 X_{m-1} Z_{m-1} X_{m-1}^{-1} = \tilde{F} + \tilde{F} Z_{m-1}^{-1} W_{0,1}V_{0,1} Z_{m-1}\,,
\end{equation}
hence the matrix on the right-hand side has the same spectrum as $Z_{m-1}$. This is equivalent to the spectral identity 
\begin{equation*}
 \prod_{1\leq k \leq n} (z_k-\lambda)=\diag(Z_{m-1}-\lambda \Id_n)=\diag\left((\tilde{F}-\lambda \Id_n)+\tilde{F}Z_{m-1}^{-1}W_{0,1}V_{0,1}Z_{m-1} \right)\,.
\end{equation*}
Using Sylvester determinant formula to evaluate the right-hand side, it equals 
\begin{align*}
& \diag(\tilde{F}-\lambda \Id_n)\, \left[1+V_{0,1}Z_{m-1}(\tilde{F}-\lambda \Id_n)^{-1} \tilde{F} Z_{m-1}^{-1} W_{0,1} \right] \\ 
=& \diag(\tilde{F}-\lambda \Id_n)\, \left[1+V_{0,1}Z_{m-1}\tilde{g}^{-1}(\diag(\mu_1,\ldots,\mu_n)-\lambda \Id_n)^{-1} \tilde{g}\tilde{F} Z_{m-1}^{-1} W_{0,1} \right] \\ 
=& \prod_{1\leq k \leq n}(\mu_k-\lambda)\, \left[ 1 + \sum_{1\leq j \leq n} \tilde{V}_j (\mu_j-\lambda)^{-1}\tilde{W}_j\right]\,.
\end{align*}
Here we used \eqref{Eq:DiagF}, and we have set 
\begin{equation}  \label{Eq:ZMain-pf3}
 \tilde{V}=V_{0,1}Z_{m-1}\tilde{g}^{-1}\,, \quad \tilde{W}=\tilde{g}\tilde{F} Z_{m-1} W_{0,1}=\diag(\mu_1,\ldots,\mu_n)\tilde{g}Z_{m-1}^{-1}W_{0,1}\,.
\end{equation}
By assumption, we work in $\mathfrak{C}_2$ so that $(\tilde{g}Z_{m-1}^{-1}W_{0,1})_i=\sum_{j}\tilde{g}_{ij}z_j^{-1}$ is nonzero and thus $\tilde{W}_i\neq 0$ for all $i$. Returning to the spectral identity, we have obtained 
\begin{equation*}
\prod_{1\leq k \leq n}(z_k-\lambda)=\prod_{1\leq k \leq n}(\mu_k-\lambda) + \sum_{1\leq j \leq n} \tilde{V}_j \tilde{W}_j \prod_{k\neq j}(\mu_k-\lambda)\,,
\end{equation*}
which gives upon setting $\lambda=\mu_j$ that 
\begin{equation} \label{Eq:ZMain-pf4}
 \tilde{V}_j\tilde{W}_j=(z_j-\mu_j)\prod_{k\neq j}\frac{z_k-\mu_j}{\mu_k-\mu_j}\,, \quad 1\leq j \leq n\,.
\end{equation}
As $\tilde{W}_j\neq 0$ for each $1\leq j \leq n$, we can determine $V_{0,1}$ using \eqref{Eq:ZMain-pf3}--\eqref{Eq:ZMain-pf4} and it has the claimed form \eqref{Eq:ZMain-1}. 
To conclude, we see from \eqref{Eq:ZMain-pf2} that $X_{m-1}$ is a solution to the diagonalisation problem \eqref{Eq:ZMain-4}. It is unique up to right multiplication by a diagonal matrix since $Z_{m-1}$ has simple spectrum, hence \eqref{Eq:ZMain-2} holds. 
\end{proof}

A converse to the statement of Lemma \ref{Lem:Zcoord-Main} holds. Consider the set of all the values
\begin{equation}
 (z_j)_{j=1}^n\in \hreg\,, \quad \{ (W_{s,\alpha})_j,\,(V_{s,\alpha})_j\mid 1\leq j \leq n,\,\, (s,\alpha)\neq(0,1)\}\,, 
\end{equation}
that correspond to the points in $\mathfrak{C}_2$. This is a submanifold of $\CC^{2n|\dfat|-n}$, denoted by $\mathfrak{D}$. 
Note that, at the moment, we obtain $n!$ points from a given one in $\mathfrak{C}_2$ due to the residual $S_n$ action. 
If we consider the manifold $\mathfrak{D}/S_n$ obtained from the orbits of the free $S_n$ action 
\begin{equation*}
 \sigma\cdot (z_j,(W_{s,\alpha})_j,\,(V_{s,\alpha})_j) = (z_{\sigma(j)},(W_{s,\alpha})_{\sigma(j)},\,(V_{s,\alpha})_{\sigma(j)})\,,
\end{equation*}
we get that the map $\mathfrak{C}_2 \to \mathfrak{D}/S_n$ is well-defined since we removed the residual $S_n$ action on $\mathfrak{C}_2$. 
Now, given a point in $\mathfrak{D}/S_n$, we can define elements $W_{s,\alpha},V_{s,\alpha}$ for any $(s,\alpha)\neq (0,1)$ in the obvious way, and we can introduce as well $Z$ and $W_{0,1}$ through \eqref{Eq:Rep-Z}. 
We can then define $\tilde{F}$ as in \eqref{Eq:tilF}, which will have spectrum in $\hreg$ by construction. Then, if we fix a matrix $\tilde{g}$ so that \eqref{Eq:DiagF} holds and a solution $X_{m-1}^\ast$ to  \eqref{Eq:ZMain-4}, we can define $V_{0,1}$ through \eqref{Eq:ZMain-1} and define $X$ using \eqref{Eq:ZMain-2}--\eqref{Eq:ZMain-3}. We end up with a point inside $\mathfrak{C}_2$, which is unique up to the $n$-dimensional choice of $X_{m-1}$ made in the last step through \eqref{Eq:ZMain-2}. 
This last choice depends on the $S_n$ action on $\mathfrak{D}$ through permutation of the values $(\zeta_i)_{i=1}^n$ given in \eqref{Eq:ZMain-2}. Note that we can pick $X_{m-1}^\ast$ to depend algebraically on the coordinates of $\mathfrak{D}/S_n$ since it is obtained by construction of eigenvalues and eigenvectors of a matrix with simple spectrum. 
To summarise, we have obtained a bijection between $\mathfrak{C}_2$ and $(\mathfrak{D}\times (\CC^\times)^n)/S_n$. Furthermore, both maps are expressed in terms of basic algebraic (hence analytic) operations. 

\begin{cor} \label{Cor:CoordZ}
 The manifold $(\mathfrak{D}\times (\CC^\times)^n)/S_n$ with coordinates 
 \begin{equation}
  z_j, \,\, (W_{s,\alpha})_j,\,(V_{s,\alpha})_j, \,\, \zeta_j\,, \quad 1\leq j \leq n,\,\, (s,\alpha)\neq(0,1)\,,
 \end{equation}
provides a set of local coordinates on $\mathfrak{C}_2$ through \eqref{Eq:Rep-Z} and the formulas of Lemma \ref{Lem:Zcoord-Main}. 
\end{cor}

\subsubsection{Degenerate integrability of $\IIl$ for $\rL=Z$} \label{sss:DIS-Z}

We use Corollary \ref{Cor:CoordZ} to write down the holomorphic functions generating $\HH_Z$ and $\II_Z$ \eqref{Eq:Palg-HL} in $\mathfrak{C}_2$. In particular, we obtain that for any $k \geq 1$, 
\begin{equation}
  \tr(Z^{km})=m\tr(Z_{m-1}\ldots Z_0)^k=m\sum_{1\leq j \leq n} z_j^{k}\,,
\end{equation}
and for any $k\geq 1$ and $(p,\gamma),(q,\epsilon)\neq(0,1)$, 
\begin{align}
 t^k_{p\gamma;q\epsilon}:=&\tr(W_{p,\gamma} V_{q,\epsilon} Z^{km+q-p}) \nonumber \\
 =&\tr(W_{p,\gamma}V_{q,\epsilon} Z_{q-1}\ldots Z_0 (Z_{m-1}\ldots Z_0)^{k-1} Z_{m-1}\ldots Z_p) \label{Eq:tk-1}\\
 =&\sum_{1\leq j \leq n} z_j^k  (W_{p,\gamma})_j (V_{q,\epsilon})_j \,, \nonumber \\
 t^k_{q\epsilon}:=&\tr(W_{0,1}V_{q,\epsilon} Z^{km+q} )=\sum_{1\leq j \leq n} z_j^k  (V_{q,\epsilon})_j\,. \label{Eq:tk-2}
\end{align}

\begin{thm} \label{Thm:DIS-Z}
 The algebra $\II_Z$ is a degenerately integrable system of rank $n$ on the irreducible component of $\Cnm^\circ$ containing $\Cnm'$. 
\end{thm}
\begin{proof}
 The functions $\frac{1}{km}\tr(Z^{km})$, $k=1,\ldots,n$, belong to $\HH_Z$ which sits inside the Poisson centre of $\II_Z$ by Proposition \ref{Pr:DgIScy}. We easily see that the Jacobian matrix of these $n$ functions taken with respect to the coordinates $(z_j)_{j=1}^n$ is the Vandermonde matrix
\begin{equation}
 \mathrm{V}_z\in \Mat(n\times n,\CC)\,, \quad (\mathrm{V}_z)_{kj}=z_j^{k-1}\,.
\end{equation}
This matrix is invertible since the $(z_j)_{j=1}^n$ are pairwise distinct by definition of $\hreg$. Thus the Poisson centre of $\II_Z$ has dimension at least $n$. 

For fixed $(q,\epsilon)\neq(0,1)$ and $k\in \N$, the function $t^k_{q\epsilon}$ \eqref{Eq:tk-2} belongs to $\II_Z$. Taking these $n$ functions for $k=1,\ldots,n$, their Jacobian matrix taken with respect to the coordinates $((V_{q,\epsilon})_j)_{j=1}^n$ is given by $\mathrm{V}_z$ due to \eqref{Eq:tk-2}. This is an invertible matrix, as we already noticed.

Let us fix $(q',\epsilon')\neq (0,1)$. For any $(p,\gamma)\neq(0,1)$, the Jacobian matrix of the functions 
$t^k_{p\gamma;q'\epsilon'}$ \eqref{Eq:tk-1} for $k=1,\ldots,n$, taken with respect to the coordinates $((W_{p,\gamma})_j)_{j=1}^n$ is given by 
$\mathrm{V}_z D'$ where $D'=\diag((V_{q',\epsilon'})_1,\ldots,(V_{q',\epsilon'})_n)$. 
Note that since the $((V_{q',\epsilon'})_j)_{j=1}^n$ are part of the local coordinates  by Corollary \ref{Cor:CoordZ}, we have at a generic point of $\mathfrak{C}_2$ that these $n$ functions are nonzero hence $D'$ is invertible.

We can now prove that $\II_Z$ has functional dimension at least $2n|\dfat|-n$. We fix $(q',\epsilon')\neq (0,1)$, and consider the functions 
\begin{equation}
 T:=\left( \frac{1}{km}\tr(Z^{km}); t^k_{p\gamma} ; t^k_{p\gamma;q'\epsilon'}\right)_{(p,\gamma)\neq (0,1)}^{1\leq k \leq n}\,.
\end{equation}
(We have $n$, $n(|\dfat|-1)$ and $n(|\dfat|-1)$ functions of each type, respectively.) 
Then, due to our previous observations, the Jacobian matrix $J_T$ of $T$ with respect to the coordinates $(z_j,(V_{s,\alpha})_j, (W_{s,\alpha})_j)$ given in Corollary \ref{Cor:CoordZ} has the form 
\begin{equation*}
 J_T=\left( 
\begin{array}{ccc}
\mathrm{V}_z&0&0 \\
\ast& A&0 \\
\ast&\ast &B
\end{array}
 \right)
\end{equation*}
for the block diagonal matrices 
\begin{align*}
 A&=\diag(\mathrm{V}_z,\ldots,\mathrm{V}_z)\in \Mat(n(|\dfat|-1)\times n(|\dfat|-1),\CC)\,, \\
B&=\diag(\mathrm{V}_zD',\ldots,\mathrm{V}_zD')\in \Mat(n(|\dfat|-1)\times n(|\dfat|-1),\CC)\,.
\end{align*}
In particular, $\mathrm{V}_z$ (hence $A$) is always invertible in $\mathfrak{C}_2$, and $B$ is invertible at a generic point where the $((V_{q',\epsilon'})_j)_{j=1}^n$ are nonzero. Thus $J_T$ has maximal rank $2n|\dfat|-n$ at a generic point, which proves our claim for the dimension of $\II_Z$. 

So far, we have obtained that 
\begin{equation} \label{Eq:dimCountZ}
 2n|\dfat|-n\leq \dim(\II_Z)\,,\quad n\leq \dim(\text{Poisson centre of }\II_Z)=:r \,,
\end{equation}
where $\dim$ denotes the functional dimension. As noted right after Proposition \ref{Pr:CyMQV}, the Poisson structure on $\Cnm^\circ$ is non-degenerate. 
Thus, in the neighbourhood of a generic point of $\mathfrak{C}_2$, we can find $r$ independent functions $(h_a)_{a=1}^r$ in the Poisson centre of $\II_Z$, as well as $r$ functions $(g_a)_{a=1}^r$ such that $\br{h_a,g_b}=\delta_{ab}$ for $a,b=1,\ldots,r$. 
In other words, there exists $r$ independent functions that are canonically conjugate to $r$ independent functions from the centre of $\II_Z$.  In particular, the $(g_a)_{a=1}^r$ can not belong to $\II_Z$ since they do not commute with all the elements in the Poisson centre of $\II_Z$. Therefore, we can pick a set of functionally independent functions given by  the $r$ functions $(g_a)_{a=1}^r$ and $\dim(\II_Z)$ functions from $\II_Z$, which yields 
\begin{equation}
 r + \dim(\II_Z)\leq \dim(\CC[\Cnm^\circ])= 2n|\dfat|\,.
\end{equation}
Combining this inequality with \eqref{Eq:dimCountZ}, we deduce that $r=n$ and $\dim(\II_Z)=2n|\dfat|-n$. By definition, this means that $\II_Z$ is a degenerately integrable system.  
\end{proof}

\subsubsection{Degenerate integrability of $\IIl$ for $\rL=X,Y$}  \label{sss:DIS-Y}

The statement of Theorem \ref{Thm:DIS-Z} can be obtained in the exact same way for $\rL=X$ or $\rL=Y$. In the case $\rL=Y$, it suffices to work in the submanifold of $\Cnm$ where $Y\in \End(\VV_\cyc)$ is invertible. We can rewrite \ref{sss:loc-Z} to get local coordinates and \ref{sss:DIS-Z} to compute the dimension in that case by replacing the matrices $X_s$ and $Z_s$ by $(X_s+Y_{s}^{-1})$ and $Y_s$, respectively. Indeed, all the constructions are based on the moment map constraint written in the form \eqref{Eq:Zmain-pf1}, and in the case $\rL=Y$ it suffices to replace it with 
\begin{equation} 
 (X_s+Y_s^{-1}) Y_s (X_{s-1}+Y_{s-1}^{-1})^{-1}=q_s \prod_{1\leq \alpha \leq d_s}^{\longleftarrow} (\Id_n+W_{s,\alpha}V_{s,\alpha}) Y_{s-1}\,, \quad s\in \Z_m\,.
\end{equation}

The case $\rL=X$ requires some adjustments. We need to replace the matrices $X_s$ and $Z_s$ by $Z_{s}$ and $X_{s-1}^{-1}$ respectively, then use the moment map constraint 
\begin{equation}  \label{Eq:Xmain-pf1}
 Z_{s}X_{s-1}^{-1}Z_{s-1}^{-1}=q_s X_s \prod_{1\leq \alpha \leq d_s}^{\longleftarrow} (\Id_n+W_{s,\alpha}V_{s,\alpha})\,, \quad s\in \Z_m\,.
\end{equation}
The key difference when compared to \eqref{Eq:Zmain-pf1} is that $X_s$ multiplies the product on the right-hand side of \eqref{Eq:Xmain-pf1} from the left.  
Since the functions $\tr(X^{km})\in \HH_X$, $k\in \N$, take a trivial form in terms of the coordinates from Proposition \ref{Pr:diffeo}, we do not spend time outlining the other differences needed to have an analogue of Corollary \ref{Cor:CoordZ} in that case.

\subsubsection{Local coordinates in the case of $\II_+$} \label{sss:loc-T}

We adapt the construction of \ref{sss:loc-Z} to obtain a set of local coordinates in which the elements of $\II_+$ take a simple form. We also work in the irreducible component containing $\Cnm'$ (and $\Cnm^{\circ,\res,1}$) that we defined in Remark \ref{rem:Cnmp}. Let us introduce the subset $J_\dfat$ of cardinality $m_\dfat$ that counts the number of nonzero entries in the vector $\dfat\in \N^{\Z_m}$, i.e. 
\begin{equation}
 J_\dfat:=\{s\in \Z_m \mid d_s\neq 0\}\,, \quad m_\dfat:=|J_\dfat|\,.
\end{equation}
By assumption, $0\in J_\dfat$ so $1\leq m_\dfat\leq m$. 
To ease notation, we let $T_s:=\Id_{\VV_s}+X_sY_s=\Id_n+X_sY_s$ for $s\in \Z_m$ and $T:=\sum_{s\in \Z_m} T_s\in \End(\VV_\cyc)$. 
Note that a point $(X,Z,V_{s,\alpha},W_{s,\alpha})\in \Cnm^\circ$ is equivalently parametrised by the tuple $(X,T,V_{s,\alpha},W_{s,\alpha})$ since $Z=X^{-1}T$. Then,  the $m$ moment map constraints \eqref{GeomCys-v2} can be written as 
\begin{equation} \label{Eq:T-Momap}
T_s X_{s-1}^{-1} T_{s-1}^{-1} X_{s-1}=q_s \prod_{1\leq \alpha \leq d_s}^{\longleftarrow} (\Id_n+W_{s,\alpha}V_{s,\alpha})\,, \quad s\in \Z_m\,.
\end{equation} 

\begin{lem} \label{Lem:Tcoord-Prel}
(1) Fix $(z_i)_{i=1}^n\in \hreg$ as defined in \eqref{Eqhreg}. There exists an $(n m_\dfat)$-dimensional submanifold of $\Cnm^\circ$ consisting of points 
$(X,T,V_{s,\alpha},W_{s,\alpha})$
such that 
$T_{s}=\frac{t_s}{t}\diag(z_1^{-1},\ldots,z_n^{-1})$ for all $s\in \Z_m$, while 
\begin{align*}
 X_{s-1}&=\Id_n& \text{ if }s\notin J_\dfat\,, \\
 W_{s,1}&=(1,\ldots,1)^T,\,\, V_{s,1}=0_{1\times n},& \text{ if }s\in J_\dfat\,, 
\end{align*}
and furthermore  $V_{s,\alpha}=0_{1\times n}$, $W_{s,\alpha}=0_{n\times 1}$ otherwise (i.e. if $\alpha\neq 1$ when $d_s>1$).

(2) Such points belong to the irreducible component of $\Cnm^\circ$ containing $\Cnm'$.
\end{lem}

\begin{proof}
(1) Note that the choice of matrices fixes the representative up to a residual $S_n$ action. 
By choice of the elements $V_{s,\alpha},W_{s,\alpha}$, the $m$ moment map constraints \eqref{Eq:T-Momap} can be written as 
\begin{equation*} 
X_{s-1}^{-1} T_{s-1}^{-1} X_{s-1}=q_s T_s^{-1}, \,\,\, \text{for } s\neq 0, \quad 
X_{m-1}^{-1} T_{m-1}^{-1} X_{m-1}=q_0 T_0^{-1}(\Id_n+W_{0,1}V_{0,1})\,.
\end{equation*} 
With the stated $T_s$ for $s\in\Z_m$ and $X_s$ for $s\notin J_\dfat$, these identities for $s\notin J_\dfat$ are trivially satisfied. For each $s\in J_\dfat\setminus\{0\}$ we obtain  that $X_{s-1}$ is a diagonal matrix. The remaining identity seen as a spectral problem allows to determine $V_{0,1}$ similarly to the proof of Lemma \ref{Lem:Zcoord-Prel}, and it defines $X_{m-1}$ up to left multiplication by a diagonal matrix. 
Therefore, we have a $n$-dimensional set of possible  $X_{s-1}$ for each $s\in J_\dfat$, yielding a $nm_\dfat$-dimensional space of points in the stated form. 

(2) There is nothing to prove if $m_\dfat=1$. Indeed, in that case $J_\dfat=\{0\}$ so we have a point in $\Cnm^{\circ,\res,1}\subset \Cnm'$. 

If $m_\dfat>1$, there exists $\check{s}\in \{1,\ldots,m-1\}$ which is the minimal nonzero index $s$ with $d_s\geq 1$. 
Fix a point $p=(X,T,V_{s,\alpha},W_{s,\alpha})$ in the stated form, as constructed in part (1). For any $\epsilon\in \CC^\times$, we define 
$p^{(\epsilon)}=(X^{(\epsilon)},T^{(\epsilon)},V_{s,\alpha}^{(\epsilon)},W_{s,\alpha}^{(\epsilon)})\in \Cnm^\circ$ obtained from $p$ by only changing the following two matrices 
\begin{align*}
 X_{\check{s}-1}^{(\epsilon)}&=\diag(\epsilon a_1,\ldots,\epsilon a_n)\,, \quad \text{ if }X_{\check{s}-1}=\diag(a_1,\ldots,a_n)\,, \\
 X_{m-1}^{(\epsilon)}&=\diag(\epsilon^{-1},\ldots,\epsilon^{-1})X_{m-1}\,.
\end{align*}
(These choices are possible by part (1) because $0,\check{s}\in J_\dfat$.) Introduce $g^{(\epsilon)}\in \Gl(n\delta)=\prod_{s\in \Z_m}\Gl(n)$ by 
\begin{equation*}
 g_s^{(\epsilon)}=\Id_n \,\, \text{ if }s<\check{s}\,, \quad 
 g_s^{(\epsilon)}=\diag(\epsilon,\ldots,\epsilon) \,\, \text{ if }\check{s}\leq s \leq m-1\,.
\end{equation*}
Then, under the action of $g^{(\epsilon)}$ on $p^{(\epsilon)}$, the representative $g^{(\epsilon)}\cdot p^{(\epsilon)}\in \Cnm^\circ$ only differs from $p^{(\epsilon)}$ in the following matrices  
\begin{equation}
 X_{\check{s}-1}^{(\epsilon)}\mapsto X_{\check{s}-1}\,, \quad 
 X_{m-1}^{(\epsilon)}\mapsto X_{m-1}\,, \quad 
 W_{s,1}\mapsto (\epsilon,\ldots,\epsilon)^T\,\, \text{ for all }s\in J_\dfat\setminus\{0\}\,.
\end{equation}
In the limit $\epsilon\to 0$, $g^{(\epsilon)}\cdot p^{(\epsilon)}$ tends toward a well-defined point of $\Cnm^{\circ,\res,1}$ since all $V_{s,\alpha},W_{s,\alpha}$ are zero except for $(s,\alpha)=(0,1)$. (We can in fact put $X_{s-1}=\Id_n$ for all $s\in J_\dfat\setminus\{0\}$ using the $\Gl(n\delta)$ action on that point.) 
Since $p=p^{(1)}=g^{(1)}\cdot p^{(1)}$, we have a $1$-parameter family in $\Cnm^\circ$ passing through $p$ that admits a point of $\Cnm^{\circ,\res,1}$ in its closure. Thus $p$ is in the irreducible component containing  $\Cnm^{\circ,\res,1}$. 
\end{proof}

We now introduce three submanifolds $\mathfrak{R}_2\subset \mathfrak{R}_1 \subset \mathfrak{R}_0$ of the irreducible component containing $\Cnm'$; they play the role of $\mathfrak{C}_2\subset \mathfrak{C}_1 \subset \mathfrak{C}_0$ that we used in \ref{sss:loc-Z}. 
We start by introducing $\mathfrak{R}_0$ as the open submanifold inside the irreducible component containing $\Cnm'$  defined by the condition that each point $(X,T,V_{s,\alpha},W_{s,\alpha})$ admits a representative such that 
\begin{equation} \label{Eq:Rep-T}
\begin{aligned}
 T_{s}=\diag(z_{s,1}^{-1},\ldots,z_{s,n}^{-1}),\,\, (z_{s,i})_{i=1}^n\in\hreg, \quad W_{s,1}=(1,\ldots,1)^T, \qquad &\text{ if }s\in J_\dfat\,; \\
 T_{s}=q_sT_{s-1}\,, \quad X_{s-1}=\Id_n\,, \qquad &\text{ if }s\notin J_\dfat\,.
\end{aligned}
\end{equation}
By Lemma \ref{Lem:Tcoord-Prel}, $\mathfrak{R}_0$ is not empty. Note that in Lemma \ref{Lem:Tcoord-Prel} the diagonal entries of the matrices $(T_s)_{s\in \Z_m}$ are all related, while in $\mathfrak{R}_0$ we may assume more generally that the entries of $(T_s)_{s\in \Z_m}$ are independent whenever $s\in J_\dfat$. 
Given a point of $\mathfrak{R}_0$, we assume that we work with a representative in the form \eqref{Eq:Rep-T}, which fixes the gauge up to a residual $S_n$ action by simultaneous permutation.   
Next, we introduce for any $s\in J_\dfat$ 
\begin{equation} \label{Eq:hatF}
 \hat{F}_s=q_s T_{s}^{-1}  \prod^{\longleftarrow}_{1< \alpha \leq d_{s}}(\Id_n+W_{s,\alpha}V_{s,\alpha}) \,.
\end{equation}
(The factor $(\Id_n+W_{s,1}V_{s,1})$ is omitted, so $\hat{F}_s=q_s T_s^{-1}$ if $d_s=1$.)   
These matrices allow us to introduce 
\begin{equation}
 \mathfrak{R}_1:=\{(X,T,V_{s,\alpha},W_{s,\alpha})\in\mathfrak{R}_0 \mid  
 \text{the spectrum of }\hat{F}_s \text{ takes value in }\hreg \,\, \forall s\in J_\dfat\}\,.
\end{equation}
Again by Lemma \ref{Lem:Tcoord-Prel}, $\mathfrak{C}_1$ is not empty since there  $\hat{F}_s=tt_{s-1}^{-1} \diag(z_1,\ldots,z_n)$.  
Finally, for each $s\in J_\dfat$, consider a matrix $\hat{g}_s\in \Gl(n)$ diagonalising $\hat{F}_s$, i.e. 
\begin{equation} \label{Eq:DiagF-T}
 \hat{g}_s\hat{F}_s\hat{g}_s^{-1}=\diag(\mu_{s,1},\ldots,\mu_{s,n})\,, \quad (\mu_{s,i})_{i=1}^n \in \hreg\,.
\end{equation}
We define $\mathfrak{R}_2$ as the submanifold of $\mathfrak{R}_1$ satisfying the condition\footnote{Note that the matrices $(T_s)_{s\in\Z_m}$ are not involved in this condition, as opposed to the condition defining $\mathfrak{C}_2$ in \ref{sss:loc-Z}.} 
$\sum_{j=1}^n(\hat{g}_s)_{ij}\neq 0$ for all $1\leq i \leq n$. 
This subspace is not empty since in the case of Lemma \ref{Lem:Tcoord-Prel} we can pick $\hat{g}_s=\Id_n$ for each $s\in J_\dfat$.

In view of \eqref{Eq:Rep-T}, all the matrices $(T_s)_{s\in \Z_m}$ are diagonal in $\mathfrak{R}_2$ (when taking a representative in that form, which we always do). 
To ease notation, we also put $T_{s}=\diag(z_{s,1}^{-1},\ldots,z_{s,n}^{-1})$ when $s\notin J_{\dfat}$ after setting  $z_{s,j}:=q_s^{-1} z_{s-1,j}$ if $s\notin J_{\dfat}$ and $1\leq j \leq n$.   

\begin{lem} \label{Lem:Tcoord-Main}
Given $(X,T,V_{s,\alpha},W_{s,\alpha})\in \mathfrak{R}_2$, pick a representative such that \eqref{Eq:Rep-T} holds. 
Let $\hat{g}_s\in  \Gl(n)$ be a diagonalising matrix for $\hat{F}$ as in \eqref{Eq:DiagF-T}. Then for any $s\in J_\dfat$, 
\begin{equation} \label{Eq:TMain-1}
 V_{s,1}=\hat{V}_s\hat{g}_s \,, \quad \text{ for } 
(\hat{V}_s)_j=(z_{s-1,j}-\mu_{s,j})\,\mu_{s,j}^{-1}\, \Big(\sum_{1\leq \ell\leq n} (\hat{g}_s)_{j\ell}\Big)^{-1} 
\prod_{k\neq j} \frac{z_{s-1,k}-\mu_{s,j}}{\mu_{s,k}-\mu_{s,j}}\,,\quad 1\leq j \leq n\,,
\end{equation}
while the matrices constituting $X$ satisfy for each $s\in J_\dfat$, 
\begin{equation}
X_{s-1}= \diag(\zeta_{s,1},\ldots,\zeta_{s,n}) X_{s-1}^\ast\,, \quad \text{ for }(\zeta_{s,j})_{j=1}^n\in (\CC^\times)^n  \,,   \label{Eq:TMain-2}
\end{equation}
where $X_{s-1}^\ast$ is a fixed solution of the diagonalisation problem 
\begin{equation} \label{Eq:TMain-4}
 (X_{s-1}^\ast)^{-1} T_{s-1}^{-1} X_{s-1}^\ast = \hat{F}_s + \hat{F}_s W_{s,1}V_{s,1}\,.
\end{equation}
\end{lem}
\begin{proof}
 To obtain \eqref{Eq:TMain-1}--\eqref{Eq:TMain-4}, we rewrite the moment map condition \eqref{Eq:T-Momap} as 
 \begin{equation} \label{Eq:T-Momap-B}
 X_{s-1}^{-1} T_{s-1}^{-1} X_{s-1}=q_s T_s^{-1} \prod_{1\leq \alpha \leq d_s}^{\longleftarrow} (\Id_n+W_{s,\alpha}V_{s,\alpha})\,, \quad s\in \Z_m\,.
\end{equation} 
The condition \eqref{Eq:T-Momap-B} is trivially satisfied for $s\notin J_\dfat$ by \eqref{Eq:Rep-T} since we omit the product on the right hand-side. If $s\in J_\dfat$, we can write \eqref{Eq:T-Momap-B} in the form \eqref{Eq:TMain-4}, so that $X_{s-1}$ has the form \eqref{Eq:TMain-2} for fixed $X_{s-1}^\ast$. To check that $V_{s,1}$ has the desired form \eqref{Eq:TMain-1} when $s\in J_\dfat$, 
it suffices to analyse the spectral problem \eqref{Eq:TMain-4} as in Lemma \ref{Lem:Zcoord-Main}. 
Note that the form of $\hat{V}_s$ requires to work in $\mathfrak{R}_2$. 
\end{proof}

By an argument similar to the one that led to Corollary \ref{Cor:CoordZ}, we get the following result. 

\begin{cor} \label{Cor:CoordT}
There exists an open submanifold $\mathfrak{Q}$ of 
\begin{equation*}
 (\hreg)^{m_{\dfat}} \times \CC^{2n(|\dfat|-m_\dfat)} \times (\CC^\times)^{n m_\dfat}\,,
\end{equation*}
 with coordinates 
 \begin{equation} \label{Eq:CoordT}
  z_{s,j}, \,\, (W_{s,\alpha})_j,\,(V_{s,\alpha})_j, \,\,  \zeta_{s,j}\,, \qquad 
  1\leq j \leq n,\quad s\in J_\dfat, \quad  (s,\alpha)\neq(s,1), 
 \end{equation}
such that after identifying the elements under the natural $S_n$ action on $\mathfrak{Q}$ by permutation, 
the manifold $\mathfrak{Q}/S_n$ provides a set of local coordinates \eqref{Eq:CoordT} on $\mathfrak{R}_2$ through \eqref{Eq:Rep-T} and the formulas of Lemma \ref{Lem:Tcoord-Main}. 
\end{cor}

\subsubsection{Degenerate integrability of $\II_+$}

We can use Corollary \ref{Cor:CoordT} to write down the holomorphic functions generating $\HH_+$ and $\II_+$ \eqref{Eq:Palg-H+} in $\mathfrak{R}_2$. In particular, we obtain that for any $k \geq 1$ and $s\in J_\dfat$, 
\begin{equation}
  \tr((\Id_n+X_sY_s)^{-k})=\tr(T_s^{-k})=\sum_{1\leq j \leq n} z_{s,j}^{k}\,.
\end{equation}
As we noted right after \eqref{Eq:Palg-H+} these holomorphic functions belong to $\HH_+$ even though the exponents of $T_s$ are negative. In the same way, the following holomorphic functions belong to $\II_+$ for any $k\geq 1$ and $s\in J_\dfat$ with $(s,\gamma),(s,\epsilon)\neq(s,1)$, 
\begin{align}
 \hat{t}^k_{s;\gamma;\epsilon}:=&\tr(W_{s,\gamma}V_{s,\epsilon} (Id_n+X_sY_s)^{-k}) 
 =\sum_{1\leq j \leq n} z_{s,j}^k  (W_{s,\gamma})_j (V_{s,\epsilon})_j \,,  \label{Eq:T-tk-1}\\ 
 \hat{t}^k_{s;\epsilon}:=&\tr(W_{s,1}V_{s,\epsilon}(Id_n+X_sY_s)^{-k} )=\sum_{1\leq j \leq n} z_{s,j}^k  (V_{s,\epsilon})_j\,. \label{Eq:T-tk-2}
\end{align}

\begin{thm} \label{Thm:DIS-T}
 The algebra $\II_+$ is a degenerately integrable system of rank $n m_\dfat$ on the irreducible component of $\Cnm^\circ$ containing $\Cnm'$. 
\end{thm}
\begin{proof}
 The idea of the proof of Theorem \ref{Thm:DIS-Z} can be followed closely. 
 We first note that for any $s\in J_\dfat$, the $n$ functions $\frac1k \tr(T_s^{-k})$, $k=1,\ldots,n$, belong to $\HH_+$, hence to the Poisson centre of $\II_+$ by Proposition \ref{Pr:DgIScy-1XY}. Their Jacobian matrix taken with respect to the $n$ coordinates $(z_{s,j})_{j=1}^n$ is the Vandermonde matrix $\mathrm{V}_{z,s}\in \Mat(n\times n,\CC)$ with entries $(\mathrm{V}_{z,s})_{kj}=z_{s,j}^k$, and which is invertible since $(z_{s,j})_{j=1}^n\in \hreg$. This gives that the Poisson centre has dimension at least $n |J_\dfat|=nm_\dfat$.

Fix $s\in J_\dfat$ with $d_s>1$, and let $1<\epsilon \leq d_s$. Consider the functions $\hat{t}^k_{s;\epsilon}$ \eqref{Eq:T-tk-2} for $k=1,\ldots,n$. Their Jacobian matrix taken with respect to the coordinates $((V_{s,\epsilon})_j)_{j=1}^n$ is given by $\mathrm{V}_{z,s}$ due to \eqref{Eq:T-tk-2}. Thus, it is invertible. 

Fix $s\in J_\dfat$ with $d_s>1$ and some $1<\epsilon' \leq d_s$. For any $1<\gamma \leq d_s$, 
the Jacobian matrix of the functions $t^k_{s;\gamma;\epsilon'}$ \eqref{Eq:T-tk-1} for $k=1,\ldots,n$, 
taken with respect to the coordinates $((W_{s,\gamma})_j)_{j=1}^n$ is given by 
$\mathrm{V}_{z,s} D'_s$ where $D'_s=\diag((V_{s,\epsilon'})_1,\ldots,(V_{s,\epsilon'})_n)$. 
By Corollary \ref{Cor:CoordZ}, at a generic point of $\mathfrak{R}_2$ the $n$ functions $((V_{s,\epsilon'})_j)_{j=1}^n$ are nonzero, hence $D'_s$ is invertible. 

Putting these three observations together, we get functionally independent elements of $\II_+$. There are $nm_\dfat$ in the first case, and 
$n\sum_{s\in J_\dfat}(d_s-1)=n(|\dfat|-m_\dfat)$ functions in the remaining two cases. Thus 
\begin{equation} \label{Eq:dimCountT}
 2n|\dfat|-nm_\dfat\leq \dim(\II_+)\,,\quad nm_\dfat \leq \dim(\text{Poisson centre of }\II_Z) \,.
\end{equation}
Arguing as in the proof of Theorem \ref{Thm:DIS-Z} using that the manifold is symplectic, we have equalities in  \eqref{Eq:dimCountT}. 
We conclude by definition that $\II_+$ is a degenerately integrable system.  
\end{proof}


\subsection{Liouville integrable systems} \label{ss:LiouIS}

Throughout this subsection, we work with the notation of \ref{sss:Embed}. 
We consider the abelian Poisson algebras $\LLl$, $\LL_+$ and $\LL_+^{\cyc}$ constructed in \ref{ss:AbDyn}. These algebras are made of invariant functions, so they descend to reduced abelian Poisson algebras on $\Cnm$  (or $\Cnm^\circ$ for $\rL=Z$). We show that they are Liouville integrable systems after reduction. 
 Let us recall that these algebras are defined by taking trace of matrices containing the restricted moment maps  
$\Phi^{(b)}:=(\Phi^{(b)}_\infty , \Phi^{(b)}_s)$ which are obtained inductively using  \eqref{Eq:chainGeomMomap} for $b\in \{0,1,\ldots,|\dfat|\}$. In $\Cnm$, we can write for any $s\in\Z_m$ that 
\begin{equation} \label{Eq:Phi-res}
\begin{aligned}
  \Phi^{(b)}_s=&
(\Id_{n}+X_s Y_s)(\Id_{n}+Y_{s-1}X_{s-1})^{-1}
\prod^{\longrightarrow}_{\substack{1\leqslant \alpha \leqslant d_s \\ (s,\alpha)\leq \rho(b)}} (\Id_{n}+W_{s,\alpha}V_{s,\alpha}) \\
=&q_s \prod^{\longleftarrow}_{\substack{1\leqslant \alpha \leqslant d_s \\ (s,\alpha)> \rho(b)}} (\Id_{n}+W_{s,\alpha}V_{s,\alpha})\,,
\end{aligned}
\end{equation}
after using the moment map constraints \eqref{GeomCys}. In particular,  
\begin{equation} \label{Eq:Phi-resB}
 \Phi_s^{(|\dfat|)}=q_s\Id_n\,, \quad \text{ for all } s\in\Z_m\,.
\end{equation}

\subsubsection{Liouville integrability of $\LLl$} \label{sss:LiouIS-L}

We fix $\rL=Y$ or $\rL=Z$, and as in \eqref{Eq:LjEnd} we let $\rL^{(b)}:=\Phi^{(b)}\rL$ for $b\in \{0,1,\ldots,|\dfat|\}$. Recall that the Poisson algebras $\HHl$ \eqref{Eq:Palg-HL}, $\IIl$ \eqref{Eq:Palg-IL} and $\LLl$ \eqref{Eq:Palg-LL} descend to the multiplicative quiver variety.  

\begin{lem} \label{Lem:Liou1}
 We have $\HHl \subset \LLl \subset \IIl$ when considered on $\Cnm$ (or $\Cnm^\circ$ if $\rL=Z$).
\end{lem}
\begin{proof}
 For the first inclusion, note that by \eqref{Eq:LjEnd} and \eqref{Eq:Phi-resB}, 
we have $\rL_s=\frac{\Phi_{s}^{(|\dfat|)}}{q_{s}}\rL_{s}=q_s^{-1}\rL_s^{(|\dfat|)}$. 
Thus for any $k\geq 1$, the generator $\tr(\rL^k)$ of $\HHl$  is a multiple of $\tr((\rL^{(|\dfat|)})^k)\in \LLl$. 

For the second inclusion, note that by \eqref{Eq:LjEnd} and \eqref{Eq:Phi-res}, we have 
\begin{equation*}
 \rL_s^{(b)}=q_s  \Big(\prod^{\longleftarrow}_{\substack{1\leqslant \alpha \leqslant d_s \\ (s,\alpha)> \rho(b)}} (\Id_{n}+W_{s,\alpha}V_{s,\alpha})\Big)\,\, \rL_s\,.
\end{equation*}
Hence, for any $k\geq 1$ and $b\in \{0,1,\ldots, |\dfat|\}$,  the generator $\tr((\rL^{(b)})^k)$ of $\LLl$ can be written as $\tr(A)$, where $A$ is a sum of products of the matrices $\rL_s$ or $(\Id_{n}+W_{s,\alpha}V_{s,\alpha})$  (for all possible indices $(s,\alpha)>\rho(b)$). Using elementary transformations as in \eqref{Eq:TrickVW}, we can thus get that $\tr((\rL^{(b)})^k)$ is a polynomial in the functions $\tr(\rL^l)$ or $\tr(W_{s,\alpha}V_{r,\beta}\rL^l)$ with $(s,\alpha),(r,\beta)>\rho(b)$ and $0\leq l\leq k$, hence it belongs to $\IIl$. 
\end{proof}

We can now use the coordinates from Corollary \ref{Cor:CoordZ} to count independent elements in $\LLl$. 
(Recall that we noticed in \ref{sss:DIS-Y} that the construction of local coordinates can be applied verbatim to $\rL=Y$ instead of $\rL=Z$.) 
The idea of the proof is based on the one from \cite[Theorem 5.5]{CF2}.

\begin{thm} \label{Thm:LiouIS-L}
 The algebra $\LLl$ is a Liouville integrable system on the irreducible component of $\Cnm^\circ$ containing $\Cnm'$. 
\end{thm}
\begin{proof}
 The claim follows if we can show that there are $n|\dfat|$ functionally independent elements in $\LLl$ at a generic point of $\mathfrak{C}_2$. 
To do so, define $h_{b,k}=\tr( (\rL^{(b)})^{km})$ for any $b\in \{0,1,\ldots,|\dfat|\}$ and $k\geq 1$. 
The proof consists in showing by (descending) induction on $b=|\dfat|,\ldots,1$ that the functions $h_{|\dfat|,k},\ldots,h_{b,k}$, $k=1,\ldots,n$, are $n(|\dfat|-b+1)$ independent functions. 

For the base case, we remark that $h_{|\dfat|,k}$ is a multiple of $\tr(\rL^{km})$, see the proof of Lemma \ref{Lem:Liou1}. 
We get from the beginning of the proof of Theorem \ref{Thm:DIS-Z} that the functions  $\tr(\rL^{km})$ with $k=1,\ldots,n$ depend only on the coordinates $(z_j)_{j=1}^n$ and are functionally independent, so this yields the case $b=|\dfat|$. 

Assume that we have independence for some $b$. 
Using \eqref{Eq:Phi-res}, we can see that each function $h_{b-1,k}$ depends on all the $(V_{s,\alpha},W_{s,\alpha})$ with $(s,\alpha)\geq \rho(b)$ because the matrix $\rL^{(b-1)}$ does. In the same way, any $h_{c,k}$ with $b\leq c\leq |\dfat|$ depends on the $(V_{s,\alpha},W_{s,\alpha})$ with $(s,\alpha)>\rho(c)\geq \rho(b)$. 
Thus the elements $(h_{b-1,k})$ with $1\leq k \leq n$ depend on the $2n$ local coordinates 
$((V_{\rho(b)})_j,(W_{\rho(b)})_j)_{j=1}^n$ from Corollary \ref{Cor:CoordZ}, while any $h_{c,k}$ with $b\leq c\leq |\dfat|$ does not. To conclude, we will find a point where the Jacobian matrix of the $(h_{b-1,k})_{k=1}^n$ taken with respect to the coordinates $\mathrm{x}_b:=((W_{\rho(b)})_j)_{j=1}^n$ is invertible, hence has rank $n$.

We can write $h_{b-1,k}=m\tr(\rL_{m-1}^{(b-1)}\ldots \rL_0^{(b-1)})^k$ for $1\leq k \leq n$. 
By \eqref{Eq:LjEnd} and \eqref{Eq:Phi-res}, we have 
\begin{equation*}
 \rL_s^{(b-1)}=q_s 
 \Big(\prod^{\longleftarrow}_{\substack{1\leqslant \alpha \leqslant d_s \\ (s,\alpha)\geq \rho(b)}} (\Id_{n}+W_{s,\alpha}V_{s,\alpha})\Big) \,\, \rL_s  \,, \quad s\in \Z_m\,,
\end{equation*}
so that in the matrix product $\rL_{m-1}^{(b-1)}\ldots \rL_0^{(b-1)}$ the factor $(\Id_{n}+W_{\rho(b)}V_{\rho(b)})$ appears exactly once.
Arguing as in \eqref{Eq:TrickVW}, we can write that 
\begin{equation*}
 h_{b-1,k}= t^k \tr\left( ((\Id_{n}+W_{\rho(b)}V_{\rho(b)})\rL^m)^k \right) + Q_{b-1,k}\,,
\end{equation*}
where $Q_{b-1,k}$ is a sum of terms of the form 
\begin{equation*}
 \tr(W_{\rho(c_1)}V_{\rho(c_2)}\rL^{l_1m})\tr(W_{\rho(c_2)}V_{\rho(c_3)}\rL^{l_2m})\ldots \tr(W_{\rho(c_K)}V_{\rho(c_1)}\rL^{l_K m})\,,
\end{equation*}
where $l_1+\ldots +l_K=k$, and for $j=1,\ldots,K$ we have $b\leq c_j \leq |\dfat|$ with at least one index $c_j> b$.  In other words, $Q_{b-1,k}$ collect all the terms in the expansion of  $h_{b-1,k}$ of the form $\tr(A)$ where the matrix $A$ contains at least a factor $W_{s,\alpha}V_{s,\alpha}$ with $(s,\alpha)\neq \rho(b)$. (We have $Q_{|\dfat|-1,k}=0$ when $b=|\dfat|$.)  
By a similar expansion, we can write 
\begin{equation} \label{Eq:Pf-LiouL}
 h_{b-1,k}=t^k \tr(\rL^{km}) + km t^k  \tr\left(W_{\rho(b)}V_{\rho(b)}\rL^{km} \right) + P_{b-1,k} + Q_{b-1,k}\,,
\end{equation}
for some polynomial 
\begin{equation} \label{Eq:Pf-LiouL2}
 P_{b-1,k}\in \CC[\tr\left(W_{\rho(b)}V_{\rho(b)}\rL^{lm} \right)\mid 0\leq l < k]\,.
\end{equation}
Using \eqref{Eq:Pf-LiouL}, we see that the Jacobian matrix of the functions $(h_{b-1,k})_{k=1}^n$ taken with respect to the local coordinates $\mathrm{x}_b:=((W_{\rho(b)})_j)_{j=1}^n$ can be decomposed as  
\begin{equation}\label{Eq:Pf-LiouL3}
 J_b=m \frac{\partial\left(kt^k  \tr(W_{\rho(b)}V_{\rho(b)}\rL^{km})\right)_{k=1}^n}{\partial \mathrm{x}_b}
 + \frac{\partial\left(P_{b-1,k}\right)_{k=1}^n}{\partial \mathrm{x}_b} 
 + \frac{\partial\left(Q_{b-1,k}\right)_{k=1}^n}{\partial \mathrm{x}_b}\,.
\end{equation}
Note that at a point where $V_{\rho(c)}=0_{1\times n}$, $W_{\rho(c)}=0_{n\times 1}$ for all $c>b$, the last matrix must vanish due to the form of $Q_{b-1,k}$. Such a point exists in $\mathfrak{C}_2$: this can be checked by adapting Lemma \ref{Lem:Zcoord-Main} to the submanifold $\Cnm^{\circ,\res,b-1}\subset \Cnm^\circ$. 
Furthermore, since $P_{b-1,k}$ satisfies \eqref{Eq:Pf-LiouL2}, the rows of the second matrix in \eqref{Eq:Pf-LiouL3} are linear combinations of those of the first matrix at any point of $\mathfrak{C}_2$. Therefore, we are left to check that the Jacobian matrix 
\begin{equation}
 \frac{\partial\left(k t^k  \tr(W_{\rho(b)}V_{\rho(b)}\rL^{km})\right)_{k=1}^n}{\partial \mathrm{x}_b} 
 = \diag(t,\ldots, n t^n) \mathrm{V}_z \diag((V_{\rho(b)})_1,\ldots, (V_{\rho(b)})_n)
\end{equation}
has maximal rank at a generic point of $\mathfrak{C}_2$. Note that $\mathrm{V}_z$ is invertible as in the proof of Theorem \ref{Thm:DIS-Z}, and the second diagonal matrix is invertible at a point of $\mathfrak{C}_2$ where $(V_{\rho(b)})_j\neq 0$ for all $1\leq j \leq n$. The claim follows.  
\end{proof}

\begin{rem}
 We did not use the functions $h_{0,k}=\tr( (\rL^{(0)})^{km})$, $k\in \N$, as part of the proof of Theorem \ref{Thm:LiouIS-L}. The reason is that they are equivalent to the functions $h_{|\dfat|,k}$, $k\in \N$, which we have seen to be multiples of the functions $\tr(\rL^{km})$. Indeed, if $\rL=Y$ we have for $s\in \Z_m$ 
 \begin{equation*}
  Y^{(0)}_s= (\Id_n+X_{s+1}Y_{s+1}) (\Id_n+Y_{s}X_{s})^{-1}Y_s= (\Id_n+X_{s+1}Y_{s+1})Y_s (\Id_n+X_sY_{s})^{-1}\,,
 \end{equation*}
so that $Y^{(0)}=(1_\cyc+XY) Y (1_\cyc + XY)^{-1}$ seen as an endomorphism of $\VV_\cyc$. It follows that the symmetric functions of $Y^{(0)}$ and $Y$ are the same, in particular $h_{0,k}=\tr(Y^{km})$. This is also true for $\rL=Z$.  
\end{rem}

Using the dimensions computed in Theorems \ref{Thm:DIS-Z} (with its equivalent for $\rL=Y$) and \ref{Thm:LiouIS-L}, we arrive at the following statement. 

\begin{cor}
 If $|\dfat|>1$, the inclusions in Lemma \ref{Lem:Liou1} are proper. 
\end{cor}

\subsubsection{Liouville integrability of $\LL_+$} \label{sss:LiouIS+}

Recall that the Poisson algebras $\HH_+$ \eqref{Eq:Palg-H+}, $\II_+$ \eqref{Eq:Palg-I+}, $\LL_+^\cyc$ \eqref{Eq:Palg-cL+} and $\LL_+$ \eqref{Eq:Palg-L+} descend to the multiplicative quiver variety.  
To ease notation, we put $T_s=\Id_n+X_sY_s$ for each $s\in\Z_m$ as in \ref{sss:loc-T}, so that  $\LL_+^\cyc$ and $\LL_+$ are respectively generated by the functions $\tr((\Phi^{(b)}T^{-1})^k)$ with $k\in \N$, or $\tr((\Phi_s^{(b)}T_s^{-1})^k)$ with $s\in \Z_m$ and $k\in \N$.

\begin{lem} \label{Lem:Liou2}
 We have $\HH_+\subset \LL_+^{\cyc} \subset \LL_+ \subset \II_+$ when considered on $\Cnm$.
\end{lem}
\begin{proof}
The second inclusion is easy. Indeed, as noted right after \eqref{Eq:Palg-cL+}--\eqref{Eq:Palg-L+}, we have $\LL_+^{\cyc}\subset \LL_+$ on the quasi-Hamiltonian variety before reduction, so it is still the case on $\Cnm$ after reduction. 

 To show the third inclusion, we proceed as for the second inclusion in Lemma \ref{Lem:Liou1}. It suffices to remark that the identity 
\begin{equation*}
\Phi_s^{(b)}T_s^{-1}=
q_s  \Big(\prod^{\longleftarrow}_{\substack{1\leqslant \alpha \leqslant d_s \\ (s,\alpha)> \rho(b)}} (\Id_{n}+W_{s,\alpha}V_{s,\alpha})\Big)\,\, T^{-1}_s\,,
\end{equation*}
which follows from \eqref{Eq:Phi-res}, implies that any generator $\tr((\Phi_s^{(b)}T_s^{-1})^k)$ of $\LL_+$ can be written as a polynomial in $\tr(T_s^{l})$ and $\tr(W_{s,\alpha}V_{s,\alpha}T_s^{l})$ with $l \leq 0$. By the Cayley-Hamilton theorem, this can be written as a polynomial in such elements with $l\geq 0$, so that $\tr((\Phi_s^{(b)}T_s^{-1})^k)\in \II_+$. 

For the first inclusion, we need to prove that $\tr(T_s^k)\in \LL_+^{\cyc}$ for each $s\in \Z_m$ and $k\geq 0$ or, using the Cayley-Hamilton theorem that this is true for $k \leq 0$. Assume that for any $k \in \N$ and $s\in \Z_m$
\begin{equation} \label{Eq:IndIScy}
  q_s^{k}\tr (T_s^{-k})-\tr (T_{s-1}^{-k}) \in \LL_+^{\cyc}\,.
\end{equation}
Then taking linear combinations of \eqref{Eq:IndIScy} yields that $(t^{k}-1)\tr (T_{m-1}^{-k}) \in \LL_+^{\cyc}$ for each $k\in \N$. Since $t$ is not a root of unity by assumption (see Proposition \ref{Pr:CyMQVbis}), we get that $\tr (T_s^{-k})\in \LL_+^{\cyc}$ for $s=m-1$. Using \eqref{Eq:IndIScy} shows that this holds for all $s\in \Z_m$ as $q_s \neq 0$. 

We now show that \eqref{Eq:IndIScy} holds for each $s\in \Z_m$, which will allow us to conclude. We will repeatedly use the identity $\tr((\Id_{n_a}+AB)^k)=\tr((\Id_{n_b}+BA)^k)$ for any $A\in \Mat(n_a\times n_b,\CC)$ and $B\in\Mat(n_b\times n_a,\CC)$. 

First, note from the first equality in \eqref{Eq:Phi-res} that 
\begin{equation*}
\Phi_s^{(0)}T_s^{-1}=T_s (\Id_n+Y_{s-1}X_{s-1})^{-1}T_s^{-1}\,, \quad s\in \Z_m\,,
\end{equation*}
so that $\tr((\Phi^{(0)}T^{-1})^k)=\tr((1_\cyc+YX)^{-k})=\tr(T^{-k})$ for each $k\in \N$. 
Next, note that by assumption, $d_0\geq 1$, and $\rho(d_0)=(0,d_0)$. Hence, using \eqref{Eq:Phi-res} for $b=0,d_0$ and $s=0$, we find that $\Phi^{(d_0)}T^{-1}$ only differs from $\Phi^{(0)}T^{-1}$ in its $0$-th diagonal block, where we have 
\begin{equation*}
\Phi_0^{(d_0)}T_0^{-1}=q_0 T_0^{-1}\,, \quad \Phi_0^{(0)}T_0^{-1}=T_0 (\Id_n+Y_{m-1}X_{m-1})^{-1}T_0^{-1}\,.
\end{equation*}
These observations imply that for any $k \in \N$
\begin{equation*}
\tr ((\Phi^{(d_0)}T^{-1})^k)-\tr ((\Phi^{(0)}T^{-1})^k)=q_0^k\tr (T_0^{-k})-\tr (T_{m-1}^{-k})\,.
\end{equation*} 
 
Then, consider $0<s\leq m-1$. Note that if $d_s=0$ we have $(\Id_n+X_s Y_s)=q_s(\Id_n+Y_{s-1}X_{s-1})$ by \eqref{GeomCys}, so there is nothing to prove since \eqref{Eq:IndIScy} just vanishes. If $d_s\neq 0$, denote by $b_s \in \{0,1,\ldots,|\dfat|\}$ the element such that $\rho(b_s)=(s,d_s)$. Proceeding as in the case $s=0$, we can use the two equalities in  \eqref{Eq:Phi-res} to see that the matrices  
$\Phi^{(b_s-d_s)}T^{-1}$ and $\Phi^{(b_s)}T^{-1}$ only differ in their $s$-th diagonal block, where we have 
\begin{equation*}
\Phi_s^{(b_s)}T_s^{-1}=q_s T_s^{-1}\,, \quad \Phi_s^{(b_s-d_s)}T_s^{-1}=T_s (\Id_n+Y_{s-1}X_{s-1})^{-1}T_s^{-1}\,.
\end{equation*}
Therefore, 
\begin{equation*}
\tr ((\Phi^{(b_s)}T^{-1})^k)-\tr ((\Phi^{(b_s-d_s)}T^{-1})^k)=q_s^k\tr (T_s^{-k})-\tr (T_{s-1}^{-k})\,,  
\end{equation*}
 for any $k \in \N$, as desired. 
\end{proof}

A stronger statement can be obtained about the second inclusion from Lemma \ref{Lem:Liou2}. 

\begin{lem} \label{Lem:Liou2-b}
 We have  $\LL_+^{\cyc} =\LL_+$ when considered on $\Cnm$.
\end{lem}
\begin{proof}
We only need to prove that $\LL_+\subset \LL_+^{\cyc}$.  We will check by induction on $b\in\{0,1,\ldots,|\dfat|\}$ that 
the generators $\tr((\Phi_s^{(b)}T_s^{-1})^k)$ of $\LL_+$ belong to $\LL_+^{\cyc}$ for all indices $s\in \Z_m$ and $k\in \N$ (equivalently for any $k\in \Z$ by the Cayley-Hamilton theorem).

For the base step $b=0$, we get from the proof of Lemma \ref{Lem:Liou2} that $\tr((\Phi_s^{(0)}T_s^{-1})^k)=\tr(T_{s-1}^{-k})$. This element belongs to $\HH_+$, hence to $\LL_+^{\cyc}$ by Lemma \ref{Lem:Liou2}. 
 
Assume that $b\in \{1,\ldots,|\dfat|\}$. Then there exists a spin index $(r,\alpha)$ such that $\rho(b)=(r,\alpha)$. 
Using \eqref{Eq:Phi-res}, we see that the two matrices $\Phi^{(b-1)}T^{-1}$ and $\Phi^{(b)}T^{-1}$ only differ in their $r$-th diagonal block. Therefore we get that for any $k\in \N$, 
\begin{align*}
 \tr((\Phi^{(b)}_sT^{-1}_s)^k)=&\tr((\Phi^{(b-1)}_sT^{-1}_s)^k)\,, \quad \text{ for all }s\in\Z_m\setminus \{r\}\,, \\
 \tr((\Phi^{(b)}_rT^{-1}_r)^k)=&\tr((\Phi^{(b)}T^{-1})^k)-\tr(\Phi^{(b-1)}T^{-1})+ \tr((\Phi^{(b-1)}_rT^{-1}_r)^k)\,.
\end{align*}
The two right-hand sides belong to $\LL_+^{\cyc}$ by induction, so we are done. 
\end{proof}

We can now use the local coordinates on $\mathfrak{R}_2$ from Corollary \ref{Cor:CoordT} to count independent elements in $\LL_+$.
\begin{thm} \label{Thm:LiouIS-T}
 The algebra $\LL_+$ is a Liouville integrable system on the irreducible component of $\Cnm^\circ$ containing $\Cnm'$. 
\end{thm}
\begin{proof}
Due to the first two inclusions in Lemma \ref{Lem:Liou2}, we note that the function $\frac1k (T_s^{-k})$ belongs to $\LL_+$ for any $k\in \N$ and $s\in\Z_m$.  Using the coordinates from Corollary \ref{Cor:CoordT}, we can pick $nm_\dfat$ such independent functions as in Theorem \ref{Thm:DIS-T}. Note that these functions only depend on the coordinates $(z_{s,j})_{1\leq j \leq n}^{s\in J_\dfat}$. 

Next, define $\hat{h}_{b,k}=\tr( (\Phi_s^{(b)}T_s^{-1})^{k})$ for any $b\in \{0,1,\ldots,|\dfat|\}$ with $\rho(b)=(s,\alpha)$ and $k\geq 1$.  
Introduce  
$$K_\dfat=\{c\mid 1\leq c \leq |\dfat|,\,\, \beta\neq d_r \text{ if }\rho(c)=(r,\beta)\, \}\,,$$
i.e. $K_{\dfat}$ parametrises the spin indices $(r,\beta)$ with are distinct from those of the form $(s,d_s)$. 
The proof consists in showing by (descending) induction on $b\in K_{\dfat}$ that the functions $(\hat{h}_{c,k})$ with $k=1,\ldots,n$ and $c\in K_{\dfat}$ for $c\geq b$ are  independent functions depending on the coordinates $((V_{s,\alpha})_j,(W_{s,\alpha})_j)$ of Corollary \ref{Cor:CoordT}. 
We will end up with $n|K_{\dfat}|=n|\dfat|-nm_\dfat$ independent functions. 
They will then complement the $nm_{\dfat}$ elements from $\HH_+$ considered above, and when we reach the minimal $b\in K_{\dfat}$ we will obtain $n|\dfat|$ independent functions, as desired. 

We can in fact consider that the base case of the induction is given by obtaining the $nm_\dfat$ independent functions of $\HH_+$. So let us now pick  some $b\in K_{\dfat}$ with $\rho(b)=(s,\alpha)$, and assume that the statement holds by induction. 
We can then write from  \eqref{Eq:Phi-res} that 
\begin{equation*}
 \Phi_s^{(b)}T_s^{-1}=q_s 
 \Big(\prod^{\longleftarrow}_{\alpha< \beta \leqslant d_s} (\Id_{n}+W_{s,\beta}V_{s,\beta})\Big) T_s^{-1}\,.
\end{equation*}
Note that there is at least one such factor in the product over $\beta$ since $b\in K_\dfat$ is equivalent to  $1\leq \alpha< d_s$. 
Thus, in a similar way to the proof of Theorem \ref{Thm:LiouIS-L} we can get the expansion 
\begin{equation} \label{Eq:Pf-LiouT}
 \hat{h}_{b,k}=q_s^k \tr(T_s^{-k})+kq_s^k \tr(W_{s,\alpha+1}V_{s,\alpha+1} T_s^{-k})+ \hat{P}_{b,k}+\hat{Q}_{b,k}\,,
\end{equation}
for some polynomial 
\begin{equation} \label{Eq:Pf-LiouT2}
 \hat{P}_{b,k}\in \CC[\tr\left(W_{s,\alpha+1}V_{s,\alpha+1}T_s^{-l} \right) \mid 0\leq l < k]\,,
\end{equation}
and where $\hat{Q}_{b,k}$ is a sum of terms of the form 
\begin{equation*}
 \tr(W_{s,\beta_1}V_{s,\beta_2}T_s^{-l_1})\tr(W_{s,\beta_2}V_{s,\beta_3}T_s^{-l_2})\ldots \tr(W_{s,\beta_K}V_{s,\beta_1}T_s^{-l_K})\,,
\end{equation*}
where $l_1+\ldots +l_K=k$, and for $j=1,\ldots,K$ we have $\alpha<\beta_j \leq d_s$ with at least one index $\beta_j> \alpha+1$.  
(We have $\hat{Q}_{b,k}=0$ when $\rho(b)=(s,d_s-1)$.)

By an argument similar to the one in the proof of Theorem \ref{Thm:LiouIS-L}, we can conclude if we can show that the Jacobian matrix $\hat{J}_b$ of the $n$ functions $(\hat{h}_{b,k})_{k=1}^n$ taken with respect to the local coordinates $\mathrm{x}_b:=((W_{s,\alpha+1})_j)_{j=1}^n$ is an invertible matrix. Indeed, all the functions $(\hat{h}_{c,k})_{k=1}^n$ with $c\in K_\dfat$ for $c>b$ do not depend on the coordinates $\mathrm{x}_b$. 
Using \eqref{Eq:Pf-LiouT}, we can decompose the Jacobian matrix $\hat{J}_b$ as  
\begin{equation}\label{Eq:Pf-LiouT3}
 \hat{J}_b=\frac{\partial\left(kq_s^k \tr(W_{s,\alpha+1}V_{s,\alpha+1} T_s^{-k} )\right)_{k=1}^n}{\partial \mathrm{x}_b}
 + \frac{\partial\big(\hat{P}_{b,k}\big)_{k=1}^n}{\partial \mathrm{x}_b} 
 + \frac{\partial\big(\hat{Q}_{b,k}\big)_{k=1}^n}{\partial \mathrm{x}_b}\,.
\end{equation}
Note that at a point of $\mathfrak{R}_2$ where $V_{s,\beta}=0_{1\times n}$, $W_{s,\beta}=0_{n\times 1}$ for all $\alpha+1<\beta\leq d_s$, the last matrix must vanish due to the form of $\hat{Q}_{b,k}$. Such a point exists in $\mathfrak{R}_2$, as this can be checked by adapting Lemma \ref{Lem:Tcoord-Main}.  
Furthermore, due to the polynomial decomposition of $\hat{P}_{b,k}$ as in  \eqref{Eq:Pf-LiouT2}, it suffices to check that the Jacobian matrix 
\begin{equation}
 \frac{\partial\left(kq_s^k \tr(W_{s,\alpha+1}V_{s,\alpha+1} T_s^{-k})\right)_{k=1}^n}{\partial \mathrm{x}_b} 
 = \diag(q_s,\ldots, n q_s^n) \mathrm{V}_{z,s} \diag((V_{s,\alpha+1})_1,\ldots, (V_{s,\alpha+1})_n)
\end{equation}
has maximal rank at a generic point of $\mathfrak{R}_2$. We have the desired result since  $\mathrm{V}_{z,s}$ is invertible as in the proof of Theorem \ref{Thm:DIS-T}, and the second diagonal matrix is invertible at a point of $\mathfrak{R}_2$ where all $(V_{s,\alpha+1})_j\neq 0$.  
\end{proof}

\begin{cor}
 The first and third inclusions in Lemma \ref{Lem:Liou2} are proper. 
\end{cor}


\section{Conclusion}  \label{S:CCL}

In this work, we derived and studied integrable systems on certain complex Poisson manifolds, which we introduced as multiplicative quiver varieties obtained by reduction from (open) representation spaces $\MM_{Q_{\dfat},\nfat}^{\bullet}$ of extended cyclic quivers, cf. \ref{ss:CycMQV} for precise definitions.
Our first main task consisted in defining Poisson subalgebras inside the coordinate rings of the representation spaces which are
either ``large with a small Poisson centre'' as $\IIl$ and $\II_+$ in \ref{ss:PoiDyn}, or abelian as $\LLl$ and $\LL_+$ in \ref{ss:AbDyn}.
We could explicitly integrate the flows of the (quasi-Hamiltonian) equations of motion associated with any simple function from their respective Poisson centres, cf. \ref{sss:PoiDyn-flow} and \ref{sss:AbDyn-flow}.

The importance of the Poisson algebras studied in Section~\ref{S:Subalg} stems from the existence of local coordinates on the corresponding multiplicative quiver varieties $\Cnqm$ obtained by quasi-Hamiltonian reduction from $\MM_{Q_{\dfat},\nfat}^{\bullet}$.
Indeed, the equations of motion of a specific Hamiltonian function were written locally as a generalisation of the spin trigonometric RS system, cf. \ref{ss:Loc-express}.
Thus, our results provide new systems in the RS family with multiple types of spin variables, one type for each vertex in the cyclic quiver.
This is a far-reaching abstraction of the trigonometric system introduced by Krichever and Zabrodin \cite{KrZ}  whose geometric formulation was obtained in collaboration with Chalykh \cite{CF2}.

Finally, in Theorems \ref{Thm:DIS-Z} and \ref{Thm:DIS-T}, we proved that the large Poisson algebras of functions  with a small Poisson centre, $\IIl$ and $\II_+$, define degenerately integrable systems on $\Cnqm$ after quasi-Hamiltonian reduction.
We proved similarly that the abelian Poisson algebras $\LLl$ and $\LL_+$ descend to Liouville integrable systems, cf. Theorems \ref{Thm:LiouIS-L} and \ref{Thm:LiouIS-T}.
This completes the table presented in Figure \ref{Fig:Tab1} which associates CM/RS systems to extended cyclic quivers.

Currently, our proofs only allow to conclude that integrability holds on a connected component of the multiplicative quiver variety where the local coordinates of Section~\ref{S:Loc} can be constructed.
We conjecture that this connected component is, in fact, the whole variety in analogy with the original case of the non-spin RS system investigated by Oblomkov \cite{Ob}.

\medskip

Two directions of future research appear to be particularly stimulating.
First, quantisation of the present systems is a challenging task that may be tackled either with the approach of cyclotomic DAHA \cite{BEF}, or with quantised multiplicative quiver varieties \cite{Jo}.
Second, the recent geometric derivation of the trigonometric van Diejen system by Chalykh and Ryan \cite{CR} (based on its quantum analogue \cite{Cha}) gave an interpretation of its phase space as a multiplicative quiver variety. The construction of (multi-)spin generalisations of the van Diejen system may therefore be obtained by adapting our methods to the setting of \cite{CR}.


\end{document}